\documentclass[sn-mathphys-num]{sn-jnl}

\usepackage[bbgreekl]{mathbbol}
\usepackage{caption}
\usepackage{graphicx}%
\usepackage{multirow}%
\usepackage{amsmath,amssymb,amsfonts}%
\usepackage{amsthm}%
\usepackage{mathrsfs}%
\usepackage{bbold}
\usepackage[title]{appendix}%
\usepackage{xcolor}%
\usepackage{textcomp}%
\usepackage{manyfoot}%
\usepackage{booktabs}%
\usepackage{algorithm}%
\usepackage{algorithmicx}%
\usepackage{algpseudocode}%
\usepackage{listings}%
\usepackage[numbers]{natbib}
\usepackage{ifthen}

\usepackage[textsize=tiny]{todonotes}
\usepackage{comment}

\theoremstyle{thmstyleone}%
\newtheorem{theorem}{Theorem}
\newtheorem{proposition}[theorem]{Proposition}%
\newtheorem{lemma}[theorem]{Lemma}

\theoremstyle{thmstyletwo}%
\newtheorem{example}{Example}%
\newtheorem{remark}{Remark}%

\theoremstyle{thmstylethree}%
\newtheorem{assumption}{Assumption}%

\newcommand{\rhoapp}{\varrho_{{\rm app}}} 
\newcommand{\rhohat}{\widehat{\varrho}}
\newcommand{\piapp}{\pi_{\text{app}}} 
\newcommand{\rev}[1]{\textcolor{black}{#1}}

\newcommand{\change}[1]{\bgroup\color{black}#1\egroup}    

\definecolor{darkgreen}{rgb}{0.0, 0.2, 0.13}

\begin{document}

\title[Delayed Acceptance Slice Sampling]{Delayed Acceptance Slice Sampling}


\author*[1]{\fnm{Kevin} \sur{Bitterlich}}\email{Kevin.Bitterlich@math.tu-freiberg.de}

\author[2]{\fnm{Daniel} \sur{Rudolf}}\email{daniel.rudolf@uni-passau.de}

\author*[1]{\fnm{Bj\"{o}rn} \sur{Sprungk}}\email{Bjoern.Sprungk@math.tu-freiberg.de}

\affil*[1]{\orgdiv{Institute of Stochastics}, \orgname{TU Bergakademie Freiberg}, \orgaddress{\street{Pr\"{u}ferstra\ss e 9}, \city{Freiberg}, \postcode{09599}, \country{Germany}}}

\affil[2]{\orgdiv{Institute of Mathematical Data Science}, \orgname{Universit\"{a}t Passau}, \orgaddress{\street{Innstra\ss e 33}, \city{Passau}, \postcode{94032}, \country{Germany}}}


\abstract{Slice sampling is a well-established Markov chain Monte Carlo method for approximate sampling of target distributions which are 
only known up to a normalizing constant. 
The method is based on choosing a new state on a slice, i.e., a superlevel set of the given unnormalized target density (with respect to a reference measure).
However, slice sampling algorithms usually require per step multiple evaluations of the target density, and thus can become computationally expensive.
This is particularly the case for Bayesian inference with costly likelihoods. 
In this paper, we exploit deterministic approximations of the target density, which are relatively cheap to evaluate, and propose delayed acceptance versions of several common (hybrid) slice samplers. 
We show ergodicity of the resulting slice sampling methods, discuss the superiority of delayed acceptance (ideal) slice sampling over delayed acceptance Metropolis-Hastings algorithms, and illustrate the benefits of our novel approach in terms of improved computational efficiency in numerical experiments.}




\maketitle

\section{Introduction}\label{introduction}

A central task in Bayesian inference and computational statistics is to (approximate) sample with respect to (w.r.t.) a (posterior) probability measure $\pi$ on a measurable space $(G,\mathcal{B}(G))$, with $G\subseteq\mathbb{R}^d$, of the form
\begin{align}\label{PMLR:distribution}
	\pi(\mathrm{d}x) 
	& = \frac {1}{Z} \, \varrho(x) \, \pi_0(\mathrm{d}x), 
	& Z & :=\int_G^{} \varrho(x)  \, \pi_0(\mathrm{d}x),
\end{align}      
where $\varrho\colon G \to (0,\infty)$ denotes a $\mathcal{B}(G)$-measurable density w.r.t. a $\sigma$-finite prior or reference measure $\pi_0$.
Often $\pi$ is only known up to the normalizing constant $Z$ and a direct sampling of $\pi$ is infeasible.
Moreover, evaluations of the density or likelihood $\varrho$ can be costly in applications due to accessing large data sets or the involvement of numerical methods for differential equations as, e.g., in Bayesian inverse problems \cite{Stuart2010}.

A well-established approach for approximate sampling of partially known target distributions $\pi$ is Markov chain Monte Carlo (MCMC). 
Here, in order to generate Markov chains with $\pi$ as their unique limit distribution Metropolis--Hastings (MH) algorithms \cite{MetropolisEtAl1953, RobertsRosenthal2004} represent a common and rich class of often also easily implementable methods. 
As an alternative, the slice sampling approach, see e.g., \cite{Neal2003}, has drawn considerable attention in recent years. 
Slice sampling methods provide several advantages compared to MH algorithms, e.g., the absence of free parameters that need to be tuned and typically a better mixing behavior than MH algorithms \cite{RudolfUllrich2018}.
Furthermore, slice sampling can be less affected by targets with multiple modes than classical MH algorithms. 
Unfortunately, slice sampling in its ideal form is usually not implementable in practice because it requires direct sampling of the normalized restriction of $\pi_0$ to superlevel sets of the density $\varrho$, which is usually not feasible in practice. 
Therefore, one relies on hybrid slice sampling, where the direct sampling step is replaced by approximate sampling according to suitable Markov chain transition kernels.
The drawback of hybrid slice sampling can be a high computational cost, since these algorithms often require several evaluations of the density $\varrho$.

To improve the computational efficiency of hybrid slice sampling we propose to exploit $\mathcal{B}(G)$-measurable deterministic approximations $\rhoapp\colon G \to (0,\infty)$ of $\varrho$, which are cheaper to evaluate than $\varrho$ itself.
In particular, we adapt the delayed acceptance approach for MH algorithms \cite{ChristenFox2005, BanterleEtAl2019} to slice sampling and derive delayed acceptance versions of popular hybrid slice sampler.
We show ergodicity of the resulting novel algorithms and illustrate their potential for a significant cost reduction in numerical experiments.
Moreover, we show that the superiority of ideal slice sampling over MH algorithms with $\pi_0$-reversible proposal kernels also holds for their delayed acceptance versions---which extends the results in \cite{RudolfUllrich2018}.

Cost reduction for hybrid slice sampling has already been considered by \citet{DuBoisEtAl2014} for cases where $\varrho$ represents a likelihood function involving large data sets.
There, similarly to minibatch stochastic gradient descent \cite{BottouEtAl2018}, a random subsampling strategy of the data was proposed leading to a biased MCMC method.
Moreover, \citet{MurrayGraham2016} combine hybrid slice sampling with the pseudo-marginal MCMC approach to handle also intractable densities or likelihoods $\varrho$ using unbiased stochastic approximations of $\varrho$.
These approaches use random approximations of $\varrho$ and follow a different paradigm than our work.
Besides subsampling of large data sets cheap approximations to $\varrho$ can also be obtained in case of Bayesian inverse problems by using a coarser numerical discretizations of the forward map.
This is commonly exploited in multilevel MCMC methods \cite{DodwellEtAl2019}.
Our work can be seen as a first step towards multilevel slice sampling following, e.g., the hierarchical delayed acceptance setting for MH algorithms by \citet{LykkegaardEtAl2022}.
We note that we already briefly outlined our delayed acceptance approach for the case of elliptical slice sampling in \cite{BitterlichEtAl2025}. 

\change{\paragraph{Contributions and outline} 
The contributions of this work can be summarized as follows
\begin{itemize}
    \item 
    we define delayed acceptance (hybrid) slice sampling, verify its well-posedness and prove that delayed acceptance ideal slice sampling dominates its delayed acceptance Metropolis--Hastings counterpart in terms of a smaller asymptotic variance,   
    \item 
    we propose explicit novel cost efficient delayed acceptance versions of common hybrid slice samplers, namely, elliptical slice sampling, hit-and-run slice sampling, and Gibbsian polar slice sampling,
    \item 
    we prove for each of these novel samplers $\pi$-invariance and ergodicity under standard assumptions and illustrate in three different numerical experiments their potential for substantial computational cost reduction.
\end{itemize}
We also emphasize that quantitative mixing bounds and the automated construction of the employed approximations are not addressed in the present work and are left for future research.}

The paper is organized as follows. 
In Section~\ref{sec:slice}, we provide a \change{conceptual} overview of ideal and hybrid slice sampling. 
Section~\ref{sec:DASS} presents our approach to exploiting approximations of $\varrho$ for slice sampling, 
describing delayed acceptance ideal slice sampling, and proving its theoretical superiority compared to its MH counterpart.
In Section~\ref{sec:DAHSS}, we discuss delayed acceptance for hybrid slice samplers, state delayed acceptance versions of popular hybrid slice samplers and show ergodicity of the resulting algorithms.
Section~\ref{sec:num} illustrates the cost benefits of our novel approach in numerical examples and we conclude in Section~\ref{sec:concl}.

\paragraph{Notation \change{and Setting}}
In the following, let $(\Omega, \mathcal{F}, \mathbb{P})$ be a probability space \change{underlying} all subsequently used random variables. By $\lambda_d$, with $d\in\mathbb N$, we denote the $d$-dimensional Lebesgue measure and by $\lVert \cdot \rVert$ the Euclidean norm on $\mathbb R^d$. 
By $\sigma_d$ we denote the surface measure on the unit sphere $\mathbb{S}_{d-1}\subset \mathbb R^d$ and write $\kappa_d=\sigma_d(\mathbb{S}_{d-1})$ for the total surface area. By $\mathrm{U}(a,b)$ we denote the uniform distribution on $(a,b)$ for $a, b\in\mathbb{R}$ with $a<b$ and by $\mathrm{U}(A)$\change{, $A\subset \mathbb R^d$,} the uniform distribution on a set $A$, whenever it is well-defined.
By $\mathrm N(m,C)$ we denote the normal distribution on $\mathbb{R}^d$ with mean $m \in \mathbb{R}^d$ and non-degenerate covariance matrix $C \in \mathbb{R}^{d\times d}$.
For convenience, we also 
abbreviate expressions 
such as \eqref{PMLR:distribution} simply by writing $\pi(\mathrm{d}x)\propto\varrho(x) \pi_0(\mathrm{d}x)$. 

\change{
Throughout this paper, let $(G,\mathcal{B}(G))$ with $G\subseteq \mathbb R^d$ be a measurable space and assume that the reference measure $\pi_0$ is $\sigma$-finite and that the mapping $\varrho\colon G\to(0,\infty)$ in \eqref{PMLR:distribution} is lower semi-continuous, which implies the $\mathcal{B}(G)$-measurability, and $\pi_0$-integrable.
Moreover, we assume that the supremum norm $\|\varrho\|_\infty$ coincides with the $\pi_0$-essential supremum of $\varrho$.    
}

\section{Slice Sampling}\label{sec:slice}
We aim to sample approximately from $\pi$ given in \eqref{PMLR:distribution} and want to compute expectations $\pi(f) := \int_G f(x)\, \pi(\mathrm{d} x)$ of $\pi$-integrable quantities of interest $f\colon G\to\mathbb R$. 
To this end, we employ MCMC, i.e., we generate a Markov chain $(X_n)_{n\in\mathbb{N}}$ with $\pi$ as its limit distribution.  
A necessary condition for the latter is that the \emph{transition} or \emph{Markov kernel} $P\colon G \times \mathcal{B}(G) \to [0,1]$ of the Markov chain,
\[
	P(x,A) 
    :=
    \mathbb{P}(X_{n+1}\in A \mid X_n=x),
    \qquad
    x \in G, \ A \in \mathcal{B}(G),
\]
is $\pi$-invariant, i.e., $\pi(A) = \pi P(A) = \int_G P(x, A)\, \pi(\mathrm d x)$ holds for any $A \in \mathcal{B}(G)$. To this end, it suffices to show \emph{$\pi$-reversibility} of $P$, i.e.,
\begin{align*}
    \int_{A} P(x,B)  \, \pi(\mathrm{d}x)=\int_{B} P(y,A)  \, \pi(\mathrm{d}y), \qquad \forall A,B\in\mathcal{B}(G).
\end{align*}

\noindent However, the $\pi$-invariance of a transition kernel $P$ alone does not guarantee the convergence of the underlying Markov chain for an arbitrary initial value $x_0 \in G$ to the target $\pi$, i.e.
 \begin{align}\label{eq:ergodic}
        \lim_{n \to \infty} 
        \| P^n(x_0,\cdot) - \pi  \|_{\mathrm{TV}} = 0,
  \end{align}
where $P^n(x_0, \cdot)$ denotes the $n$-step transition kernel of $P$ and $\| \cdot \|_{\mathrm{TV}}$ the total variation distance. Therefore, we require $P$ to be $\pi$-invariant and either \emph{ergodic}, namely when (\ref{eq:ergodic}) holds for all $x_0 \in G$, or at least \emph{$\pi$-almost everywhere ergodic}, namely when (\ref{eq:ergodic}) is satisfied for $\pi$-almost every $x_0 \in G$. \change{Note that (\ref{eq:ergodic}) does not provide any information about how fast the Markov chain approaches $\pi$. In contrast, so-called geometric ergodicity additionally guarantees an exponential convergence (w.r.t.~$n$) to the invariant distribution.} 
Given a 
$\pi$-invariant Markov chain $(X_n)_{n\in\mathbb{N}}$ we can use the path average estimator 
\begin{align}\label{path}
A_n(f) := \frac 1n \sum_{k=1}^n f(X_k)
\end{align}
to approximate $\pi(f)$ where almost sure convergence of $A_n(f)$ to $\pi(f)$ as $n\to\infty$ is ensured under additional mild assumptions (cf. \cite[Theorem 3]{Tierney1994}).

Slice sampling is a classical technique for the construction of $\pi$-invariant Markov chains $(X_n)_{n\in\mathbb{N}}$, see e.g. \cite{besag1993spatial,Neal2003}, which are, moreover, (geometrically) ergodic under relatively mild assumptions on the target distribution $\pi$, see e.g. \cite{RobertsRosenthal1999,MiraTierney2002}.
In the following, we provide basic definition and a \change{conceptual} overview of ideal and hybrid slice sampling.

\subsection{Ideal Slice Sampling}
For $t \in (0, \|\varrho\|_\infty)$
we define the $\mathcal{B}(G)$-measurable superlevel set of $\varrho$ by
\begin{align}\label{PMLR:levelset}
	G_{t,\varrho}:=\left\{x\in G \colon \varrho(x) > t\right\}	
\end{align}
and denote with
\begin{align*}
	\pi_{0,t}(A):=\frac{\pi_0(A \cap G_{t,\varrho})}{\pi_0(G_{t,\varrho})}, \quad A\in\mathcal{B}(G),
\end{align*}
the normalized restriction of the reference measure $\pi_{0}$ to the superlevel set $G_{t,\varrho}$. The transition mechanism of ideal slice sampling given a current state $X_n = x\in G$ is presented in Algorithm \ref{alg:SS}.

\begin{algorithm}
	\caption{Ideal Slice Sampling}\label{alg:SS}
	\renewcommand{\algorithmicrequire}{\textbf{Input:}}
	\renewcommand{\algorithmicensure}{\textbf{Output:}}
	\begin{algorithmic}[1]
		\Require current state $X_n = x$
		\Ensure next state $X_{n+1} = y$
		\State Draw sample $t$ of $\mathrm{U}(0,\varrho(x))$
		\State Draw sample $y$ of $\pi_{0,t}$
	\end{algorithmic}
\end{algorithm}

The transition kernel $S\colon G \times \mathcal{B}(G) \to [0,1]$ of Markov chains $(X_n)_{n\in\mathbb{N}}$ generated by ideal slice sampling reads as follows
\begin{align*}
	S(x,A) 
    =
    \frac{1}{\varrho(x)}\int_{0}^{\varrho(x)} \pi_{0,t}(A) \, \mathrm{d}t, \quad x\in G, \phantom{.} A\in\mathcal{B}(G).
\end{align*}
It is straightforward to show that $S$ is \emph{$\pi$-reversible}.
If the reference measure $\pi_0$ is the $d$-dimensional Lebesgue measure, then in the second step of Algorithm \ref{alg:SS} we draw samples w.r.t.~the uniform distribution on the superlevel set $G_{t,\varrho}$. 
In this case, we also refer to Algorithm \ref{alg:SS} as \emph{simple slice sampling}. 
A convergence analysis for simple slice sampling can be found, e.g., in \cite{RobertsRosenthal1999, MiraTierney2002,  NatarovskiiEtAl2021b} and theoretical properties of different ideal slice sampler are, e.g., developed in \cite{rudolf2024dimension,schar2023wasserstein,LatuszynskiRudolf2024,PowerEtAl2024}.

\subsection{Hybrid Slice Sampling}\label{subsec:hss}
In practice, a direct simulation of the distribution $\pi_{0,t}$ is often not feasible and ideal slice sampling, hence, often not implementable.
A remedy are hybrid slice sampling (HSS) methods, where sampling w.r.t. $\pi_{0,t}$ is substituted by applying a $\pi_{0,t}$-invariant transition kernel $H_t$.
In particular, let $(H_t)_{t \in (0, \|\varrho\|_\infty)}$ denote in the following a family of $\pi_{0,t}$-reversible transition kernels $H_t$. 
Then, hybrid slice sampling is given by the transition kernel 
\begin{align*}
	H(x,A)=\frac{1}{\varrho(x)}\int_{0}^{\varrho(x)} H_t(x,A) \, \mathrm{d}t, \quad x\in G,\, A\in\mathcal{B}(G).   	
\end{align*}
Algorithmically, this corresponds to replacing the second step in Algorithm \ref{alg:SS} with

\vspace*{0.2cm}

\begin{algorithmic}[1]
\setcounter{ALG@line}{1}
       \State Draw sample $y$ of $H_t(x,\cdot)$.
\end{algorithmic}

\vspace*{0.2cm}

By the reversibility of the kernels $H_t$ the transition kernel $H$ is itself $\pi$-reversible \cite[Lemma~1]{LatuszynskiRudolf2024}.
However, it should be noted that hybrid slice sampling cannot perform better than ideal slice sampling in terms of the speed of convergence to $\pi$ or the asymptotic variance of $A_n(f)$, see, e.g., \cite{LatuszynskiRudolf2024, RudolfUllrich2018} for more details.

Well-known examples of implementable hybrid slice sampling methods are 
elliptical slice sampling (ESS), developed in \cite{MurrayEtAl2010}, hit-and-run uniform slice sampling combined with stepping-out and shrinkage (HRUSS), considered in \cite{LatuszynskiRudolf2024}, and Gibbsian polar slice sampling (GPSS), suggested in \cite{SchaerEtAl2023}.

\section{Delayed Acceptance Slice Sampling}\label{sec:DASS}
Hybrid slice sampling methods may require several evaluations of the density or likelihood $\varrho$ for generating a single new state, as can be seen, e.g., in the shrinkage procedure of ESS. This can be a computational bottleneck, particularly if the evaluation of $\varrho$ is costly. To reduce the computational burden of \change{(hybrid)} slice sampling, we propose to exploit an approximation $\rhoapp\colon G \to (0,\infty)$ to $\varrho$ which is cheaper to evaluate. In general, the construction of a suitable $\rhoapp$ is problem-dependent. Examples of such approximations are provided in the numerical experiments in Section \ref{sec:num}.

In order to obtain again $\pi$-invariant Markov chains 
we adapt the concept of delayed acceptance for MH algorithms \cite{ChristenFox2005} to slice sampling.
Given an approximating density $\rhoapp\colon G\to(0,\infty)$ we can decompose the original density $\varrho\colon G\to(0,\infty)$ as follows
\begin{align}
	\varrho = \widehat{\varrho} \cdot \rhoapp, \quad \widehat{\varrho}:=\varrho/\rhoapp \label{PMLR:factorization}.
\end{align}  
Consequently, the target can be represented as $ \pi(\text{d} x) \propto \widehat{\varrho}(x) \, \rhoapp(x) \, \pi_0(\text{d}x)$. 
In fact, the algorithm we propose in Section \ref{sec:DASS} corresponds to product slice sampling for a two-factor product density as considered in \cite{RobertsRosenthal1999, MiraTierney2002}.
However, the crucial novelty lies in the consideration and reduction of the computational costs for a specific algorithmic implementation as well as the extension to feasible hybrid slice samplers. 

For comparison, we recall the approach of delayed acceptance for MH algorithms before explaining how to exploit this idea for slice sampling.

\subsection{Delayed Acceptance Metropolis-Hastings Algorithm}\label{sec:DAMH}
Let $Q\colon G \times \mathcal B(G)\to [0,1]$ be a $\pi_0$-reversible transition kernel used to propose possible new states of the Markov chain by drawing a sample $y$ from $Q(x,\cdot)$ given $X_n=x$. 
In the delayed acceptance approach, the proposed $y$ is then accepted with probability
\begin{align*}
	\widetilde{\alpha}(x,y):=\underbrace{\text{min}\Bigg\{1, \frac{\rhoapp(y)}{\rhoapp(x)} \Bigg\}}_{=:\alpha_1(x,y)} \ \underbrace{\text{min}\Bigg\{1, \frac{ \widehat{\varrho}(y)}{\widehat{\varrho}(x)} \Bigg\}}_{=:\alpha_2(x,y)}.
\end{align*}
This allows for a two-stage procedure: First accept or reject w.r.t.~$\rhoapp$, i.e., drawing $u_1$ from $\mathrm{U}(0,1)$ and reject immediately if $u_1 \geq \alpha_1(x,y)$, then, in case of $u_1 < \alpha_1(x,y)$, draw $u_2$ from $\mathrm{U}(0,1)$ independently and finally accept $y$ as new state only if also $u_2 < \alpha_2(x,y)$.
The computational advantage here is that if $y$ is already rejected in the first stage, then we do not have to evaluate the costly density $\varrho(y)$.
Only if the proposed state $y$ seems promising in terms of passing the acceptance test for the cheap proxy $\rhoapp$ we pay the additional cost of computing $\widehat{\varrho}(y) = \varrho(y)/\rhoapp(y)$. 
The algorithmic description of the resulting delayed acceptance Metropolis-Hastings (DA-MH), see \cite{ChristenFox2005}, can be found in Algorithm \ref{alg:DAMH}. 
Its transition kernel is given by 
\begin{align*}
	M_{\text{DA}}(x,A):=\int_{A} &\widetilde{\alpha}(x,y) \, Q(x,\mathrm{d}y) + \Big(1-\int_{G} \widetilde{\alpha}(x,y) \,  Q(x,\mathrm{d}y)\Big)\delta_x(A),
\end{align*}
where $x\in G$, $A\in\mathcal{B}(G)$, and $\delta_x(A)$ denotes the Dirac measure located at $x$. 
By construction $M_{\text{DA}}$ is $\pi$-reversible and can, interestingly, be rewritten as a hybrid slice sampler; see Proposition \ref{PMLR:prop:DAMH:HSS} below. 

\begin{algorithm}
	\caption{Delayed Acceptance Metropolis-Hastings}\label{alg:DAMH}
	\renewcommand{\algorithmicrequire}{\textbf{Input:}}
	\renewcommand{\algorithmicensure}{\textbf{Output:}}
	\begin{algorithmic}[1]
		\Require current state $X_n = x$
		\Ensure next state $X_{n+1} = y$
        \State Set $X_{n+1} = x$
        \State Draw sample $y$ of $Q(x, \cdot)$
        \State Draw sample $u_1$ of $\mathrm{U}(0,1)$
        \If{$u_1<\alpha_1(x,y)$}
        \State Draw sample $u_2$ of $\mathrm{U}(0,1)$
        \If{$u_2<\alpha_2(x,y)$}
        \State Set $X_{n+1} = y$
        \EndIf
        \EndIf
	\end{algorithmic}
\end{algorithm}

\subsection{Delayed Acceptance Ideal Slice Sampling}\label{sec:DAISS}
We now transfer the cost reduction strategy employed in DA-MH to slice sampling methods.
To this end, we perform slice sampling on the intersection of two superlevel sets $G_{s,\rhoapp}$ and $G_{t,\widehat{\varrho}}$ of $\rhoapp$ and $\widehat{\varrho}$, respectively.
These superlevel sets are defined analogously to (\ref{PMLR:levelset}). 
However, the sampling is performed in a delayed acceptance manner, which avoids unnecessary evaluations of the costly density $\varrho$.
To this end, for $s \in (0, \|\rhoapp\|_\infty)$ define the normalized restriction of $\pi_0$ to $G_{s,\rhoapp}$ by
\begin{align*}
	\widetilde{\pi}_{0,s}(A):=\frac{\pi_0(A \cap G_{s,\rhoapp})}{\pi_0(G_{s,\rhoapp})}, \quad A\in\mathcal{B}(G).
\end{align*} 
Then, the transition mechanism of delayed acceptance slice sampling (DA-SS) is given by Algorithm \ref{alg:DASS}. 
Moreover, we provide illustrations of slice sampling and DA-SS in Figure \ref{fig:DASS}.
There, the left panel illustrates one transition by Algorithm \ref{alg:SS} and the middle and right panel depict the one transition by Algorithm \ref{alg:DASS} by the two-step procedure of first generating a possible new state $y$ by sampling from $\widetilde{\pi}_{0,s}$ and then checking whether it is accepted, i.e., whether $y \in G_{t,\widehat{\varrho}}$, or not. \change{We emphasize that the DA-SS construction is primarily of theoretical interest, whereas practically relevant implementations rely on delayed acceptance hybrid slice sampling techniques, which are discussed in detail in Section~\ref{sec:DAHSS}.}

\begin{algorithm}
	\caption{Delayed Acceptance Ideal Slice Sampling}\label{alg:DASS}
	\renewcommand{\algorithmicrequire}{\textbf{Input:}}
	\renewcommand{\algorithmicensure}{\textbf{Output:}}
	\begin{algorithmic}[1]
		\Require current state $X_n = x$
		\Ensure next state $X_{n+1} = y$
     	\State Draw sample $s$ of $\mathrm{U}(0,\rhoapp(x))$
		\State Draw sample $t$ of $\mathrm{U}(0,\widehat{\varrho}(x))$
        \Repeat
		\State Draw sample $y$ of $\widetilde{\pi}_{0,s}$
		\Until{$y\in G_{t,\widehat{\varrho}}$}
	\end{algorithmic}
\end{algorithm}

\begin{figure*}[t]
	\begin{minipage}{0.32\linewidth}
		\centering
		\includegraphics[height=0.9\linewidth, width=0.88\linewidth]{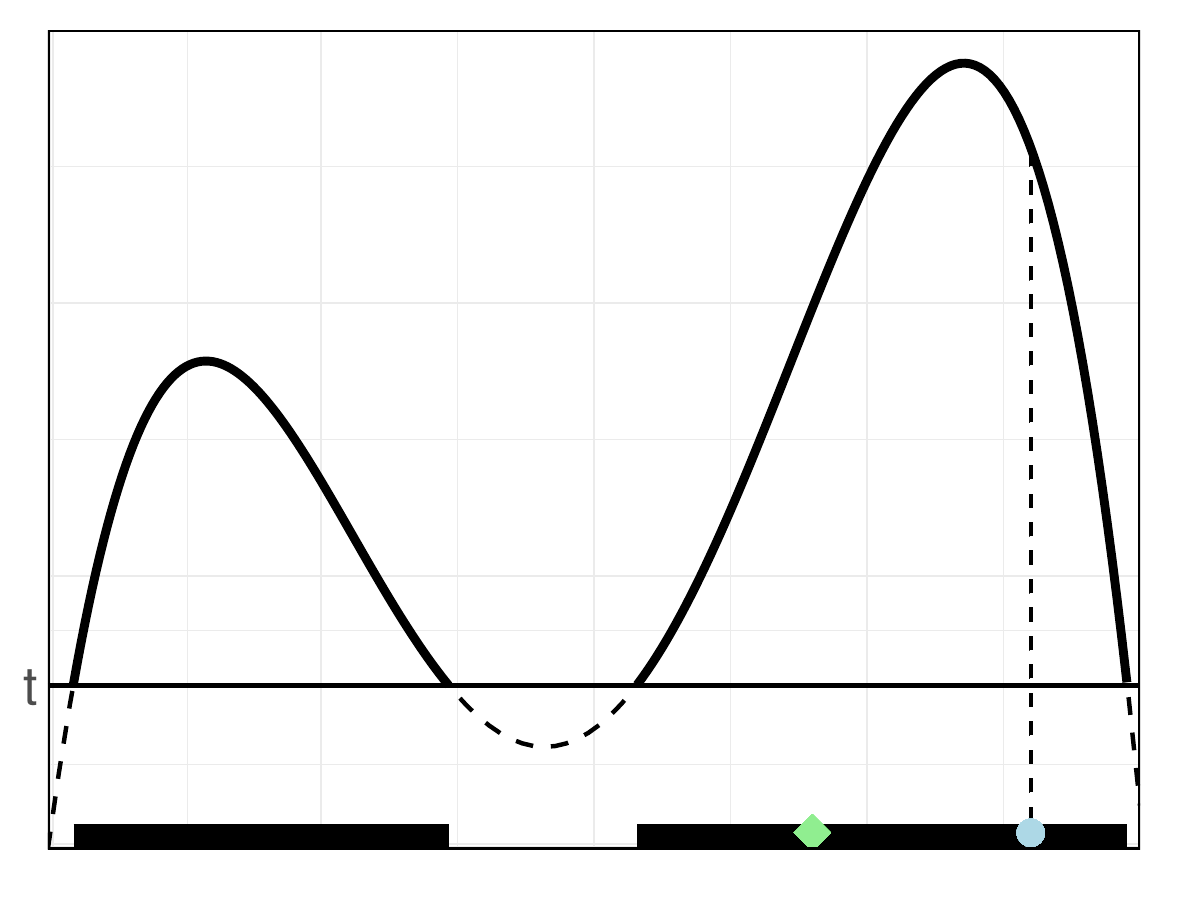}
	\end{minipage}
	\begin{minipage}{0.32\linewidth}
		\centering
		\includegraphics[height=0.9\linewidth, width=0.88\linewidth]{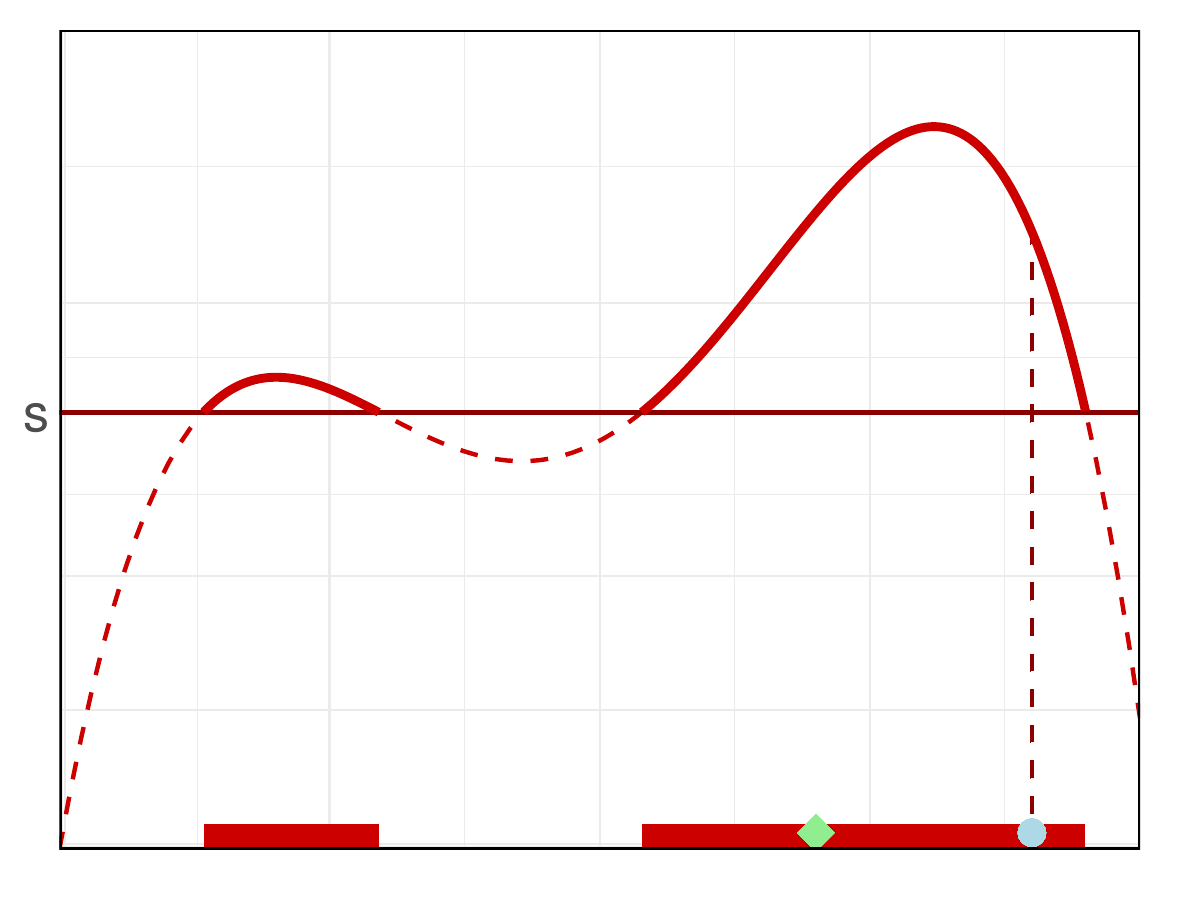}
	\end{minipage}
	\begin{minipage}{0.32\linewidth}
		\centering
		\includegraphics[height=0.9\linewidth, width=0.88\linewidth]{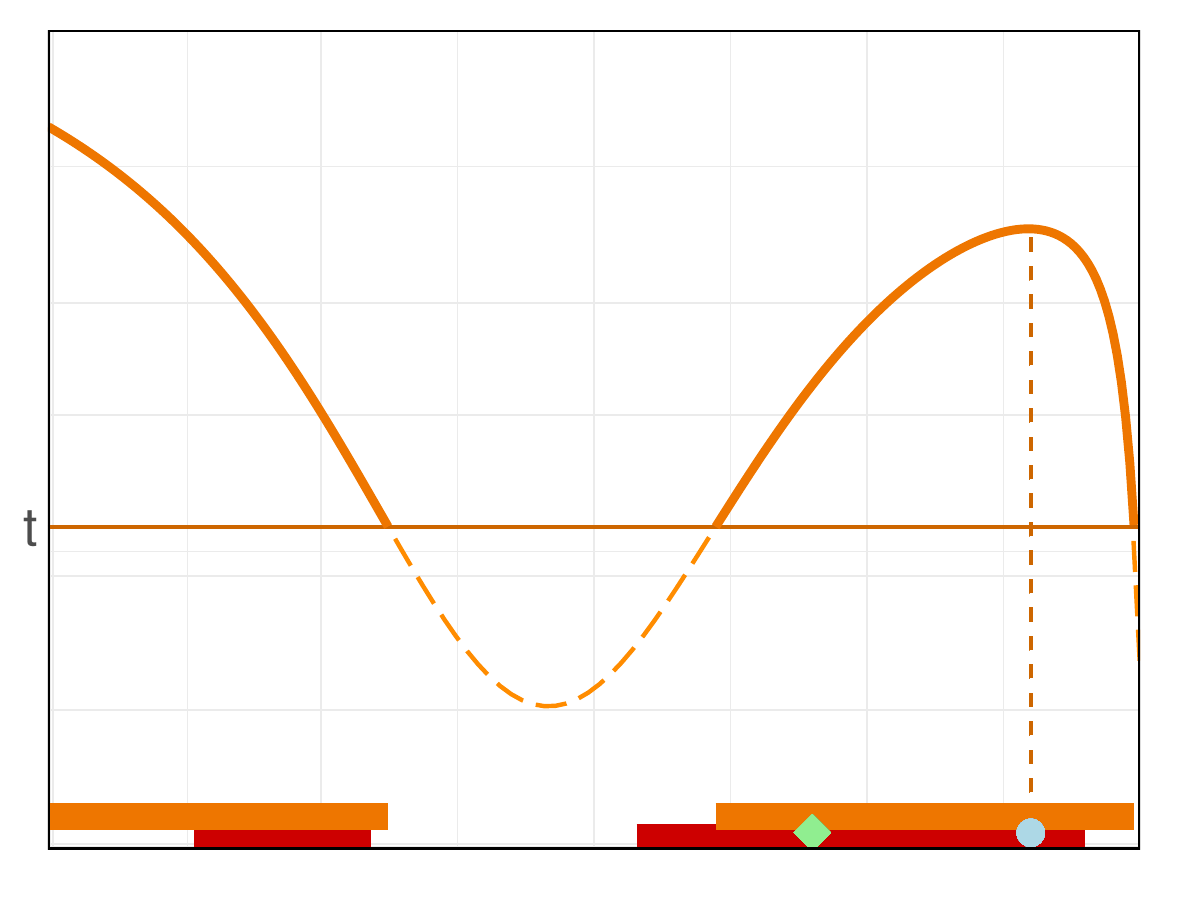}
	\end{minipage}
	\caption{\rev{\emph{Left:} Illustration of Algorithm \ref{alg:SS}, i.e., density $\varrho$ of target $\pi(\text{d} x) \propto \varrho(x) \, \text{d} x$ and resulting superlevel set $G_{t,{\varrho}}$ (thick black lines) given sample $t$ of $\mathrm{U}(0,\varrho(x))$ for the transition from current state $x$ (blue ball) to next state $y$ (green diamond).
			\emph{Middle:} Illustration of step 4 in Algorithm \ref{alg:DASS}, i.e., approximation $\rhoapp$ of $\varrho$ with superlevel set $G_{s,{\rhoapp}}$ (thick red lines) given sample $s$ of $\mathrm{U}(0,\rhoapp(x))$ for generating proposed new state $y$ (green diamond) given current state $x$ (blue ball).
			\emph{Right:} Illustration of step 5 in Algorithm \ref{alg:DASS}, i.e., ratio $\widehat{\varrho}=\varrho/\rhoapp$ with superlevel set $G_{t,\widehat{\varrho}}$ (thick orange lines) given sample $t$ of $\mathrm{U}(0,\widehat{\varrho}(x))$ and verification that proposed new state $y$ (green diamond) belongs to $G_{t,\widehat{\varrho}}$.}}
	\label{fig:DASS}
\end{figure*}
In order for the loop in line 3 to 5 in Algorithm \ref{alg:DASS} to terminate, we require that 
\change{
\begin{align}\label{eq:positivity_condition}
    \pi_0(G_{s,\rhoapp} \cap G_{t,\widehat{\varrho}}) & > 0& & \forall 
    (s,t)\in \mathcal T
\end{align}
where
\[
    \mathcal T 
    := 
    \bigcup_{x\in G} (0,\rhoapp(x)) \times (0,\widehat{\varrho}(x))
\]
This condition ensures that the corresponding superlevel sets at some location $x\in G$ have a non-degenerate overlap with respect to the reference measure. In order guarantee the well-definedness of Algorithm \ref{alg:DASS} and all subsequently derived delayed acceptance hybrid slice sampling algorithms we make throughtout this paper the following standing assumption.
\begin{assumption}\label{assum:positivity_condition}
The functions $\rhoapp\colon G\to(0,\infty)$ and $\rhohat\colon G\to(0,\infty)$ in \eqref{PMLR:factorization} are lower-semicontinuous and every $\pi_0$-null set in $\mathcal B(G)$ is also a $\lambda_d$-null set. 
\end{assumption}

Lower semicontinuity of the density is a standard assumption in slice sampling and the above assumption particularly implies the lower semicontinuity of $\varrho$ as the product of $\rhoapp$ and $\rhohat$. The assumption regarding the reference measure $\pi_0$ is easily satisfied if $\pi_0$ possesses a strictly positive Lebesgue density on $G$. 
Finally, we note that Assumption \ref{assum:positivity_condition} ensures \eqref{eq:positivity_condition}.
}

A computational advantage of delayed acceptance ideal slice sampling over ideal slice sampling is not yet present in Algorithm \ref{alg:DASS}, since we also have to evaluate $\varrho$ in each iteration of the loop.
However, it will become obvious when we extend the delayed acceptance approach to hybrid slice sampling in Section \ref{sec:DAHSS}.
It should be noted that the number of loop iterations in Algorithm \ref{alg:DASS}, and consequently the number of evaluations of $\varrho$, depend on the quality of the approximation $\rhoapp$. A poor choice can result in a large number of samples $y$ drawn from $\widetilde{\pi}_{0,s}$ before one of them is accepted (that is $y\in G_{t,\widehat{\varrho}}$).

\paragraph{Reversibility}

\change{In view of Assumption \ref{assum:positivity_condition} ensuring \eqref{eq:positivity_condition},
the transition kernel of DA-SS is for all $x\in G$ and $A\in\mathcal{B}(G)$ given by
\begin{align}\label{PMLR:DASS:transitionkernel}
	S_{DA}(x,A):= \frac{1}{\varrho(x)}\int_0^{\widehat{\varrho}(x)} \int_0^{\rhoapp(x)} \pi_{0,t,s}(A)  \, \mathrm{d}s  \, \mathrm{d}t, 
\end{align}
where 
\begin{align*}
	\pi_{0,t,s}(A):=\frac{\pi_0(A \cap G_{s,\rhoapp} \cap  G_{t,\widehat{\varrho}})}{\pi_0(G_{s,\rhoapp} \cap G_{t,\widehat{\varrho}})}.
\end{align*}}

\begin{remark}[Relation to product slice sampling]
Delayed acceptance slice sampling coincides with product slice sampling for the factorized target $\pi(\text{d} x) \propto \widehat{\varrho}(x)\,\rhoapp(x) \, \pi_0(\text{d}x)$.
Results on geometric and uniform ergodicity of product slice sampling are provided in \cite{RobertsRosenthal1999} as well as in \cite{MiraTierney2002}.
However, since their assumptions require $\| \widehat{\varrho}\|_\infty < \infty$ common approximations $\rhoapp$ may be not applicable.
Also, the proof for the uniform ergodicity in \cite{MiraTierney2002} works only if $\pi_0(G)<\infty$, which excludes important cases, e.g. when $G=\mathbb R^d$ and the reference measure $\pi_0$ is the Lebesgue measure.
\end{remark}

DA-SS and the transition kernel $S_{DA}$, as well as in fact any product slice sampler, can also be rewritten as hybrid slice sampling. 
To this end, we introduce an auxiliary probability measure $\piapp$ on $(G, \mathcal B(G))$ given by
\begin{align}\label{PMLR:distribution:factorization}
	\piapp(\text{d} x) 
    & \propto \rhoapp(x)\ \pi_0(\text{d}x),
    & \pi(\mathrm{d}x)
	& \propto \widehat{\varrho}(x) \ \pi_{\text{app}}(\mathrm{d}x).
\end{align} 
We can then rewrite $S_{DA}$ as a hybrid slice sampler, where the Markov kernels $H_t$ coincide with ideal slice sampling of the normalized restriction of $\piapp$ to $G_{t,\widehat{\varrho}}$.
\begin{proposition}\label{PMLR:prop:hybrid}
	For $t \in (0, \|\widehat\varrho\|_\infty)$ let
	\begin{align*}
		\pi_{{\rm app},t}(A):=\frac{\pi_{{\rm app}}(A \cap G_{t,\widehat{\varrho}})}{\pi_{{\rm app}}(G_{t,\widehat{\varrho}})}, \quad A\in\mathcal{B}(G) 	
	\end{align*}
	and \change{given Assumption \ref{assum:positivity_condition}} define 
	\begin{align}\label{eq:S_DA_Ht}
		H_t(x,A):=\frac{1}{\rhoapp(x)} \int_0^{\rhoapp(x)} \pi_{0,t,s}(A)  \, \mathrm{d}s
	\end{align}
	where $x\in G$, $A\in\mathcal{B}(G)$ .
	Then, $H_t$ is $\pi_{{\rm app},t}$-reversible for each $t \in (0, \|\widehat\varrho\|_\infty)$ and
	\begin{align}\label{eq:S_DA_hybrid}
		S_{DA}(x,A) 
		& = \frac{1}{\widehat\varrho(x)} \int_{0}^{\widehat\varrho(x)}
		H_t(x,A) \, \mathrm{d}t. 
	\end{align}  
Moreover, the transition kernel $S_{DA}$ is $\pi$-reversible.
\end{proposition}

\noindent The proof of Proposition \ref{PMLR:prop:hybrid} can be found in Appendix \ref{app:prop:hybrid}. 

\paragraph{Ergodicity}

We start with stating an auxiliary tool which follows by \cite[Corollary 1 and Theorem 1]{Tierney1994} that is also formulated and used in \cite[Section~1]{Schaer2025}. 

\begin{proposition}
\label{Schaer:Convergence}
Assume that transition kernel $P\colon G \times \mathcal{B}(G)\to[0,1]$ is $\pi$-invariant as well as that for each $x\in G$ and $A\in \mathcal{B}(G)$ we have
\begin{align}\label{eq:ergodic:abscon1}
    \pi(A) > 0
    &\quad \Rightarrow \ \
    P(x,A) > 0.
\end{align}
Then, $P$ is $\pi$-almost everywhere ergodic, i.e., satisfies \eqref{eq:ergodic} for $\pi$-almost all $x_0 \in G$. \\
If additionally, for each $x\in G$ and $A\in \mathcal{B}(G)$ we also have
\begin{align}\label{eq:ergodic:abscon2} 
    \pi(A) = 0
    & \quad \Rightarrow \ \
    P(x,A) = 0,
\end{align}
then $P$ is ergodic, i.e., satisfies \eqref{eq:ergodic} for all $x_0 \in G$.
\end{proposition}

Note that \eqref{eq:ergodic:abscon1} and \eqref{eq:ergodic:abscon2} are equivalent to $\pi \ll P(x,\cdot)$ and $P(x,\cdot) \ll \pi$, respectively, for any $x\in G$.
These sufficient conditions can be verified for delayed acceptance slice sampling, see Appendix \ref{app:DASS_convergence}, and we obtain
the following result.

\begin{theorem}\label{DASS:Ergodicity}
    Given Assumption \ref{assum:positivity_condition} the transition kernel $S_{DA}$ of delayed acceptance slice sampling is $\pi$-reversible and ergodic.
\end{theorem}

\subsection{Comparison of Delayed Acceptance Metropolis--Hastings and Delayed Acceptance Ideal Slice Sampling}\label{sec:comparison}

We now argue for the superiority of DA-SS over DA-MH.
The following result can be considered as extension of the comparison arguments of \cite{RudolfUllrich2018} for slice sampling and MH to their delayed acceptance counterparts.
Firstly, we rewrite the DA-MH algorithm as suitable hybrid slice sampler. 

\begin{proposition}\label{PMLR:prop:DAMH:HSS}
	Consider the family $(H_t)_{t \in (0, \|\widehat\varrho\|_\infty)}$ of $\pi_{{\rm app},t}$-reversible transition kernels defined by
	\begin{align*}
		H_t(x,A):= &\frac{1}{\rhoapp(x)} \int_0^{\rhoapp(x)} \Big(Q(x,A \cap G_{s,\rhoapp} \cap  G_{t,\widehat{\varrho}}) 
		\\ & \quad +\big(1-Q(x,G_{s,\rhoapp} \cap  G_{t,\widehat{\varrho}})\big)\delta_x(A) \Big) \, \mathrm{d}s,
	\end{align*}
	where $x\in G$, $A\in\mathcal{B}(G)$.
	Then, for all $x\in G$ and $ A\in\mathcal{B}(G)$ we have 
	\begin{align*}
		M_{DA}(x,A) = \frac{1}{\widehat\varrho(x)}\int_{0}^{\widehat\varrho(x)} H_t(x,A)\, \mathrm{d}t. 
	\end{align*}	
\end{proposition}
The proof of Proposition \ref{PMLR:prop:DAMH:HSS} can be found in Appendix \ref{sec:DAMH:HSS}.
For the comparison result between (ideal) delayed acceptance slice sampling and the corresponding DA-MH algorithm we consider the so-called covariance ordering for Markov operators, cf. \cite{Mira2001, RudolfUllrich2018}.
Let $L^2_{\pi}(G)$ denote the Hilbert space of (w.r.t. $\pi$) square integrable measurable functions $f:G\rightarrow \mathbb{R}$ with inner product given by
\begin{align*}
	\langle f_1,f_2\rangle_{\pi}=\int_{G} f_1(x) f_2(x)  \, \pi(\mathrm{d}x).	
\end{align*}
For a transition kernel $P$ define the corresponding Markov operator, which is also denoted by $P\colon L^2_{\pi}(G) \to L^2_{\pi}(G)$, through
\begin{align*}
	Pf(x)=\int_{G} f(y)  \, P(x,\mathrm{d}y), \quad f\in L^2_{\pi}(G).	
\end{align*}
It is well-known that for a $\pi$-reversible transition kernel the corresponding operator $P$ is self-adjoint. 
We call such a self-adjoint operator $P$ positive if 
\begin{align*}
	\langle Pf,f\rangle_{\pi}\geq 0, \quad \forall f\in L^2_{\pi}(G). 	
\end{align*}
For two Markov operators $P_1$ and $P_2$, we then define a partial ordering by
\begin{align}\label{PMLR:ordering}
	P_1\leq P_2
    \quad \Longleftrightarrow \quad  \langle P_1f,f \rangle_{\pi}\geq \langle P_2f,f\rangle_{\pi}, \quad \forall f\in L^2_{\pi}(G).
\end{align}
This is a useful ordering, since 
for example in \cite{MiraLeisen2009} it has been shown that
\begin{align*}
	P_2\geq P_1 \quad  \iff \quad V(f, P_2)\leq V(f, P_1), \quad \forall f\in L^2_{\pi}(G), 	
\end{align*}
where, $V(f, P) := \lim_{n\to\infty} n\mathbb E|A_n(f) - \pi(f)|^2$, for $f\in L^2_{\pi}(G)$ and $\pi$-invariant transition kernel $P$, denotes the asymptotic variance of the path average estimator $A_n(f)$ based on a Markov chain $(X_k)_{k\in\mathbb N}$ with transition kernel $P$ and initial distribution $\pi$. 
Moreover, the relation \eqref{PMLR:ordering} also has implications on the speed of convergence of the associated Markov chains. As outlined in \cite{RudolfUllrich2018}, $P_2\geq P_1$ implies that the spectral gap, conductance, and log-Sobolev constant of $P_2$ are greater or equal to those of $P_1$. 
\begin{theorem}\label{PMLR:theorem:comparison}
	Let $S_{DA}$ and $M_{DA}$ be the Markov operators of the delayed acceptance slice sampling and of the delayed acceptance Metropolis-Hastings, respectively, with a $\pi_0$-reversible proposal kernel $Q$, which is positive on $L^2_{\pi_0}(G)$, where $\pi_0$ is a $\sigma$-finite measure. Then
	\begin{align*}
		M_{DA}\leq S_{DA}, 
	\end{align*}   
\change{where the ordering is understood in the 
sense of $(\ref{PMLR:ordering})$}.
\end{theorem}

Theorem \ref{PMLR:theorem:comparison} is proven in Appendix \ref{Apendix:Comparison} and yields that delayed acceptance slice sampling performs better than delayed acceptance Metropolis-Hastings with $\pi_0$-reversible proposals.
\change{In particular, the result implies that, when started in stationarity, i.e. w.r.t. $\pi$, the asymptotic variance of the path average estimator is smaller and the spectral gap larger for DA-SS than for DA-MH. However, this is an asymptotic comparison and DA-SS may be more difficult to implement than DA-MH in many applications. We illustrate the observation regarding the asymptotic variance in the following numerical experiment.}

\begin{example}\label{exam:DASS} 
    Consider $G=\mathbb R$, the Lebesgue measure as reference measure, i.e. $\pi_0=\lambda_1$ and the unnormalized densities
\begin{align*}
    \varrho(x)=\exp(\lvert x \rvert - \lvert x \rvert^2/2), \qquad \rhoapp(x)=\exp(-\lvert x \rvert^2/2).
\end{align*}
    For the representation of the normalized densities, see Figure \ref{fig:Comparison}. We apply delayed acceptance ideal slice sampling and the delayed acceptance version of the Gaussian Random Walk MH with the proposal kernel $Q(x, \cdot)=\mathrm{N}(0,s^2)$, where $s>0$ denotes a step size parameter, to sample approximately from the resulting probability measure $\pi$. 
    We choose the step size $s$ so that the acceptance rate for the delayed acceptance Gaussian Random Walk is around $30\%$. We set $f(x):=x$ and run both algorithms $n=10^6$ iterations after a burn-in of $n_0=10^5$ iterations (believing this is `almost' starting w.r.t. $\pi$) and compare the estimated asymptotic variances. We obtain
    \[          
    \widehat{V}(f, S_{DA}) = 2.2912, \qquad         \widehat{V}(f, M_{DA}) = 31.1482,     
    \] 
    that is, the estimated asymptotic variance for DA-SS is significantly lower than for DA-MH. In addition, the histograms of the Markov chain realizations in Figure \ref{fig:Comparison} further illustrate the conclusions of Theorem \ref{PMLR:theorem:comparison} on the convergence 
    behavior. 
\end{example}


\begin{figure*}[t]
	\begin{minipage}{0.32\linewidth}
		\centering
		\includegraphics[width=1.1\linewidth]{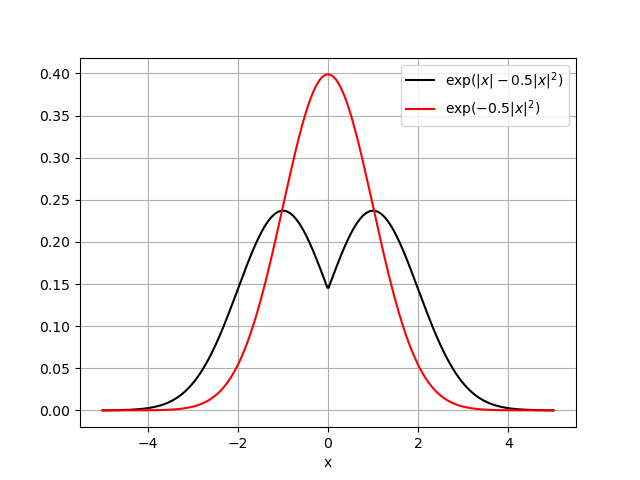}
	\end{minipage}
	\begin{minipage}{0.32\linewidth}
		\centering
		\includegraphics[width=1.1\linewidth]{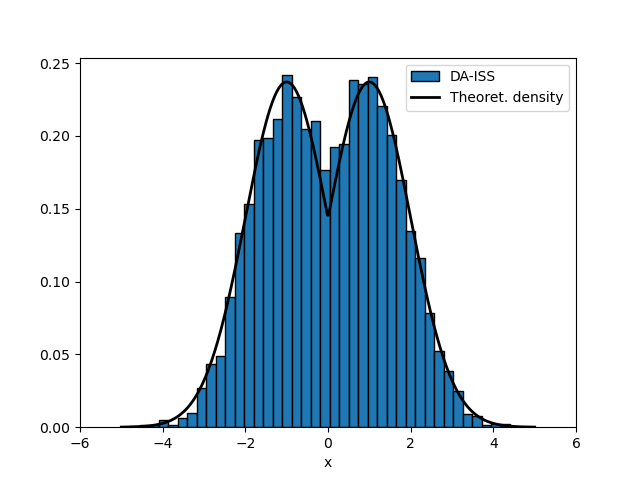}
	\end{minipage}
	\begin{minipage}{0.32\linewidth}
		\centering
		\includegraphics[width=1.1\linewidth]{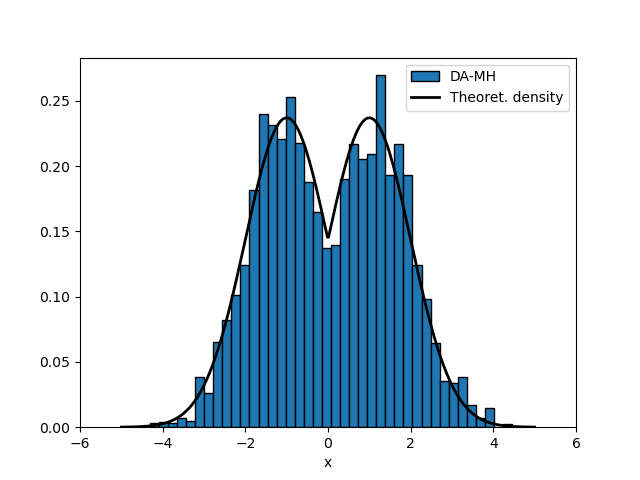}
	\end{minipage}
	\caption{Normalized densities $\varrho$ and $\rhoapp$ (left), histogram of DA-SS simulation compared to theoretical normalized density $\varrho$ (middle) and histogram of DA-MH simulation compared to theoretical normalized density $\varrho$ (right) each for the first $10.000$ iterations of the burn-in.}
	\label{fig:Comparison}
\end{figure*}


\change{As mentioned}, the difficulty in implementing DA-SS still persists, whereas this disadvantage is not present for the DA-MH algorithm. 
Therefore, the comparison result presented above should be interpreted primarily on a theoretical level and is largely confined to settings in which delayed acceptance ideal slice sampling is implementable, that is, to cases where the normalized restrictions $\widetilde{\pi}_{0,s}$ of the prior to superlevel sets of $\rhoapp$ can be simulated directly. 
In order to derive implementable delayed acceptance slice sampling algorithms, we again turn to hybrid slice sampling.

\section{Delayed Acceptance Hybrid Slice Sampling}\label{sec:DAHSS}
As 
a direct simulation of 
$\widetilde{\pi}_{0,s}$ 
 in step 4 of Algorithm \ref{alg:DASS} is in general not possible, we assume to have access to
a family $(\widetilde H_s)_{s \in (0, \|\varrho_\text{app}\|_\infty)}$ of $\widetilde{\pi}_{0,s}$-reversible transition kernels $\widetilde H_s$ and replace step 4 of Algorithm \ref{alg:DASS} by  

\vspace*{0.2cm}

\begin{algorithmic}[1]
\setcounter{ALG@line}{3}
	\State Draw sample $y$ of $\widetilde H_s(x,\cdot)$.
\end{algorithmic}
\vspace*{0.2cm}
This gives a first version of delayed acceptance hybrid slice sampling (DA-HSS) with corresponding transition kernel
\begin{align*}
	\widetilde{H}_{DA}(x,A)= \frac{1}{\varrho(x)}\int_0^{\widehat{\varrho}(x)} \int_0^{\rhoapp(x)} \frac{\widetilde H_s(x,A \cap G_{t, \widehat{\varrho}})}{\widetilde H_s(x,G_{t, \widehat{\varrho}})}  \, \mathrm{d}s  \, \mathrm{d}t,
\end{align*}
where $x\in G$, $A\in\mathcal{B}(G)$, which follows due to the representation of $\pi_{0,t,s}$ by 
\[
\pi_{0,t,s}(A)=\frac{\widetilde{\pi}_{0,s}(A \cap G_{t, \widehat{\varrho}})}{\widetilde{\pi}_{0,s}(G_{t, \widehat{\varrho}})}, \quad A\in\mathcal{B}(G).
\]
The transition kernel $\widetilde{H}_{DA}$ is well-defined, if Assumption \ref{assum:positivity_condition} holds and if we have for all $x\in G$, $s\in(0, \rhoapp(x))$ and $t\in (0,\widehat{\varrho}(x))$ that $\widetilde H_s(x,G_{t, \widehat{\varrho}})>0$. For the hybrid samplers considered in this paper, this can be shown under relatively mild conditions.

Another approach to define delayed acceptance hybrid slice sampling is via the representation \eqref{PMLR:DASS:transitionkernel} of $S_{DA}$. 
Since often it is not possible to sample from $\pi_{0,t,s}$ directly, we now assume a family $(H_{t,s})_{\change{(s,t)\in \mathcal T}}$ of $\pi_{0,t,s}$-reversible transition kernels and set
\begin{align}
    \label{PMLR:DAHSS:transitionkernel}
	H_{DA}(x,A)
    := \frac{1}{\varrho(x)}\int_0^{\widehat{\varrho}(x)} \int_0^{\rhoapp(x)} H_{t,s}(x,A)  \, \mathrm{d}s  \, \mathrm{d}t,
\end{align}
where $x\in G$, $A\in\mathcal{B}(G)$.
Algorithmically this transition kernel corresponds to replacing steps 3 to 5 in Algorithm \ref{alg:DASS} by the single step
\begin{algorithmic}[1]
	\setcounter{ALG@line}{3}
	\State Draw sample $y$ of $H_{t,s}(x,\cdot)$.
\end{algorithmic}
This version of delayed acceptance hybrid slice sampling appears to be more promising in terms of computational efficiency, see the following remark.

\begin{remark}[Computational aspects of DA-HSS]
Any common hybrid slice sampling method can be applied to draw approximate samples of $\widetilde{\pi}_{0,s}$ or $\pi_{0,t,s}$, respectively. 
However, in variant $\widetilde{H}_{DA}$ of DA-HSS we apply hybrid slice sampling within a while-loop and, hence, possibly multiple times, whereas in the algorithmic implementation of $H_{DA}$ we apply hybrid slice sampling exactly once per step in order to generate a sample approximately distributed according to $\pi_{0,t,s}$.
In particular, many common hybrid slice samplers involve a shrinkage procedure on one-dimensional sets. 
For $\widetilde{H}_{DA}$ such a shrinkage procedure looks for a suitable $y \in G_{s,\rhoapp}$, and thus, would have to be started for any single try within the while-loop until a $y \in G_{s,\rhoapp} \cap  G_{t,\widehat{\varrho}}$ is found.
On the other hand, for $H_{DA}$ we can apply a single shrinkage procedure: firstly find $y \in G_{s,\rhoapp}$ and then apply the shrinkage a bit further until we obtain $y \in G_{s,\rhoapp} \cap  G_{t,\widehat{\varrho}}$.
Therefore, $H_{DA}$ seems to be the more advantageous variant of DA-HSS compared to $\widetilde{H}_{DA}$.
\end{remark}

The latter arguments let us focus on $H_{DA}$ as general delayed acceptance hybrid slice sampling. The next result provides a statement of reversibility of $H_{DA}$.

\begin{proposition}\label{prop:DAHSS:reversible}
Let Assumption \ref{assum:positivity_condition} be satisfied and let $(H_{t,s})_{\change{(s,t)\in \mathcal T}}$ be a family of $\pi_{0,t,s}$-reversible transition kernels. Then, $H_{DA}$ given as in \eqref{PMLR:DAHSS:transitionkernel} is $\pi$-reversible. If 
the members 
$H_{t,s}$ of the family are only assumed to be $\pi_{0,t,s}$-invariant (instead of $\pi_{0,t,s}$-reversible), then $H_{DA}$ is $\pi$-invariant.   
\end{proposition}
The proof of Proposition \ref{prop:DAHSS:reversible} can be found in Appendix \ref{app:DAHSS:reversible}.
To derive ergodicity statements we can again use Proposition~\ref{Schaer:Convergence}.

\begin{proposition}
    \label{prop:DAHSS:ergodic}
    Let Assumption \ref{assum:positivity_condition} be satisfied and let $(H_{t,s})_{\change{(s,t)\in \mathcal T}}$ be a family of $\pi_{0,t,s}$-invariant transition kernels.
    \begin{enumerate}
        \item \label{statement1}
        If for any $x\in G$ and any $A\in\mathcal{B}(G)$ with $\pi(A)>0$ there exist $T_{x,A}\in(0, \widehat{\varrho}(x)]$ and $S_{x,A}\in(0, \rhoapp(x)]$ such that for almost all $t\in(0,T_{x,A})$ and almost all $s\in(0,S_{x,A})$ we have 
        \[
           H_{t,s}(x,A)>0,
        \]
        then $H_{DA}$ is $\pi$-almost everywhere ergodic.
        \item \label{statement2}
        If in addition to 1. we have for any $x\in G$ and (almost) all $t\in(0,\widehat{\varrho}(x))$ as well as (almost) all $s\in(0, \rhoapp(x))$ that
        \[
        H_{t,s}(x,\cdot) \ll \pi,
        \]    
        then $H_{DA}$ is ergodic.
    \end{enumerate}
%
\end{proposition}
\begin{proof}
   The 
   $\pi$-almost everywhere ergodicity follows by Proposition \ref{Schaer:Convergence} given that 
    for $A\in\mathcal{B}(G)$ with $\pi(A)>0$ we have 
    \begin{align*}
        H_{DA}(x,A) & = \frac{1}{\varrho(x)}\int_0^{\widehat{\varrho}(x)} \int_0^{\rhoapp(x)} H_{t,s}(x,A)  \, \mathrm{d}s  \, \mathrm{d}t\\
        & \geq \frac{1}{\varrho(x)}\int_0^{T_{x,A}} \int_0^{S_{x,A}} H_{t,s}(x,A)  \, \mathrm{d}s  \, \mathrm{d}t>0.
    \end{align*}
    If the additional assumption in \ref{statement2}. is satisfied, then for any $A\in\mathcal{B}(G)$ with $\pi(A)=0$ we also have
    \begin{align*}
        H_{DA}(x,A) & = \frac{1}{\varrho(x)}\int_0^{\widehat{\varrho}(x)} \int_0^{\rhoapp(x)} H_{t,s}(x,A)  \, \mathrm{d}s  \, \mathrm{d}t
        =
        0
    \end{align*}
    implying ergodicity by Proposition \ref{Schaer:Convergence}.
\end{proof}

\begin{remark}
In Proposition \ref{prop:DAHSS:ergodic} we could replace the assumption formulated in \ref{statement2}. also by
\[
    \forall x\in G \ 
    \forall s\in(0, \rhoapp(x)) \ 
    \forall t\in(0,\widehat{\varrho}(x)) \ \colon \
    H_{t,s}(x,\cdot) \ll \pi_{0,t,s}
\]
and the assumptions formulated in \ref{statement1}. by ``\,for all $x\in G$ exists $S_x \in (0, \rhoapp(x)]$ and $T_x \in (0, \widehat{\varrho}(x)]$ such that 
\[
    \forall s\in(0, S_x) \ 
    \forall t\in(0,T_x) \ \colon \ 
    \pi_{0,t,s} \ll H_{t,s}(x,\cdot)\text{"}.
\]
\end{remark}


\noindent We provide delayed acceptance algorithms of common hybrid slice samplers and discuss whether the assumptions of Propositions \ref{prop:DAHSS:reversible} and \ref{prop:DAHSS:ergodic} are satisfied. 

The hybrid slice sampling algorithms we consider follow a common structure. They assume a certain reference measure $\pi_0$, e.g., a Gaussian distribution or the Lebesgue measure, and, given a current state $x\in G$ and a certain measurable subset $B \subset G \subseteq \mathbb{R}^d$ with $\pi_0(B)>0$ they generate approximate samples of $\pi_0$ restricted to $B$ (in the sense of samples drawn from a transition kernel $H(x, \cdot)$ which is invariant w.r.t.~$\pi_0$ restricted to $B$).
In hybrid slice sampling we have $B = G_{t,\varrho}$ and for delayed acceptance hybrid slice sampling $B = G_{s,\rhoapp} \cap G_{t,\widehat{\varrho}}$ where the latter allows for a reduction in computational cost as discussed in the following. 

\subsection{Delayed Acceptance Elliptical Slice Sampling}
\label{sec:DA-ESS}
Elliptical slice sampling (abbreviated as ESS) was introduced by \citet{MurrayEtAl2010} and assumes a Gaussian reference measure $\pi_0=\mathrm{N}(0,C)$ on $G=\mathbb{R}^d$. 
In order to generate approximate samples from $\pi_0$ restricted to a set $B\subset G$ with $\pi_0(B)>0$, ESS takes a sample $v$ from $\pi_0$ and considers the ellipse defined by the function $p: \mathbb{R}^d \times \mathbb{R}^d \times \change{[-\pi,\pi)} 
\rightarrow \mathbb{R}^d$,
\begin{align*}
	p(x, v, \theta) := \cos(\theta) \ x + \sin(\theta)\ v, \quad \theta\in\change{[-\pi,\pi)} 
\end{align*}
where $x\in B$ denotes the current state of the Markov chain.
ESS then generates randomly a point $y = p(x,v,\theta)$ on the ellipse by uniformly sampling $\theta$, and checks if it belongs to $B$ in which case the point $y$ is the output.
To conveniently refer to the angles $\theta$ for which the corresponding point on the ellipse is an element of $B$, we define
\begin{align}\label{eq:theta_B_def}
   \Theta_B(x,v):=\{\theta\in\change{[-\pi,\pi)}: \, p(x, v, \theta)\in B\}.
\end{align}
To increase efficiency, i.e. for finding such a $\theta\in\Theta_B(x,v)$, a shrinkage procedure is applied. 
Specifically, the segments of the ellipse between rejected points, which do not contain the current state $x$, are excluded. New points are then drawn randomly from the remaining part of the ellipse. 
For details on the transition kernel of the ESS and a proof of its $\pi$-reversibility we refer to \cite{HasenpflugEtAl2023}. 
Furthermore, a convergence analysis in terms of geometric ergodicity of ESS can be found in \cite{NatarovskiiEtAl2021}.

For \emph{delayed acceptance elliptical slice sampling} (DA-ESS) the set $B$ is the intersection of two superlevel sets, i.e. $B = G_{s,\rhoapp} \cap G_{t,\widehat{\varrho}}$.
This can be exploited by first checking if the proposed point $y = p(x,v,\theta)$ belongs to $G_{s,\rhoapp}$ and only if this is the case, a checking whether $y\in G_{t,\widehat{\varrho}}$ takes place, since the latter is more costly.
For this we use a delayed acceptance shrinkage procedure, which is generally formulated  in Algorithm~\ref{alg:DA-Shrinkage}, since we also use it in other DA-HSS methods later.

\begin{algorithm}
    \caption{Delayed Acceptance Shrinkage Routine $\mathtt{DA\text{\_}Shrink\change{(V,W,z,l,r)}}$}\label{alg:DA-Shrinkage}
    \renewcommand{\algorithmicrequire}{\textbf{Input:}}
	\renewcommand{\algorithmicensure}{\textbf{Output:}}
    \begin{algorithmic}[1]
        \Require 
        sets $V$ and $W$, \change{current search parameter} $z$, interval bounds $l$ and $r$ 
	    \Ensure updated \change{search parameter} $z$
        
        \While{true}
        \If{$z\in V$}
		\If{$z\in W$}
        \State  \textbf{break}
		\EndIf
		\EndIf
		\State \textbf{do Shrinkage:}
		\State \phantom{....} \textbf{if} $z<0$ \textbf{then} Set $l \leftarrow z$ \textbf{else} Set $r\leftarrow z$
		\State \phantom{....} Draw sample $z$ of $\mathrm{U}(l,r)$
        \EndWhile
    \end{algorithmic}
\end{algorithm}

Using the delayed acceptance shrinkage routine, we can now provide an algorithmic description of DA-ESS in Algorithm \ref{alg:DAESS}.

\begin{algorithm}
	\caption{Delayed Acceptance Elliptical Slice Sampling}\label{alg:DAESS}
	\renewcommand{\algorithmicrequire}{\textbf{Input:}}
	\renewcommand{\algorithmicensure}{\textbf{Output:}}
	\begin{algorithmic}[1]
		\Require current state $X_n = x$
	    \Ensure next state $X_{n+1} = y$
		\State Draw sample $s$ of $\mathrm{U}(0,\rhoapp(x))$ and $t$ of $\mathrm{U}(0,\widehat{\varrho}(x))$
	    \State Draw sample $v$ of $\mathrm{N}(0,C)$ and $\theta$ of \change{$\mathrm{U}(-\pi,\pi)$}
	    \State Set $\theta_{\text{min}}\leftarrow \change{-\pi}$, $\theta_{\text{max}}\leftarrow \change{\pi}$
        \State Set $\theta \leftarrow \mathtt{DA\text{\_}Shrink}(\Theta_{G_{s,\rhoapp}}(x,v), \Theta_{G_{t,\widehat{\varrho}}}(x,v), \theta, \theta_{\text{min}}, \theta_{\text{max}})$
        \State Set $y \leftarrow p(x, v, \theta)$ 
	\end{algorithmic}
\end{algorithm}

\change{Note, that DA-ESS satisfies Assumption \ref{assum:positivity_condition} for lower semi-continuous $\rhoapp$ and $\widehat{\varrho}$. Now,} we can extend the reversibility result for ESS given by \citet{HasenpflugEtAl2023} to DA-ESS and also the \change{assumptions of Proposition \ref{prop:DAHSS:ergodic}} can be verified for DA-ESS.

\begin{theorem}\label{theo:DAESS}
\change{Given Assumption \ref{assum:positivity_condition}} the transition kernel that corresponds to Algorithm~\ref{alg:DAESS}, i.e. the DA-ESS method, is $\pi$-reversible and \change{ergodic}.
\end{theorem}
The proof of Theorem \ref{theo:DAESS} can be found in Appendix \ref{app:DAESS}.


\subsection{Delayed Acceptance Hit-and-Run Slice Sampling with Stepping-Out and Shrinkage}\label{sec:DA-HRUSS}

Hit-and-run uniform slice sampling (abbreviated as HRUSS) is a combination of the classical hit-and-run algorithm, cf. \cite{Smith1984}, with slice sampling. Here $\pi_0$ is assumed to be the $d$-dimensional Lebesgue measure on $G\subseteq \mathbb{R}^d$. 
In order to generate an approximate sample of the uniform distribution on a subset $B\subset G$, where $0 < \pi_0(B) < \infty$, HRUSS first draws a random direction $v$ according to $\mathrm{U}(\mathbb{S}_{d-1})$ 
and then, with $x\in B$, generates an approximate sample of the uniform distribution on the one-dimensional \change{line section}
\begin{align}\label{eq:line_section_def}
    L_{B}(x,v)=\{p \in \mathbb R\colon \ x+p\, v \in B \}.
\end{align}
Since direct sampling of $\mathrm U(L_{B}(x,v))$ is often not efficiently implementable, particularly, when $B$ represents superlevel sets, we propose, as in \cite{LatuszynskiRudolf2024}, to apply the stepping-out and shrinkage method developed in \cite{Neal2003}. 
To this end, an initial interval $[\ell, r]$ around $0$ is gradually enlarged by subtracting, respectively adding, multiples of a step size $w>0$ to $\ell$ and $r$, respectively, until $\ell,r \notin L_{B}(x,v)$. Then, candidates $y=x + p v$ are generated by taking samples $p$ of $\mathrm U(\ell,r)$. If $y \in B$, then $y$ is the output, otherwise the interval $[\ell,r]$ is shrunk and new candidates are generated until the algorithm found a valid candidate $y \in B$.
For further details such as the formulation of the transition kernel, the proof of $\pi$-reversibility and a convergence analysis in terms of geometric ergodicity we refer to \cite{LatuszynskiRudolf2024}.    

For a transition of \emph{delayed acceptance hit-and-run uniform slice sampling} (DA-HRUSS), see Algorithm \ref{alg:DAHRSS}. We again exploit the structure of $B = G_{s,\rhoapp} \cap G_{t,\widehat{\varrho}}$ to reduce the computational cost: the stepping-out procedure is carried out w.r.t. 
\[
    L_{s,\rhoapp}(x,v) := L_{G_{s,\rhoapp}}(x,v) 
    \supset L_{G_{s,\rhoapp}\cap G_{t,\widehat{\varrho}}}(x,v)
\]
and during the shrinkage procedure we first shrink $(\ell,r)$ until we find a $y = x + p v \in G_{s,\rhoapp}$ and then, if $y \notin G_{s,\rhoapp}\cap G_{t,\widehat{\varrho}}$, apply shrinkage further until a valid point $y \in G_{s,\rhoapp}\cap G_{t,\widehat{\varrho}}$ is found. 
This exploits the evaluation of the cheap approximate $\rhoapp$ as much as possible in our setting, before a first evaluation of the costly $\widehat{\varrho}$ is required.

\begin{algorithm}
   \caption{Delayed Acceptance Hit-And-Run Slice Sampling}\label{alg:DAHRSS}
   \renewcommand{\algorithmicrequire}{\textbf{Input:}}
   \renewcommand{\algorithmicensure}{\textbf{Output:}}
   \begin{algorithmic}[1]
		\Require current state $X_n = x$, step size parameter $w>0$
	    \Ensure next state $X_{n+1} = y$
        \State Draw sample $s$ of $\mathrm{U}(0,\rhoapp(x))$, $t$ of $\mathrm{U}(0,\widehat{\varrho}(x))$
        \State Draw sample $v$ of $\mathrm{U}(\mathbb{S}_{d-1})$ and $u$ of $\mathrm{U}(0,1)$ 
        \State Set $\ell \leftarrow -uw$ and $r \leftarrow \ell+w$
        \While{$\ell \in L_{s,\rhoapp}(x,v)$} 
        \State Set $\ell \leftarrow \ell-w$
		\EndWhile
        \While{$r \in L_{s,\rhoapp}(x,v)$} 
        \State Set $r \leftarrow r+w$
		\EndWhile
        \State Draw sample $p$ of $\mathrm{U}(\ell,r)$
        \State Set $p \leftarrow \mathtt{DA\text{\_}Shrink}(L_{G_{s,\rhoapp}}(x,v), L_{G_{t,\widehat{\varrho}}}(x,v), p, \ell, r)$
        \State Set $y \leftarrow x+p\,v$
	\end{algorithmic}
\end{algorithm}

We further formulate the following Assumption.

\begin{assumption}\label{assum:bimodal}
   Let $\rhoapp$ be unimodal, i.e. for all $s\in(0,\|\rhoapp\|_\infty)$ the corresponding superlevel set $G_{s,\rhoapp}$ is connected and $\widehat{\varrho}$ be either unimodal, i.e. for all $t\in(0,\|\widehat{\varrho}\|_\infty)$ the corresponding superlevel set $G_{t,\widehat{\varrho}}$ is connected, or bimodal, i.e. there exist levels $t_1, t_2\in(0,\|\widehat{\varrho}\|_\infty)$ with $t_1\leq t_2$, such that
   \begin{enumerate}
       
       \item for all $t\in(0,\|\widehat{\varrho}\|_\infty)\setminus[t_1,t_2)$: \, $G_{t,\widehat{\varrho}}$ is connected and   

       \item for all $t\in [t_1,t_2)$: \, $G_{t,\widehat{\varrho}}=G^{(1)}_{t,\widehat{\varrho}}\cup G^{(2)}_{t,\widehat{\varrho}}$, where $G^{(1)}_{t,\widehat{\varrho}}$ and $G^{(2)}_{t,\widehat{\varrho}}$ are disjoint, have positive Lebesgue measure and $G^{(i)}_{t,\widehat{\varrho}}\subseteq G^{(i)}_{\tilde{t},\widehat{\varrho}}$ for $t_1\leq\tilde{t}\leq t<t_2$, $i=1,2$, holds.     
       
   \end{enumerate}
   \change{In addition, for all $s\in(0,\|\rhoapp\|_\infty)$, the superlevel set $G_{s,\rhoapp}$ is bounded and convex. If $\widehat{\varrho}$ is unimodal, then $G_{t,\widehat{\varrho}}$ is convex for all $t\in(0,\|\widehat{\varrho}\|_\infty)$. If $\widehat{\varrho}$ is bimodal, then $G_{t,\widehat{\varrho}}$ is convex for all $t\in(0,\|\widehat{\varrho}\|_\infty)\setminus[t_1,t_2)$, and $G^{(1)}_{t,\widehat{\varrho}}$ and $G^{(2)}_{t,\widehat{\varrho}}$ are convex for all $t\in[t_1,t_2)$.}
\end{assumption}

\change{
Assumption~\ref{assum:bimodal} is satisfied in the unimodal case, for example, for lower semi-continuous log-concave $\rhoapp$ with bounded superlevel sets and lower semi-continuous log-concave $\widehat{\varrho}$.}
\change{Assumption~\ref{assum:bimodal} guarantees a simple geometric structure of the one-dimensional sets encountered by DA-HRUSS. During each iteration, the algorithm first constructs the interval $L_{s,\rhoapp}(x,v)$ by the stepping-out procedure. 
The convexity of $G_{s,\rhoapp}$ ensures that this interval is connected. Subsequently, the delayed acceptance shrinkage routine searches for a point in the intersection $L_{G_{s,\rhoapp}\cap G_{t,\widehat{\varrho}}}(x,v)$ which, by the assumption, consists of at most two disjoint intervals. This geometry preserves among other things reversiblity and ergodicity of the underlying Markov chain, which we can see in the following statement.}


\begin{theorem}\label{theo:DAHRUSS}
\change{Let Assumption \ref{assum:positivity_condition} and Assumption~\ref{assum:bimodal} be satisfied.} Then, the transition kernel that corresponds to Algorithm~\ref{alg:DAHRSS}, i.e. the DA-HRUSS method, is $\pi$-reversible and ergodic.
\end{theorem}
The proof of Theorem \ref{theo:DAHRUSS} is stated in Appendix \ref{app:DAHRUSS}. 
Note that the choice of the step size parameter $w>0$ in DA‑HRUSS does not affect ergodicity, provided the assumptions of Theorem \ref{theo:DAHRUSS} are satisfied. 

\subsection{Delayed Acceptance Gibbsian Polar Slice Sampling}\label{sec:DA-GPSS}

Polar slice sampling relies on the reference measure $\pi_0(\mathrm dx) = \|x\|^{1-d} \, \mathrm dx$ on $G \subseteq \mathbb R^d$ and exploits polar coordinates for $y\in G$, i.e. $y = r\,v$ where $v \in \mathbb{S}_{d-1}$ and $r \in [0, \infty)$. 
In order to generate an approximate sample of $\pi_0$ restricted to a subset $B\in\mathcal B(G)$ with $0 < \pi_0(B) < \infty$ we need to sample a suitable $v \in \mathbb{S}_{d-1}$ and $r \in [0, \infty)$ such that $r\,v \in B$.
In Gibbsian polar slice sampling this is done within two separate steps---one for sampling the direction $v$ and one for sampling the radius $r$. 
To this end, we assume $x = r_0\,v_0 \in B$ to be the current state of the Markov chain in polar coordinates.  

For the direction update in GPSS, we first draw a point $v_{\perp}$ according to the uniform distribution on 
\begin{align*}
   \mathbb{S}^{v_0}_{d-2}:=\{v_{\perp} \in\mathbb{S}_{d-1}: v_0^\top v_{\perp} = 0 \}
\end{align*}
and then seek for a next direction
\[
    v 
    = p(v_0, v_\perp, \theta)
    = \cos(\theta) \ v_0 + \sin(\theta)\ v_\perp, \quad \change{\theta\in[-\pi,\pi)}
\]
by finding a $\theta$ of
\[
   \Theta_B(v_0,v_\perp,r_0):=\{\theta\in\change{[-\pi,\pi)}: \, r_0\cdot p(v_0, v_\perp, \theta)\in B\}.     
\]
To this end, we apply the same shrinkage procedure as for elliptical slice sampling.

For the radius update, we draw a sample $r$ according to the uniform distribution on the segment 
\begin{align}\label{segment:gpss}
    L_{B}(v):=\{r \in [0, \infty)\colon \ r\, v \in B \}.
\end{align}
To this end, Neal’s stepping-out and shrinkage procedure \cite{Neal2003} is applied, cf. Section~\ref{sec:DA-HRUSS}.

For the \emph{delayed acceptance Gibbsian polar slice sampling} (DA-GPSS) with $B = G_{s,\rhoapp} \cap G_{t,\widehat{\varrho}}$ we can apply the same cost reduction strategies for both shrinkage procedures as already explained in Section~\ref{sec:DA-ESS} and Section~\ref{sec:DA-HRUSS}. 
A single transition of DA-GPSS is presented in Algorithm \ref{alg:DAGPSS}.

\begin{algorithm}
	\caption{Delayed Acceptance Gibbsian Polar Slice Sampling}\label{alg:DAGPSS}
	\renewcommand{\algorithmicrequire}{\textbf{Input:}}
	\renewcommand{\algorithmicensure}{\textbf{Output:}}
	\begin{algorithmic}[1]
		\Require current state $X_n = x$, where $x = r_0\, v_0$ with $r_{0}\in[0,\infty)$ and $v_{0} \in \mathbb{S}_{d-1}$, initial interval length $w>0$ 
		\Ensure next state $X_{n+1} = y$
		\State Draw sample $s$ of $\mathrm{U}(0,\rhoapp(x))$ and $t$ of $\mathrm{U}(0,\widehat{\varrho}(x))$
        \State Draw sample $v_\perp$ of $\mathrm{U}(\mathbb{S}^{v_0}_{d-2})$, $\theta$ of $\change{\mathrm{U}(-\pi,\pi)}$ and $u$ of $\mathrm{U}(0,1)$  
		\State Set $\theta_{\text{min}} \leftarrow \change{-\pi}$, $\theta_{\text{max}}\leftarrow \change{\pi}$
        \State Set $\theta \leftarrow \mathtt{DA\text{\_}Shrink}(\Theta_{G_{s,\rhoapp}}(v_0,v_\perp, r_0), \Theta_{G_{t,\widehat{\varrho}}}(v_0,v_\perp, r_0), \theta, \theta_{\text{min}}, \theta_{\text{max}})$
        \State Set $v \leftarrow p(v_0, v_\perp, \theta)$
        \State Set $r_\text{min} \leftarrow \text{max}(r_0-u\cdot w, 0)$ and $r_\text{max} \leftarrow r_0+(1-u)\cdot w$
        \While{$r_\text{min}>0$ and $r_\text{min}v \in G_{s,\varrho_{\text{app}}}$} 
        \State Set $r_\text{min} \leftarrow \text{max}(r_\text{min} - w, 0)$
		\EndWhile
        \While{$r_\text{max}v \in G_{s,\varrho_{\text{app}}}$}
        \State Set $r_\text{max} \leftarrow  r_\text{max} + w$
        \EndWhile
        \State Draw sample $r$ of $\mathrm{U}(r_\text{min},r_\text{max})$
        \State Set $r \leftarrow \mathtt{DA\text{\_}Shrink}(L_{G_{s,\rhoapp}}(v), L_{G_{t,\widehat{\varrho}}}(v), r, r_\text{min}, r_\text{max})$
		\State Set $y \leftarrow r \, v$
	\end{algorithmic}
\end{algorithm}


We further assume the following:
\begin{assumption}\label{assum:radius:update}
The update of the radius, i.e. calling the shrinkage procedure in line~14 of Algorithm~~\ref{alg:DAGPSS}, is always realized according to 
$\mathrm{U}(L_{G_{s,\rhoapp} \cap G_{t,\widehat{\varrho}}}(v))$.   
\end{assumption}

\change{
Assumption~\ref{assum:radius:update} is satisfied, for example, if $\rhoapp$ and $\widehat{\varrho}$ are lower semi-continuous and}
unimodal along rays, i.e. for every $v \in \mathbb{S}_{d-1}$, every $s\in(0,\|\rhoapp\|_\infty)$ and every $t\in(0,\|\widehat{\varrho}\|_\infty)$ the sets $L_{G_{s,\varrho_{\text{app}}}}(v)$ and $L_{G_{t,\widehat{\varrho}}}(v)$ are both intervals in $[0,\infty)$ and, hence, also $L_{G_{s,\varrho_{\text{app}}}\cap G_{t,\widehat{\varrho}}}(v)$ is an interval, i.e., a connected set in $[0,\infty)$.
\change{By Assumption \ref{assum:radius:update}, the admissible radii form a connected interval along every fixed direction $v$. This yields a particularly simple transition kernel for the radius update, which provides one of the ingredients in the following well-definedness result, whose proof is provided in Appendix \ref{app:DAGPSS}.}

\begin{theorem}\label{theo:DAGPSS}
\change{Let Assumption \ref{assum:positivity_condition} and Assumption~\ref{assum:radius:update} be satisfied.} Then, the transition kernel that corresponds to Algorithm~\ref{alg:DAGPSS}, i.e. the DA-GPSS method, is $\pi$-invariant and \change{ergodic}.
\end{theorem}




\section{Numerical Experiments} \label{sec:num}

\rev{We consider \change{three} test scenarios for Bayesian inference on $G = \mathbb R^d$ with Gaussian prior $\pi_0=\mathrm{N}(0,C)$. In each scenario, we apply the described hybrid slice sampling methods and their delayed acceptance variants to estimate the posterior expectation $\pi(f)$ of a given quantity of interest $f\colon G \to \mathbb R$. To this end, we consider posterior distributions of the form
\begin{align}\label{num:exp:posterior}
	\pi_h(\text{d} x)
	& \propto
    \varrho_h(x)\, \pi_0(\mathrm{d}x),
\end{align}
where $h\in [0,1]$ denotes a discretization or approximation parameter. We then choose $\pi = \pi_{h_\text{ref}}$ for a suitable $h_\text{ref}\in[0,1]$ with $\varrho = \varrho_{h_\text{ref}}$ to be the ``true" or ``fine" target as defined in (\ref{PMLR:distribution}) and consider $\piapp=\pi_h$, $h\in[0,1]$, for various $h>h_\text{ref}$, with approximations $\rhoapp = \varrho_{h}$ of the density $\varrho$ to be the ``approximations" or ``coarser" targets as given in (\ref{PMLR:distribution:factorization}).  
We show that DA-HSS can benefit from reduced costs, measured by the \change{\emph{total computational time}} $\mathrm{t_{{total}}}$ (in seconds) required by each algorithm to generate $n$ iterates of a Markov chain targeting $\pi$. However, in addition to costs, accuracy also matters. For quantifying accuracy, we use the \change{\emph{effective sample size}} $\mathrm{n_{eff}}$ of the Markov chains underlying the path average estimator $A_n(f)$ in (\ref{path}) which is given by
\begin{align*}
  \mathrm{n_{eff}}
    = n \cdot \Bigg(1+2\sum_{j=1}^{\infty}\text{Corr}(f(X_1),f(X_{1+j}))\Bigg)^{-1}
\end{align*}
and represents roughly speaking the number of i.i.d. samples from $\pi$ that would lead to a Monte Carlo estimator with the same variance as $A_n(f)$. Finally, combining cost and accuracy leads to the following \change{\emph{relative efficiency measure}} 
\begin{align}\label{rel:efficiency}
   \frac{\mathrm{n}^{\text{DA-HSS}}_{\mathrm{eff}}(h)}{\mathrm{n}^{\text{HSS}}_{\mathrm{eff}}(h_{\text{ref}})} \times \frac{\mathrm{t}_{\mathrm{total}}^{\text{HSS}}(h_{\text{ref}})}{\mathrm{t}_{\mathrm{total}}^{\text{DA-HSS}}(h)},    
\end{align}
where we denote by $\mathrm{t}_{\mathrm{total}}^{\text{HSS}}(h_{\text{ref}})$ and $\mathrm{t}_{\mathrm{total}}^{\text{DA-HSS}}(h)$ the total computational time of HSS and DA-HSS, as well as by $\mathrm{n}^{\text{HSS}}_{\mathrm{eff}}(h_{\text{ref}})$ and $\mathrm{n}^{\text{DA-HSS}}_{\mathrm{eff}}(h)$ the effective sample size of the Markov chains generated by HSS and DA-HSS. The criterion (\ref{rel:efficiency}) quantifies by which factor DA‑HSS generates more (or fewer) effective independent samples per second compared to its plain version.
}





\subsection{\change{Bayesian inference for a univariate boundary value problem}}\label{sec:BIP1}
Given noisy observation $\delta_j = q(0.25j, x)+\varepsilon_j$, $\varepsilon_j \sim \mathrm{N}(0,\sigma^2)$ i.i.d., $j=1,2,3$, of the solution $q\colon[0,1]\times \mathbb R^d\to\mathbb R$ of 
\begin{align}\label{PMLR:PDE}
	- \frac{\text{d}}{\text{d} \tau} \left(\mathrm{e}^{u(\tau, x)} \frac{\text{d}}{\text{d} \tau} q(\tau, x)\right) = 0, 
	\quad q(0, x)=0, \ q(1, x)=2,	
\end{align}
we infer the coefficient function $u\colon [0,1]\times \mathbb R^d\to\mathbb R$ which is given as 
\begin{align*}
   u(\tau, x)=\frac{\sqrt{2}}{\pi} \ \sum_{k=1}^{d} x_k \, \sin(k\pi \tau). 
\end{align*}
In particular, we infer the coefficients $x=(x_1,\ldots,x_d)\in\mathbb R^d$ in the former expansion.
We take a Gaussian prior $\pi_0 = \mathrm{N}(0, C)$ with $C=\textmd{diag}\{k^{-2}: k=1,\dots,d\}$ which relates to a Brownian bridge prior for $u$ itself.
The prior $\pi_0$ is then conditioned on observing $\delta=F(x)+\varepsilon$, where $F: \mathbb{R}^d \rightarrow \mathbb{R}^3$ denotes the forward mapping
\begin{align*}
	F \colon x \mapsto u(\cdot, x) \mapsto q(\cdot, x) \mapsto \left(q(0.25j, x\right))_{j=1,\ldots,3},
\end{align*}
$\delta=(\delta_1,\delta_2,\delta_3) \in \mathbb R^3$ the observational data, and $\varepsilon=(\varepsilon_1, \varepsilon_2, \varepsilon_3) \sim \mathrm{N}(0,\sigma^2 I_3)$ the measurement noise.
The solution $q$ of \eqref{PMLR:PDE} is given by
\begin{align*}
	q(\tau, x)= 2\frac{S_{\tau}(e^{-u(\cdot, x)})}{S_1(e^{-u(\cdot, x)})}, \quad S_{\tau}(e^{-u(\cdot, x)}) :=\int\limits_0^{\tau} e^{-u(t, x)}\, \mathrm{d}t	
\end{align*}
and requires numerical integration for computing $S_{\tau}$.
To this end, we apply the trapezoidal rule based on uniform grid points $\tau_i = i\,h$, $i=0,\ldots, 1/h \in \mathbb N$.
Let us denote by $F_h\colon \mathbb R^d \to \mathbb R^3$ the mapping $x \mapsto \left(q(0.25j, x)\right)_{j=1,\ldots,3}$ which employs integration by trapezoidal rule with grid size $h$ for computing $q(0.25j, x)$ given $u$ or $x$, respectively. As a resulting posterior distribution we get (\ref{num:exp:posterior})
with
\[
\varrho_h(x) = \exp\left(- \lVert \delta - F_h(x) \rVert^2 /2\sigma^2\right).
\]
We consider $h_{\text{ref}}=2^{-11}$ as sufficiently fine and choose $\pi = \pi_{h_\text{ref}}$ to be the ``true" target. In inverse problems involving differential equations in the forward map as in our numerical example, natural candidates for $\rhoapp$ are based on coarser numerical solutions to the differential equation. Therefore, as approximations to $\varrho = \varrho_{h_\text{ref}}$, we choose 
$\rhoapp = \varrho_{h}$ for $h\in\{2^{-10},\dots,2^{-2}\}$. Here, we apply ESS and its delayed acceptance version since they are designed for Gaussian priors, to estimate the posterior expectation $\pi(f)$ of $f(x):=\int_{0}^{1}e^{u(\tau,x)} \textmd{d}\tau$. We set $d=100$ and $\sigma^2 = 0.01$ and run each algorithm for $n=2.5\cdot 10^6$ iterations after a burn-in of $n_0=10^5$ iterations, \change{repeating the experiment over 20 independent runs}. In Figure \ref{fig:BIP} we report several results regarding the performance of the DA-ESS based on $\rhoapp = \varrho_{h}$ for varying $h\in\{2^{-10},\dots,2^{-2}\}$. 

\begin{figure*}[t]
    \centering
    \begin{minipage}{0.32\textwidth}
        \centering
        \includegraphics[width=\linewidth, height=3.9cm]{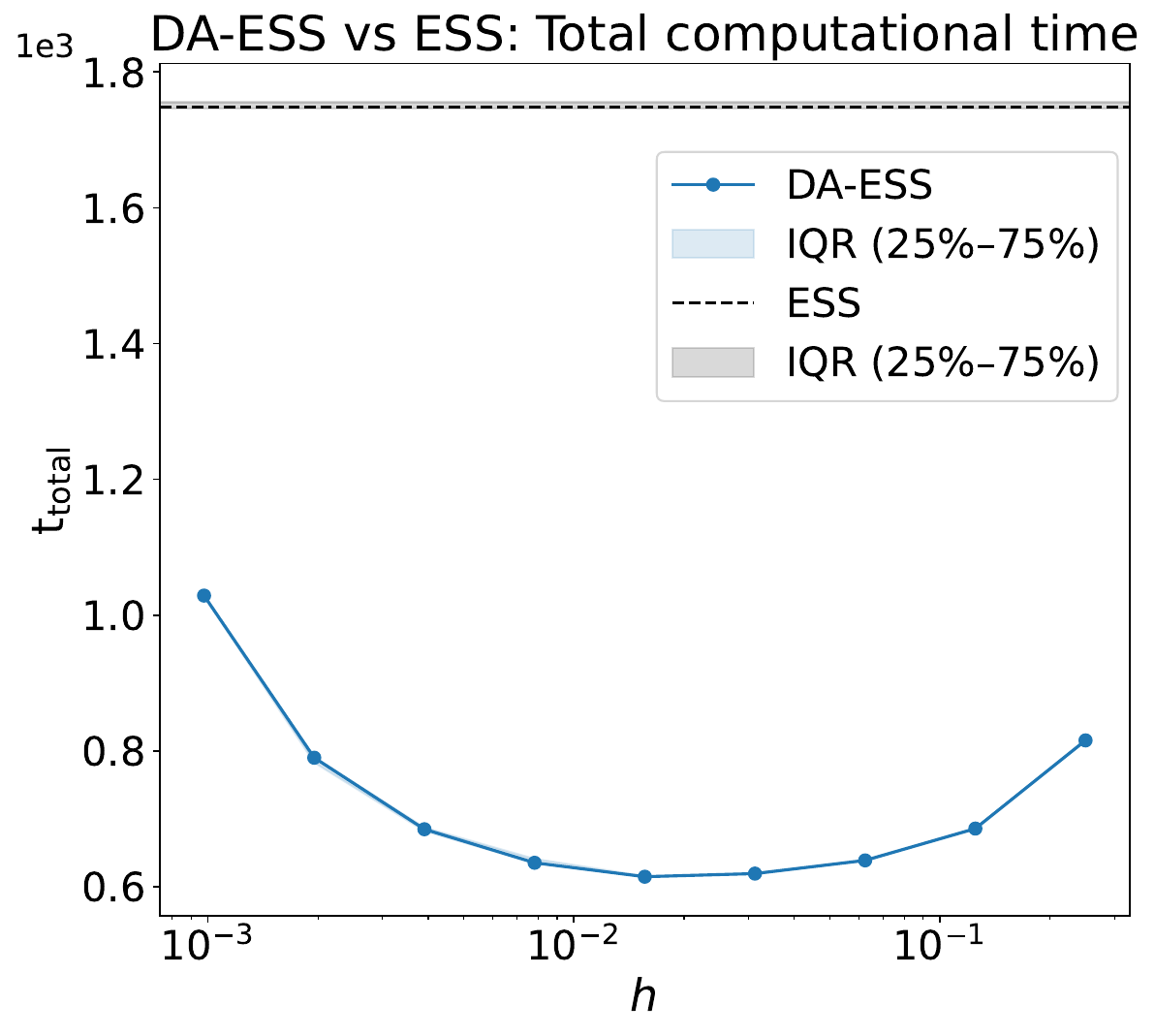}
    \end{minipage}
    \hfill
    \begin{minipage}{0.32\textwidth}
        \centering
        \includegraphics[width=\linewidth]{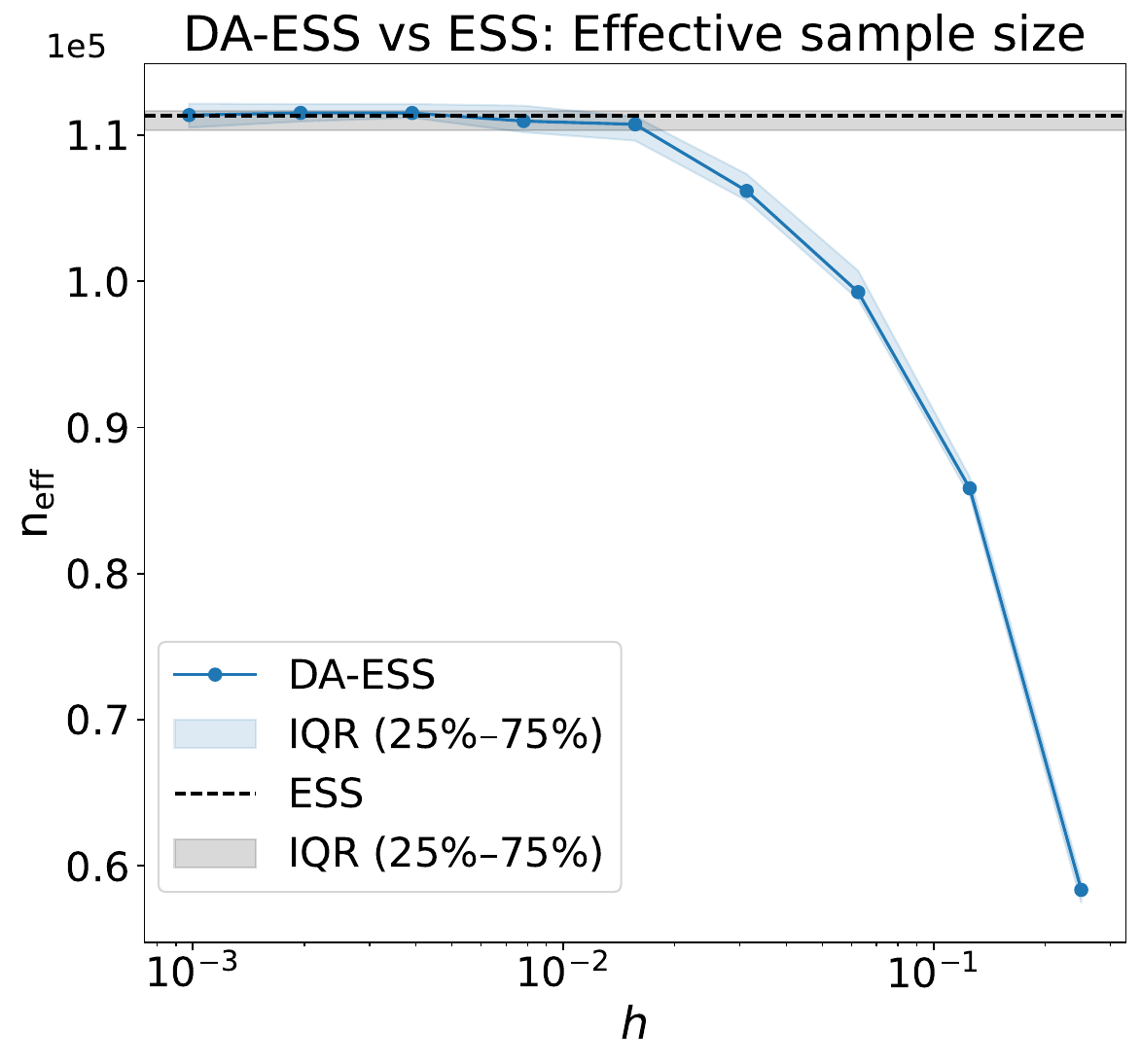}
    \end{minipage}
    \hfill
    \begin{minipage}{0.32\textwidth}
       \centering
        \includegraphics[width=\linewidth]{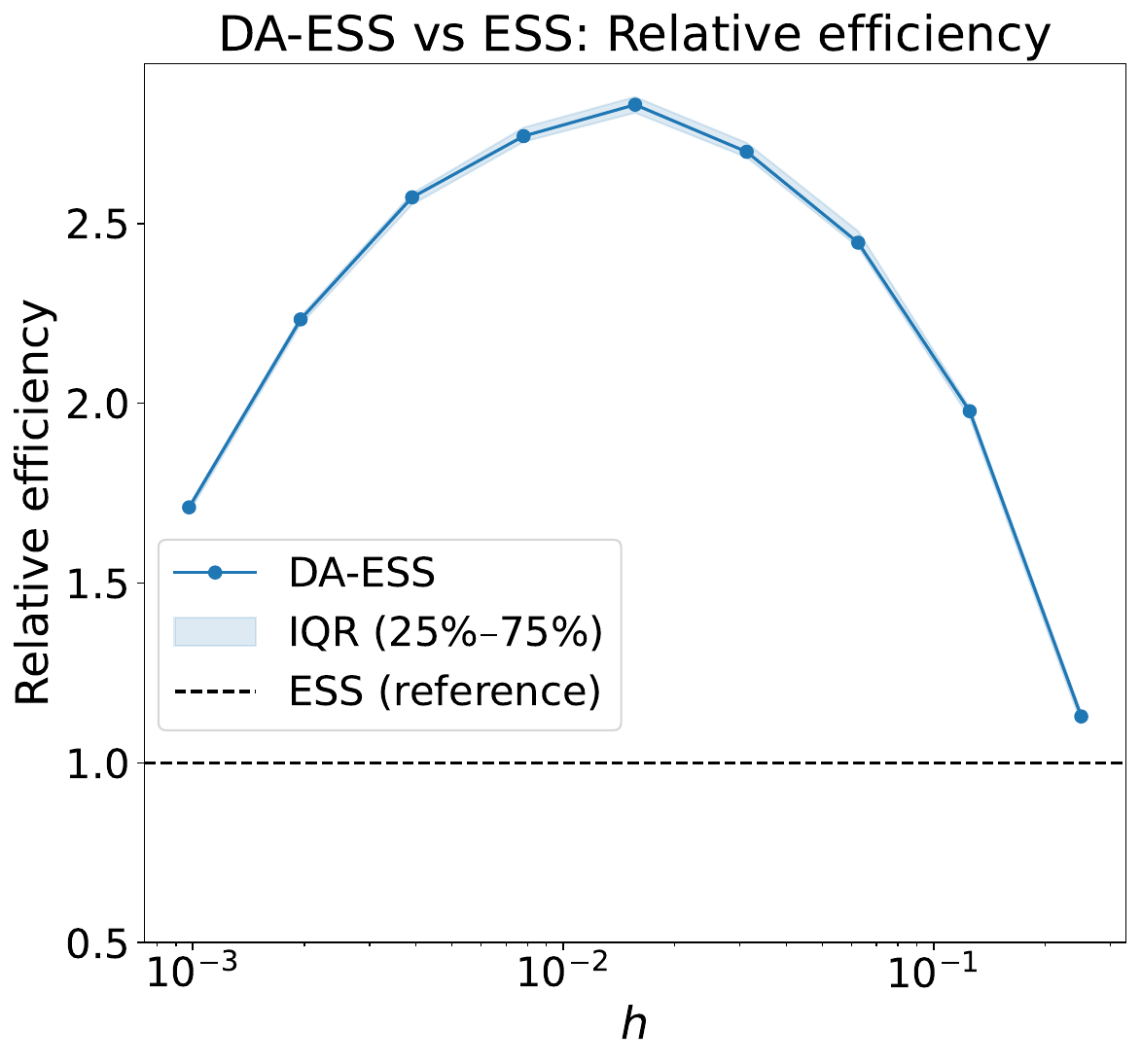}
    \end{minipage}
    \caption{\change{Median of the total computation time (left), median of the effective sample size for MCMC integration of $f$ (middle), and resulting median of the relative efficiency \eqref{rel:efficiency} (right) of DA-ESS for the Bayesian inverse problem of Subsection \ref{sec:BIP1} for 20 independent runs. The shaded regions indicate the corresponding interquartile ranges.}}
    \label{fig:BIP}
\end{figure*}






\change{The lines show the medians over the 20 independent runs, while the shaded regions indicate the corresponding interquartile ranges.} The left panel of Figure \ref{fig:BIP} shows the \change{median of the} total computational time $\mathrm{t}_{\mathrm{total}}^{\text{ESS}}(h_{\text{ref}})$ and $\mathrm{t}_{\mathrm{total}}^{\text{DA-ESS}}(h)$ of ESS and DA-ESS, respectively. We observe that \change{the median of} $\mathrm{t}_{\mathrm{total}}^{\text{DA-ESS}}(h)$ is \change{remarkably} lower for \change{all approximation choices}, where we have minimal total computational time at \change{$h=2^{-6}$}, as opposed to \change{the median of} $\mathrm{t}_{\mathrm{total}}^{\text{ESS}}(h_{\text{ref}})$. 
We show in the \change{middle} panel of Figure \ref{fig:BIP} \change{the median of the} estimates of the effective sample size of the Markov chains generated by ESS and DA-ESS, which we will denote by $n^{\text{ESS}}_{\mathrm{eff}}(h_{\text{ref}})$ and $n^{\text{DA-ESS}}_{\mathrm{eff}}(h)$, respectively. We notice that  $n^{\text{DA-ESS}}_{\mathrm{eff}}(h)$ tends to $n^{\text{ESS}}_{\mathrm{eff}}(h_{\text{ref}})$ for $h\to h_{\text{ref}}$. In the right panel of Figure \ref{fig:BIP} we plot the \change{resulting} relative efficiency criterion (\ref{rel:efficiency}). We observe that DA-ESS is up to \change{2.80} times more efficient than ESS in terms of producing effective samples per second for $h\in\{2^{-10},\dots,2^{-2}\}$. 


\change{
\subsection{Bayesian inference for a partial differential equation}\label{sec:BIP2}

We next consider a two-dimensional analogue of the Bayesian inverse problem from Subsection~\ref{sec:BIP1}, based on the elliptic partial differential equation
\begin{align*}
    -\nabla \cdot \left(
    e^{u(\tau,x)} \nabla q(\tau,x)
    \right) = 0,
    \qquad \tau=(\tau_1,\tau_2)\in D:=[0,1]^2.
\end{align*}
We impose the Dirichlet boundary conditions on the lefthand and righthand side boundaries together with homogeneous Neumann boundary conditions on the top and bottom boundaries, i.e.,
\begin{align*}
    q(0,\tau_2,x) &= 0, & q(1,\tau_2,x) &= 1,
    &\frac{\partial q}{\partial n}(\tau_1,0,x)
    & =
    \frac{\partial q}{\partial n}(\tau_1,1,x)
    =0.
\end{align*}
We parametrize $u(\tau,x)$ by the truncated expansion of the form
\begin{align*}
    u(\tau,x)
    =
    2\sum_{k=1}^d
    x_k(m_k^2+n_k^2)^{-1}\,
    \sin(m_k\pi\tau_1)\sin(n_k\pi\tau_2),
\end{align*}
where $(m_k, n_k)$ are the $d$ frequency pairs in $\mathbb N^2$ with the smalles values of $m^2+n^2$, ordered increasingly. Again, we infer the coefficients $x = (x_1,\dots,x_d) \in \mathbb{R}^d$. Here, we take the Gaussian prior $\pi_0=\mathrm{N}(0,I_d)$. Analogously to Subsection~\ref{sec:BIP1}, we consider synthetic noisy pointwise observations of the form $\delta_j=q_{h_{\mathrm{ref}}}(\tau_{(j)},x^\dagger)+\epsilon_j, \ \epsilon_j\sim \mathcal N(0,\sigma^2) \ \mathrm{i.i.d.},  \ j=1,\ldots,J$ at $J=13$ fixed observation locations $\tau_{(j)} \in D$, where $x^\dagger$ denotes a fixed realization from the prior and $\sigma^2=0.02^2$. In contrast to the one-dimensional problem, the solution $q$ is now approximated by a conservative finite-difference discretization on uniform grids with mesh width $h_\ell = 2^{-\ell}, \ \ell=2,\ldots,7$. Thus, each spatial direction contains $N_\ell = h_\ell^{-1}+1 = 2^\ell+1$ grid points. The numerical solution is evaluated at the observation locations by bilinear interpolation. We denote the resulting forward mapping by $F_h:\mathbb R^d\to\mathbb R^{13}$. For a given mesh width $h$, the corresponding likelihood is then again given by
\begin{align*}
   \varrho_h(x) = \exp\left(- \lVert \delta - F_h(x) \rVert^2 /2\sigma^2\right).
\end{align*}
As in Subsection~\ref{sec:BIP1}, we use a fine numerical discretization for the target and coarser discretizations of the forward problem to construct the deterministic approximations. More precisely, the fine likelihood $\varrho=\varrho_{h_{\mathrm{ref}}}$ is obtained using $h_{\mathrm{ref}}=2^{-8}$ corresponding to $257\times257$ grid points, whereas the likelihoods
$\varrho_h$, with $h\in \left\{2^{-2},2^{-3},\dots,2^{-7}\right\}$ are the deterministic approximations. 

As quantity of interest, we consider the integrated flux in the
$\tau_1$-direction at the left boundary,
\begin{align*}
    f(x)=-\int_0^1 e^{u(0,\tau_2,x)} \frac{\partial q}{\partial\tau_1}(0,\tau_2,x) \,\mathrm d\tau_2.
\end{align*}
For its numerical evaluation, we use the reference mesh width $h_{\mathrm{ref}}=2^{-8}$. Again, we set $d=100$ and apply ESS and its delayed acceptance version to estimate the expectation $\pi(f)$. We run each algorithm for $n=2\cdot 10^4$ iterations after a burn-in of $n_0=5\cdot 10^3$ iterations and repeat the experiment over 10 independent runs. Figure \ref{fig:BIP2} reports the results analogously to Subsection~\ref{sec:BIP1}. Again, the lines show the median over the 10 independent runs, while shaded regions represent the corresponding interquartile ranges. The median of the total computational time (left panel), the median of the effective sample size (middle panel) and the median of the relative efficiency (right panel) show behavior similar to that observed for the one-dimensional inverse problem. DA-ESS achieves a maximum median relative efficiency of $6.40$ and outperforms ESS for all but one of the considered approximations, for which its efficiency is lower than that of ESS.   

\begin{figure*}[t]
    \centering
    \begin{minipage}{0.32\textwidth}
        \centering
        \includegraphics[width=\linewidth, height=3.9cm]{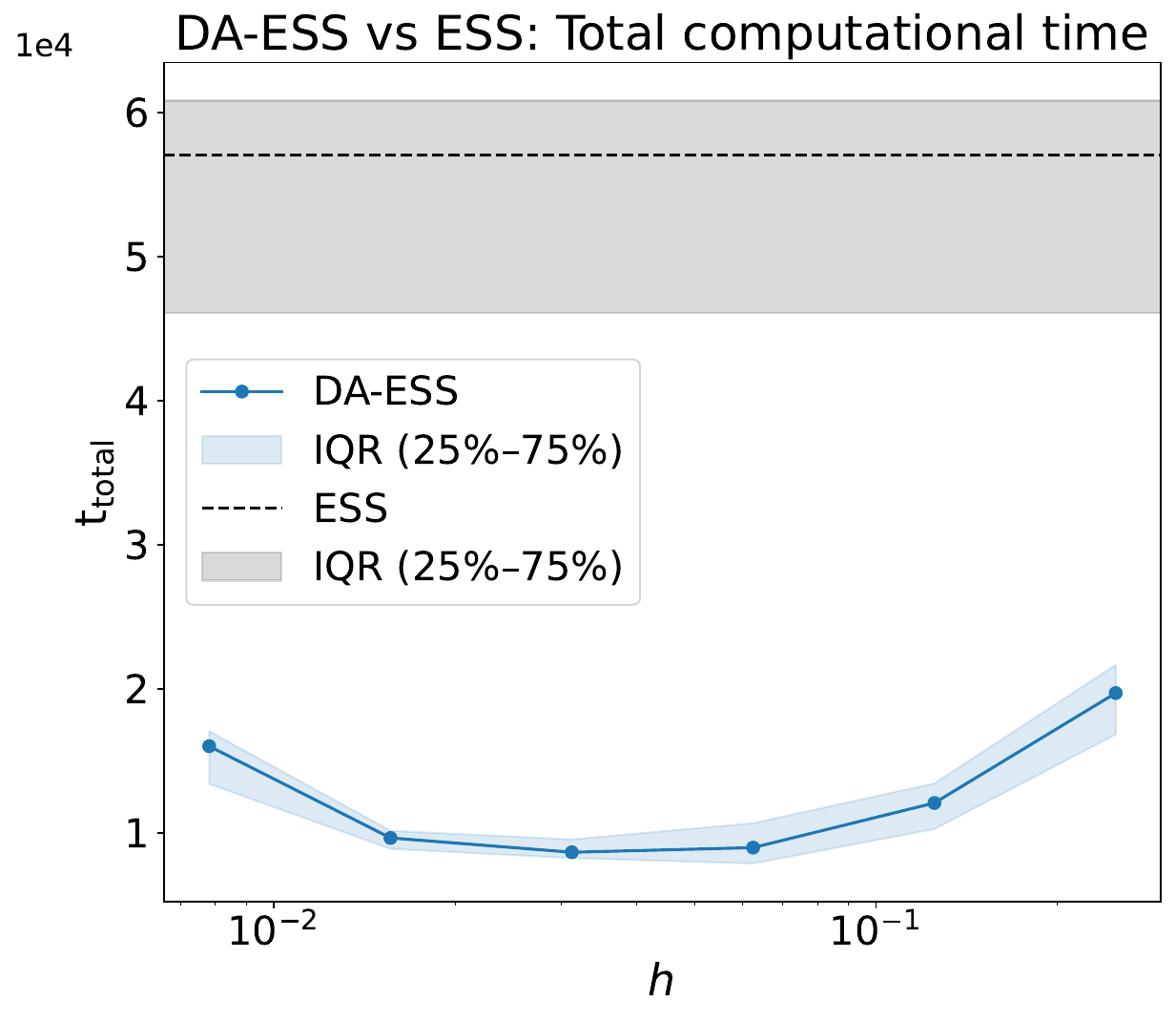}
    \end{minipage}
    \hfill
    \begin{minipage}{0.32\textwidth}
        \centering
        \includegraphics[width=\linewidth]{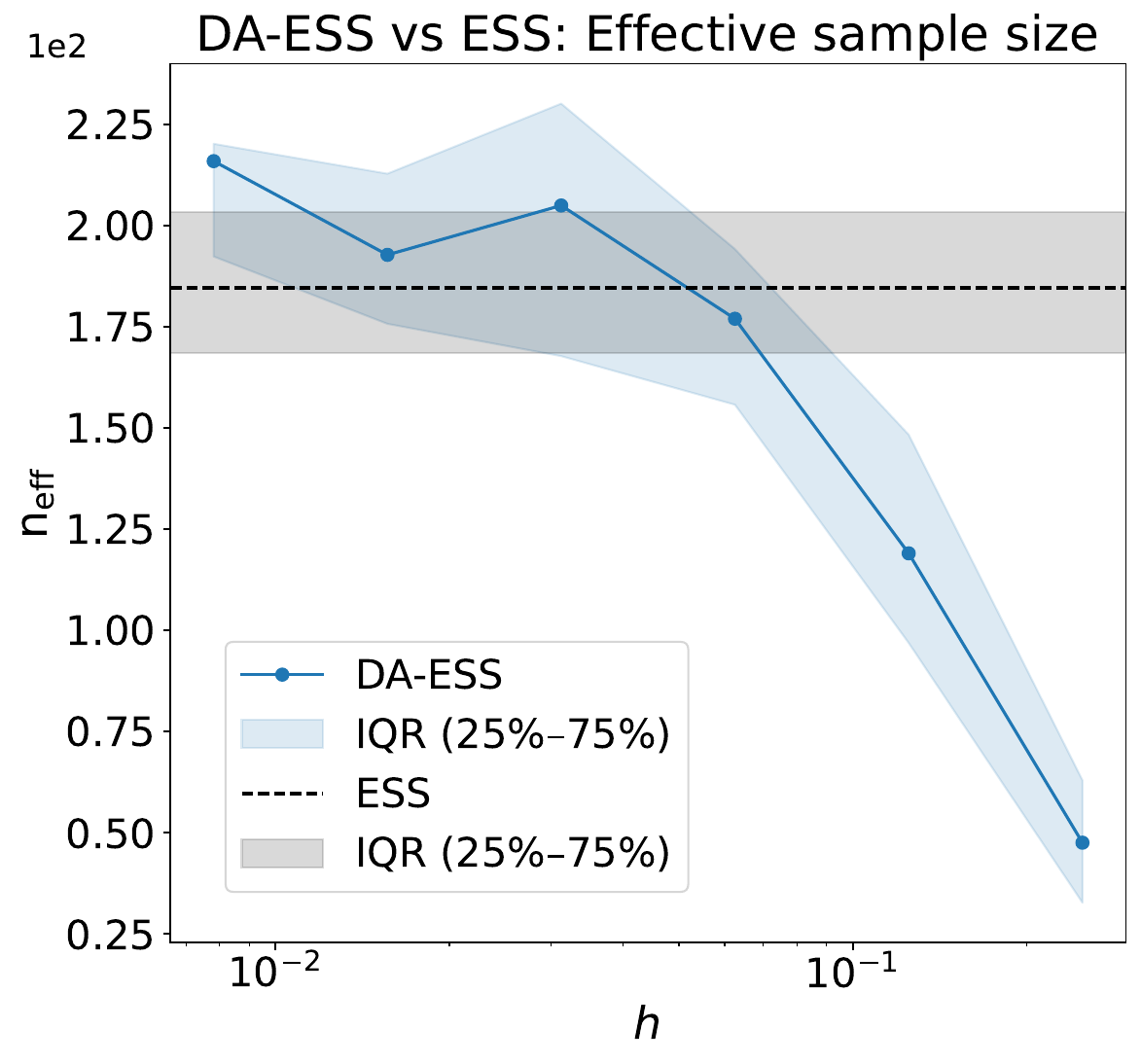}
    \end{minipage}
    \hfill
    \begin{minipage}{0.32\textwidth}
       \centering
        \includegraphics[width=\linewidth]{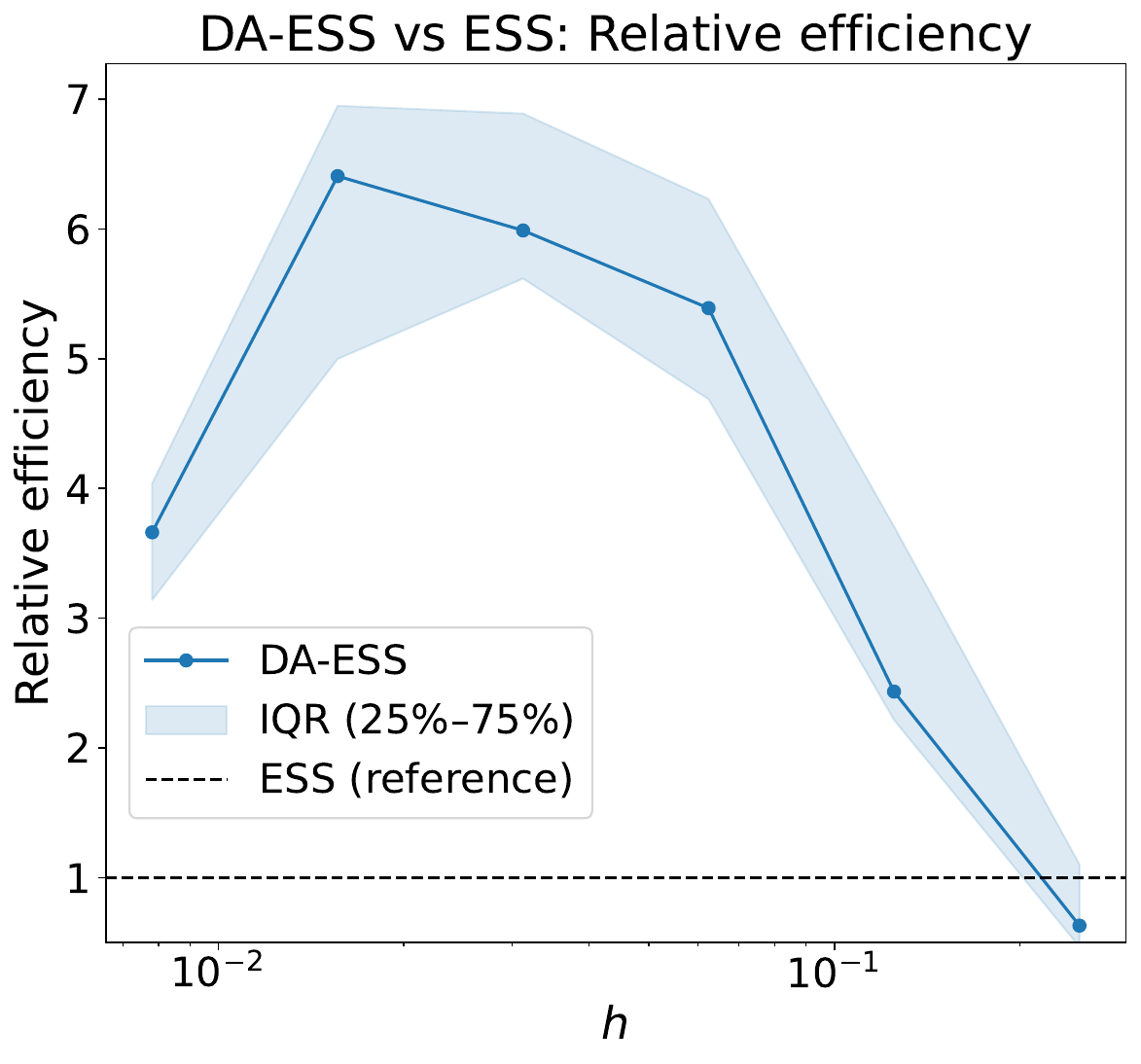}
    \end{minipage}
    \caption{\change{Median of the total computation time (left), median of the effective sample size for MCMC integration of $f$ (middle) and resulting median of the relative efficiency \eqref{rel:efficiency} (right) of DA-ESS for the Bayesian inverse problem of Subsection \ref{sec:BIP2} for 10 independent runs. The shaded regions indicate the corresponding interquartile ranges.}}
    \label{fig:BIP2}
\end{figure*}

}

\subsection{Bayesian Logistic Regression}\label{sec:BLR}

We now consider Bayesian logistic regression for the City Types Dataset \cite{Elebiary} which contains for $6$ locations, e.g. Zurich or Beijing, the air quality, which was measured every hour on every day in the year 2024 in each city. In total we have $m= 52 \, 705$ data points of air quality each one consisting of six real-valued variables $\xi_i\in\mathbb{R}^{6}$ such as measured carbon monoxide or ozone concentration, as well as labels $\delta_i\in\{-1,1\}$, $i=1, \dots, m$, denoting whether the area is residential ($\delta_i = 1$) or industrial ($\delta_i=-1$).        
We aim now to predict the probability of $\delta \in \{\pm 1\}$ given features $\xi$ by logit model, i.e.,
\[
L(\delta, \xi; x)
=
\frac{1}{1+ \exp(-\delta (x_0 + (x_1,\ldots,x_{6})^\top \xi ))}.
\]
To this end, we infer the unknown coefficients $x = (x_0, \ldots, x_{\change{6}})$ using a Bayesian approach. We take again a Gaussian prior $\pi_0=\mathrm{N}(0,C)$ with $C=0.1^2\cdot I_{7}$ this time.
Hence, the posterior for $x$ resulting from conditioning on the available data is of the form (\ref{PMLR:distribution}) with density 
\begin{align}\label{logreg:density}
	\varrho(x)
	=
	\prod_{i=1}^{m} L(\delta_i, \xi_i; x).
\end{align}

 \begin{figure*}[t]
    \begin{minipage}{0.32\textwidth}
		\centering
		\includegraphics[width=\linewidth]{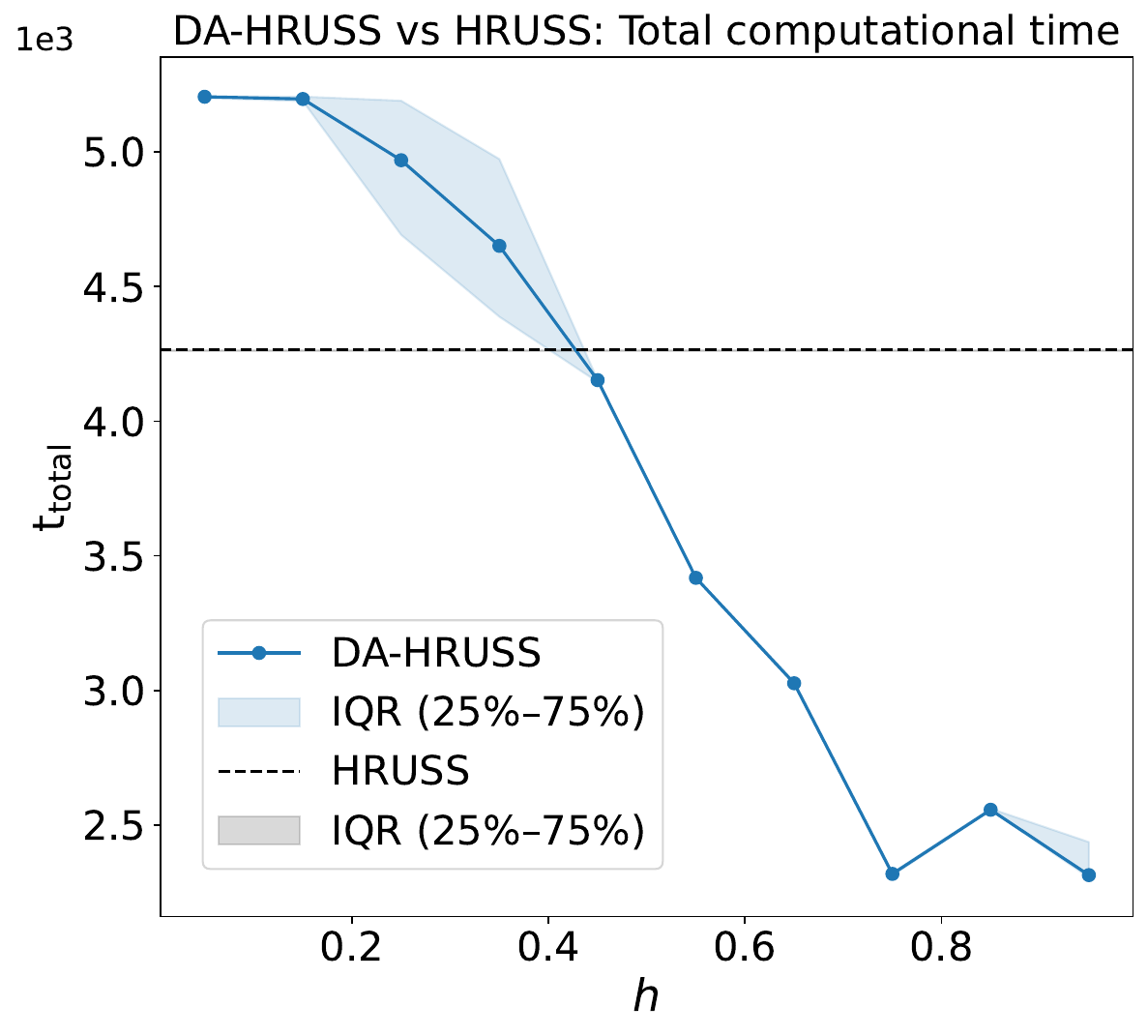}
	\end{minipage}
	\begin{minipage}{0.32\textwidth}
		\centering
		\includegraphics[width=\linewidth]{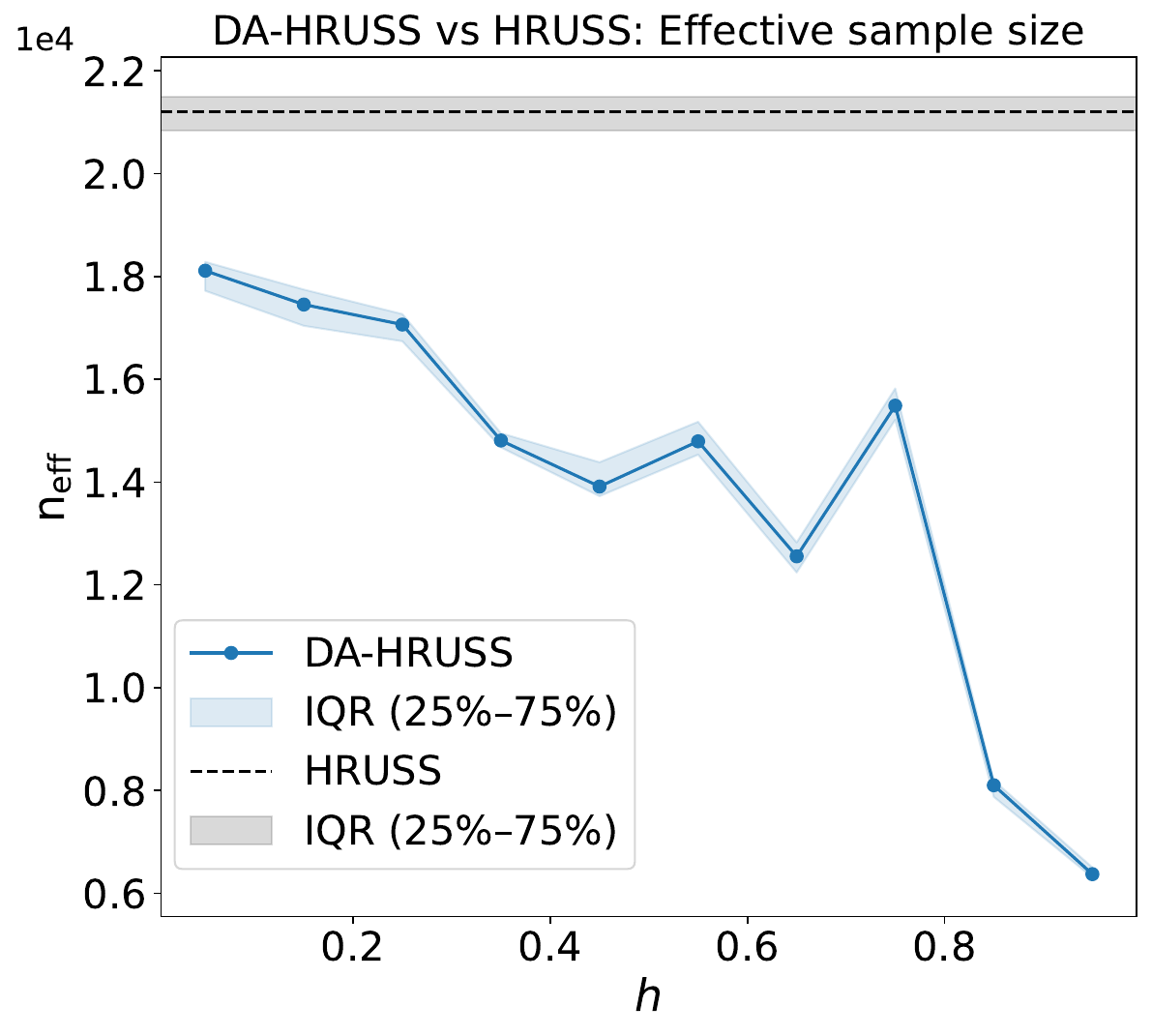}
	\end{minipage}
	\begin{minipage}{0.32\textwidth}
		\centering
		\includegraphics[width=\linewidth]{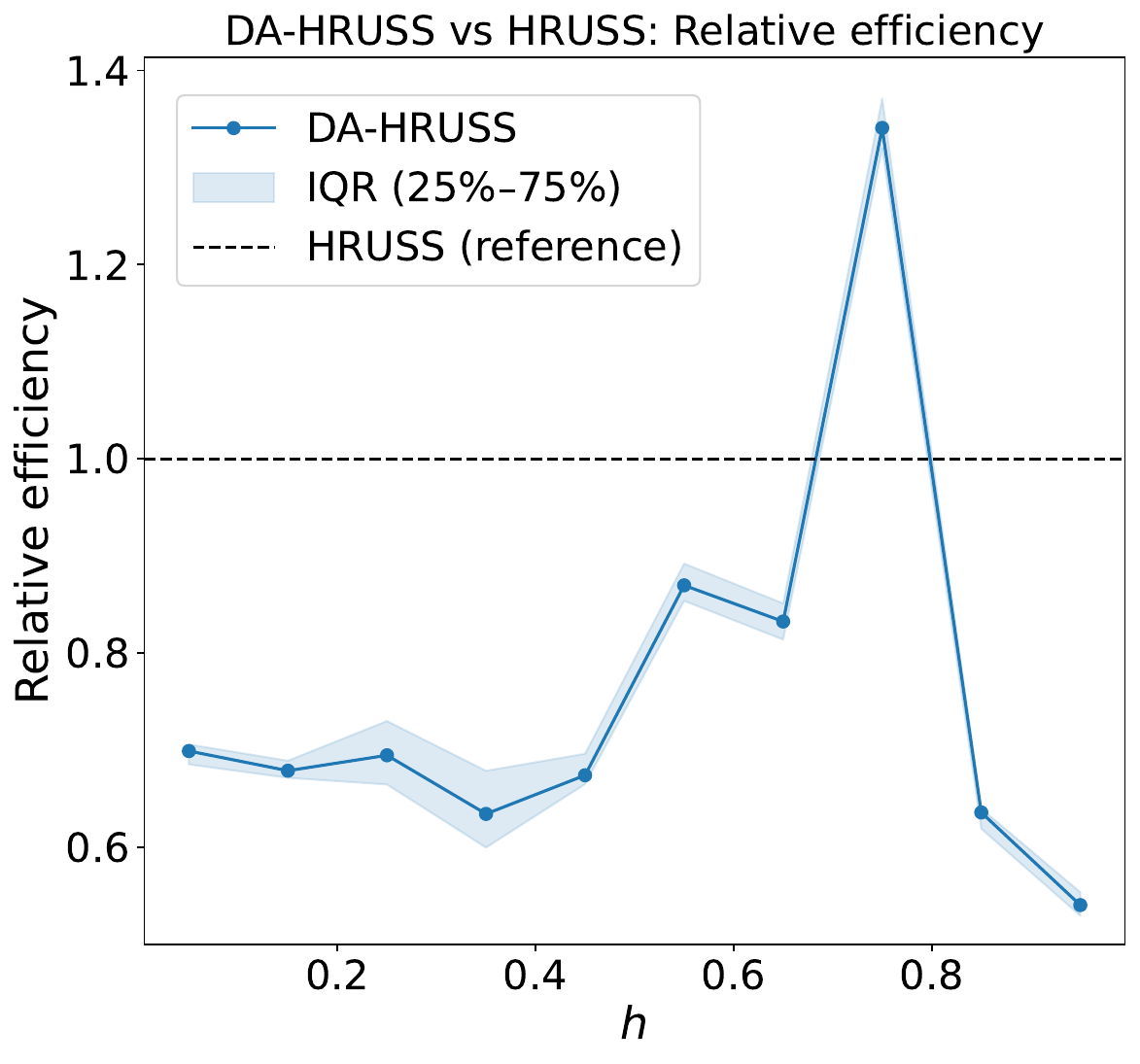}
	\end{minipage}
	\begin{minipage}{0.32\textwidth}
		\centering
		\includegraphics[width=\linewidth]{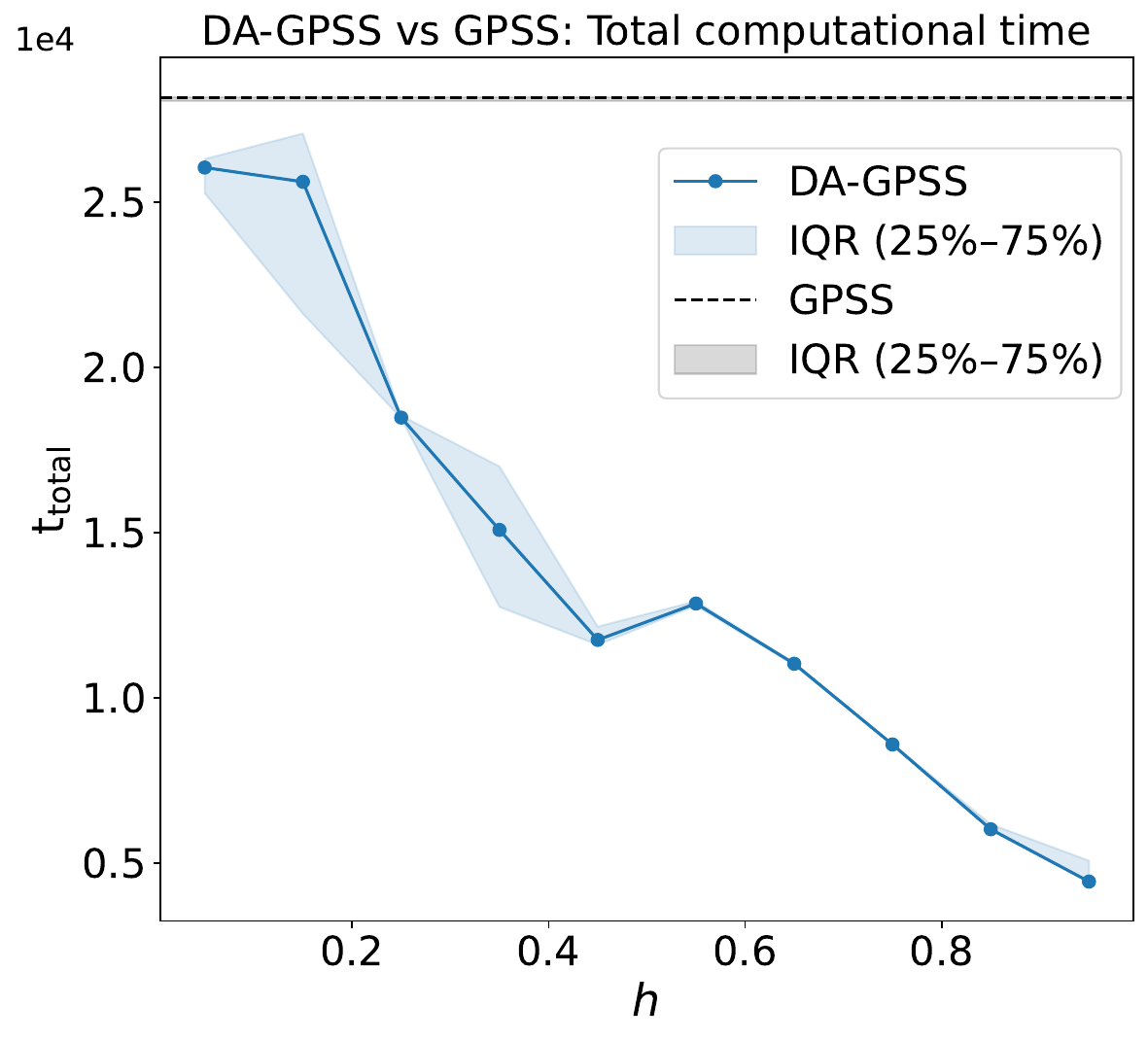}
	\end{minipage}
	\begin{minipage}{0.32\textwidth}
		\centering
		\includegraphics[width=\linewidth]{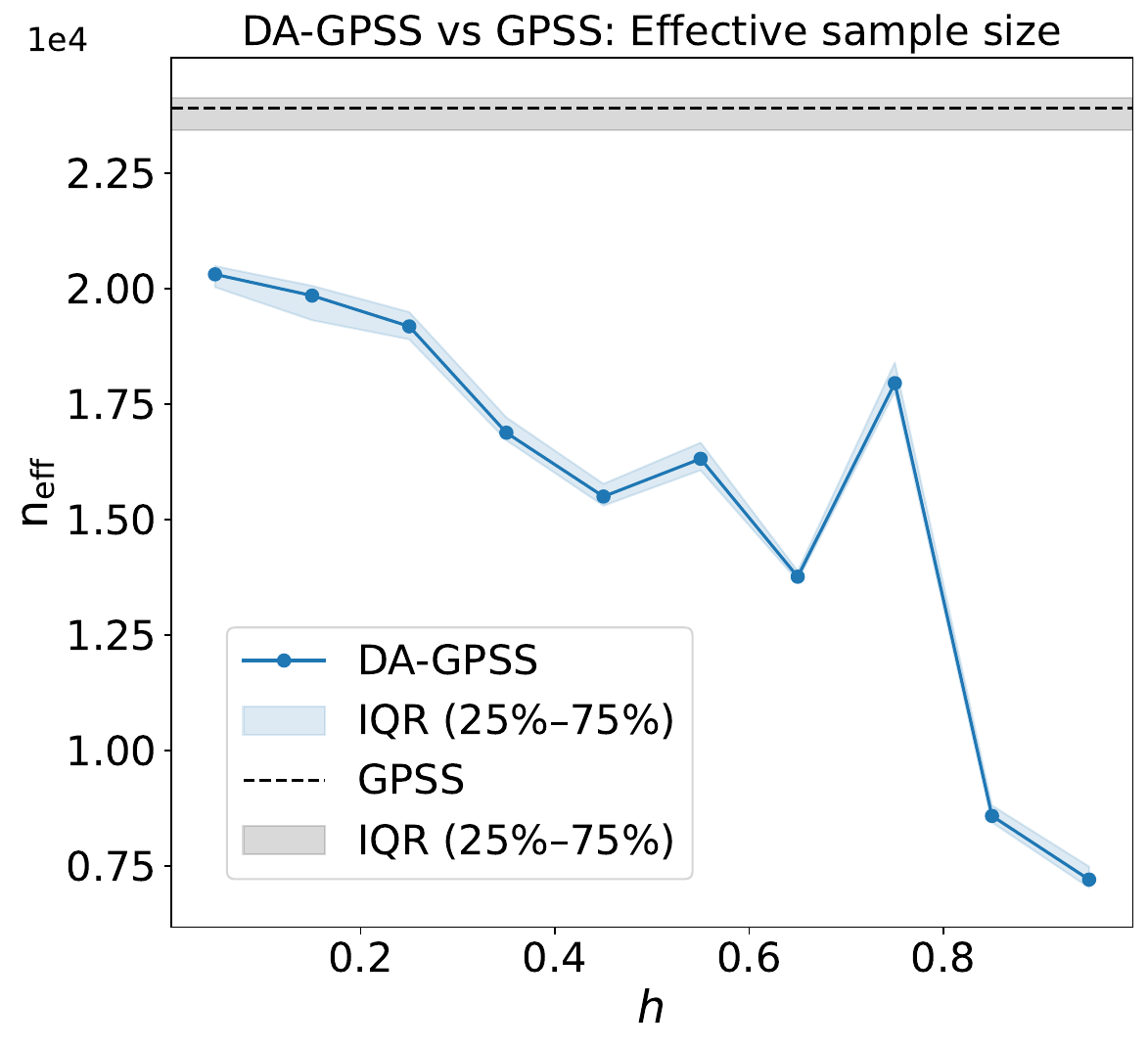}
	\end{minipage}
	\begin{minipage}{0.34\textwidth}
		\centering
		\includegraphics[width=\linewidth]{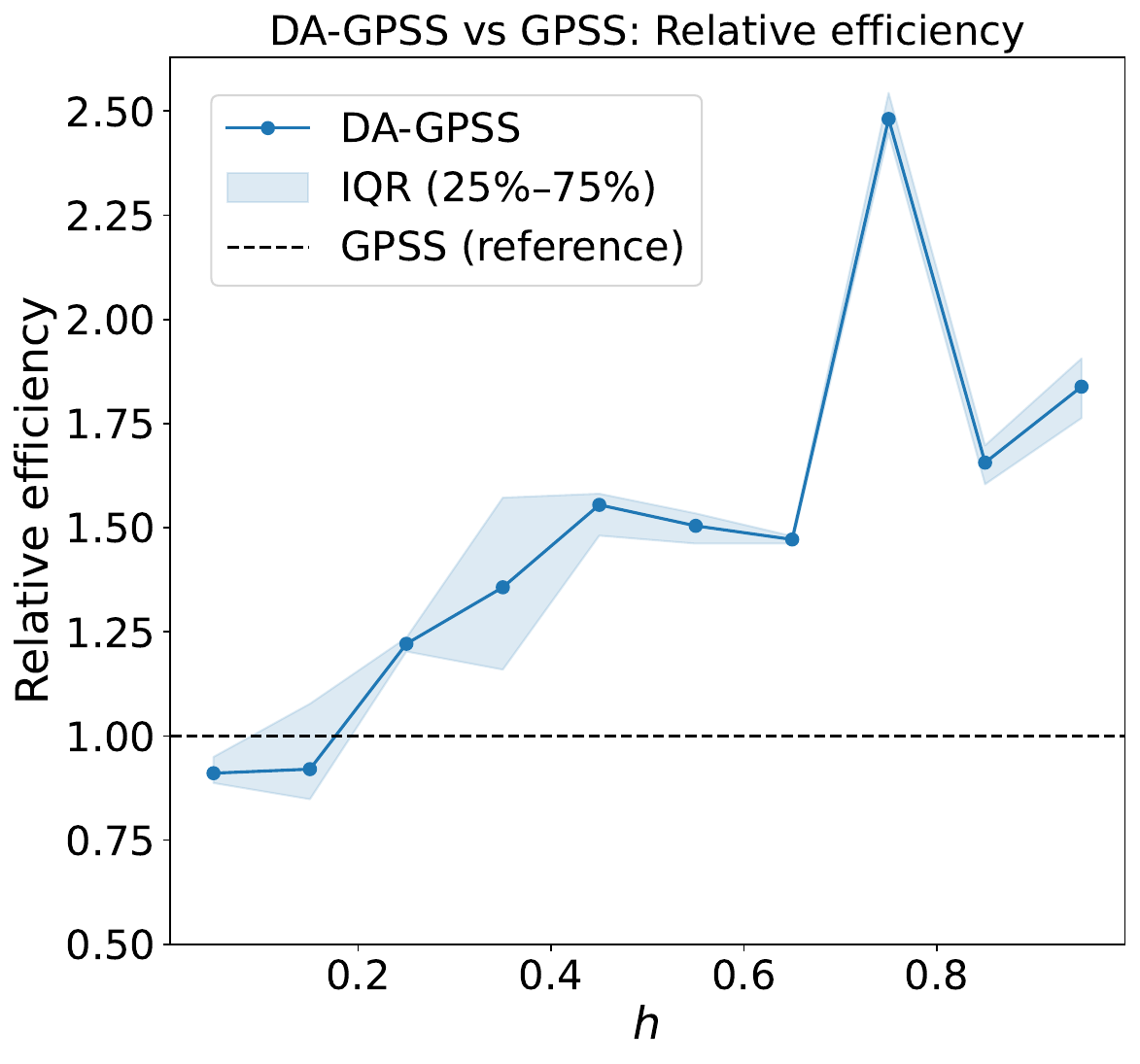}
	\end{minipage}
	\caption{\change{Median of the total computation time (left), median of the effective sample size for MCMC integration of $f$ (middle) and resulting median of the relative efficiency \eqref{rel:efficiency} (right) of DA-HRUSS (top row) and DA-GPSS (bottom row) for the Bayesian logistic regression of Subsection~\ref{sec:BLR} for 10 independent runs. The shaded regions indicate the corresponding interquartile ranges.}}
    \label{fig:BLR}
\end{figure*}

A natural construction of an approximation $\rhoapp=\varrho_h$ to $\varrho$ is then to consider only a subset with size $m_\text{app}(h)=(1-h)m$ of the original data set $(\xi_i, y_i)_{i=1,\cdots, m}$. This subset is determined randomly beforehand and remains unchanged during simulation. To ease notation, we assume that the data set was randomly shuffled at the beginning \change{using a stratified permutation such that the proportions of the class labels are approximately preserved in all subsets}. This leads in the sense of (\ref{num:exp:posterior}) to the density
\begin{align*}
	\varrho_h(x) 
	=
	  \Big(\prod_{i=1}^{m_\text{app}(h)}
	L(\delta_i, \xi_i; x)\Big)^{m/m_\text{app}(h)},
\end{align*}  
where we added the exponent ${m/m_\text{app}(h)}$ to avoid that the approximate posterior $\pi_\text{app}$ is much flatter or less concentrated than the ``true" posterior $\pi$. 
Here we choose $\pi = \pi_{h_\text{ref}}$ with $h_{\text{ref}}=0$ as the ``true" target, i.e. no data points are omitted. As approximations to $\varrho=\varrho_{h_\text{ref}}$, we choose $\rhoapp=\varrho_h$ for $h \in \{0.05, 0.15, \ldots, 0.95\}$, resulting in the omission of $5\%$ up to $95\%$ of the data points. Note that the subsets for the approximations in the simulation are nested. 

We apply HRUSS and DA-HRUSS as well as GPSS and DA-GPSS to sample approximately from $\pi$. Since the reference measure $\pi_0$ for these methods is not Gaussian but the Lebesgue measure and $\pi_0(\mathrm dx) = \|x\|^{1-d} \, \mathrm dx$, respectively, we have to adjust the densities $\varrho$ and $\varrho_h$ to
\[
    \varrho_h^{\text{HRUSS}}(x)
    =
    \varrho_h(x) \cdot \frac{1}{\sqrt{(2\pi)^d\cdot \det(C)}} \exp(- x^\top C^{-1}x/2)
\]
and
\[
    \varrho_h^{\text{GPSS}}(x)
    =
    \varrho_h(x) \cdot \frac{1}{\sqrt{(2\pi)^d\cdot \det(C)} \lVert x \rVert^{1-d}} \exp(-x^\top C^{-1}x/2)
\]
respectively, where $d$ is the number of entries of $x$.

Here, each Markov chain is run \change{10 times independently} for $n=5\cdot 10^5$ iterations after a burn-in of $n_0=10^4$ iterations using the step size parameter $w=0.1$. We estimate $\pi(f)$ for $f(x) := L(1, \xi; x)$, where $\xi\in\mathbb R^{6}$ is the first point of the shuffled data set, which we removed before generating the Markov chains.

In Figure \ref{fig:BLR} we report the results analogous to the previous examples. \change{So, the lines show the medians and the shaded regions the interquartile ranges.} In the left panel of Figure \ref{fig:BLR} we observe again the \change{median} total computational time of (DA-)HRUSS and (DA-)GPSS. \change{The median} total computational time shows for both delayed acceptance algorithms a monotonic decay with \change{increasing coarseness of the approximation}. \change{However, the reduction is considerably more pronounced for DA-GPSS than for DA-HRUSS. In particular, DA-GPSS is faster than its plain counterpart GPSS for all considered approximations, whereas DA-HRUSS does not consistently outperform HRUSS in terms of computational time.} Also the \change{median} effective samples sizes in the center panel of Figure \ref{fig:BLR} show for DA-HRUSS and DA-GPSS a monotonic decay with \change{increasing coarseness of the approximation}, but here punctuated by small fluctuations. Nonetheless, the \change{median} effective sample sizes departs fairly rapidly from that of the plain (full-data) algorithms HRUSS and GPSS. Still, the reduction is gradual across most values of the subsample size $m_\text{app}(h)$, with a sharp drop occurring only for very small $m_\text{app}(h)$. \change{Interestingly, the effective sample sizes of both delayed acceptance algorithms follow a very similar pattern across the considered approximations. Combining these two observations leads to contrasting efficiency results for the two algorithms, as shown in the right panel of Figure \ref{fig:BLR}. Both delayed acceptance algorithms can yield a significant efficiency improvement over their plain counterparts. For DA-HRUSS, however, such an improvement is observed only for $\rhoapp$ with $h=0.75$, \change{where a maximum median relative efficiency of $1.35$ is attained,} whereas DA-GPSS outperforms its plain counterpart for $\rhoapp$ with $h\in\{0.25,\ldots,0.95\}$ \change{and achieves a maximum median relative efficiency of $2.50$}. Thus, the differences in the total computational time reductions achieved by the delayed acceptance algorithms relative to their respective plain counterparts have a substantial impact on the resulting efficiency gains.} One possible explanation for the
differences in the smoothness \change{of} the plots \change{compared with in Subsections \ref{sec:BIP1} and  \ref{sec:BIP2}} is, that the approximations \change{considered here} are no longer numerical, \change{but instead depend directly on the particular subset of the data used. Consequently, the specific structure of the dataset may play an important role, in particular with respect to which and how many data points are omitted when constructing the approximations, rather than only the proportions of the class labels being relevant.}


\section{Discussion}\label{sec:concl}
\rev{The delayed acceptance approach offers a promising extension to slice sampling. In the two test examples of Section \ref{sec:num} the efficiency of the hybrid slice samplers was already noticeably increased for certain approximations. However, these improvements may require a suitable choice of the discretizations parameter $h$ for the approximate density $\rhoapp = \varrho_h$. How to choose such an $h$ in advance without quantitative tuning -- i.e. without running the delayed acceptance algorithm for different approximations for a small number of iterations and then checking (\ref{rel:efficiency}) -- remains an open question for future research. However, delayed acceptance slice sampling allows for further improvement by using also the non-accepted generated samples of $\pi_{\rm app}$ for variance reduction by the control variate $X\sim\pi_{\rm app}$. This can be extended to a full multilevel slice sampling model in the future, cf. \cite{LykkegaardEtAl2022}.} 


\bmhead{Acknowledgements}
\change{The authors acknowledge computing time on the compute cluster of the Faculty of Mathematics and Computer Science of Technische Universität Bergakademie Freiberg, operated by the computing center (URZ) and funded by the Deutsche Forschungsgemeinschaft (DFG) under DFG grant number 397252409. Moreover, the} authors thank Philip Schär for valuable discussions and helpful comments on this paper.

\begin{appendices}


\section{Proofs} \label{appendix:proofs}


\subsection{Proof of Proposition \ref{PMLR:prop:hybrid}}\label{app:prop:hybrid}
The representation of \eqref{eq:S_DA_hybrid} is obvious, thus, we only show the reversibility of $H_t$ defined in \eqref{eq:S_DA_Ht}. 
Let $x\in G$ and $A,B\in\mathcal{B}(G)$. 
Then, using the fact that 
$\mathbb{1}_{G_s,\rhoapp}(x)=\mathbb{1}_{(0,\rhoapp(x))}(s)$ for any $s\in(0,\Vert \rhoapp \Vert_{\infty})$
and setting $c:=1/(Z_{\text{app}}\cdot \pi_{\text{app}}(G_{t, \widehat{\varrho}}))$, where 
\begin{align}\label{normalizing:app}
   Z_{\text{app}} :=\int_G^{} \rhoapp(x)  \, \pi_0(\mathrm{d}x)    
\end{align}
yields
\begin{align*}
	&\int_{A} H_t(x,B)  \, \pi_{\text{app},t}(\mathrm{d}x) = \int_{A} \left( \frac{1}{\rhoapp(x)} \int_{0}^{\rhoapp(x)} \pi_{0,t,s}(B) \mathrm{d}s \right) \frac{\mathbb{1}_{G_{t,\widehat{\varrho}}}(x)}{\pi_{\text{app}}(G_{t, \widehat{\varrho}})}  \, \pi_{\text{app}}(\mathrm{d}x) \\
	&= c \int_{A\cap G_{t,\widehat{\varrho}}} \int_{0}^{\rhoapp(x)} \frac{\pi_0(B \cap G_{s,\rhoapp} \cap  G_{t,\widehat{\varrho}})}{\pi_0(G_{s,\rhoapp} \cap G_{t,\widehat{\varrho}})}
	 \, \mathrm{d}s  \, \pi_0(\mathrm{d}x) \\
	&= c  \int_{A\cap G_{t,\widehat{\varrho}}} \int_{B\cap G_{t,\widehat{\varrho}}} \int_{0}^{\infty}  \frac{\mathbb{1}_{(0,\rhoapp(x))}(s) \mathbb{1}_{G_{s,\rhoapp}}(y)}{\pi_0(G_{s,\rhoapp} \cap G_{t,\widehat{\varrho}})}  \, \mathrm{d}s  \, \pi_0(\mathrm{d}y)  \, \pi_0(\mathrm{d}x) \\
	&= c  \int_{A\cap G_{t,\widehat{\varrho}}} \int_{B\cap G_{t,\widehat{\varrho}}} \int_{0}^{\infty}  \frac{\mathbb{1}_{G_s,\rhoapp}(x) \mathbb{1}_{G_{s,\rhoapp}}(y)}{\pi_0(G_{s,\rhoapp} \cap G_{t,\widehat{\varrho}})}  \, \mathrm{d}s  \, \pi_0(\mathrm{d}y)  \, \pi_0(\mathrm{d}x).
\end{align*}
Interchanging the two outer integrals and arguing by the same steps (with $x$ taking the role of $y$) backwards gives the desired reversibility, i.e. 
 \begin{align*}
        \int_{A} H_t(x,B)  \, \pi_{\text{app},t}(\mathrm{d}x)=\int_{B} H_t(y,A)  \, \pi_{\text{app},t}(\mathrm{d}y), \qquad A,B\in\mathcal{B}(G).
    \end{align*}
An immediate consequence is the $\pi$-reversibility of $S_{DA}$. Therefore, let $x\in G$ as well as $A,B\in\mathcal{B}(G)$ and set $\tilde{c}:=Z^{-1} \, Z^{-1}_{\text{app}} \, \pi_{\text{app}}(G_{t,\widehat{\varrho}})$. Then       
 \begin{align*}
        \int_{A} S_{DA}(x,B) \, \pi(\mathrm{d}x)&= \int_{A} \left(\frac{1}{\widehat{\varrho}(x)} \int_{0}^{\widehat{\varrho}(x)} H_t(x,B) \mathrm{d}t \right) \, Z^{-1} \, Z^{-1}_{\text{app}} \, \widehat{\varrho}(x) \, \pi_{\text{app}}(\mathrm{d}x) \\ 
        &=Z^{-1} \, Z^{-1}_{\text{app}}\int_{A} \int_0^\infty H_t(x,B) \mathbb{1}_{G_{t,\widehat{\varrho}}}(x) \frac{\pi_{\text{app}}(G_{t,\widehat{\varrho}})}{\pi_{\text{app}}(G_{t,\widehat{\varrho}})} \, \mathrm{d}t \, \pi_{\text{app}}(\mathrm{d}x) \\
        &= \tilde{c} \int_0^\infty \int_{A} H_t(x,B) \, \pi_{\text{app},t}(\mathrm{d}x) \, \mathrm{d}t \\
        &= \tilde{c} \int_0^\infty \int_{B} H_t(y,A) \, \pi_{\text{app},t}(\mathrm{d}y) \, \mathrm{d}t.
 \end{align*}
 Again, arguing by the same steps backwards (with $y$ instead of $x$) gives the desired result.  
\begin{flushright}    
$\square$
\end{flushright}

\subsection{Proof of Theorem \ref{DASS:Ergodicity}}\label{app:DASS_convergence}






In order to show the ergodicity of delayed acceptance slice sampling, we need to verify the assumptions of Proposition \ref{Schaer:Convergence}.
To this end, we state the following auxilliary result.

\begin{proposition}\label{positive:measure:iff}
    Let $\pi$ be a probability measure given by (\ref{PMLR:distribution}) and $A\in\mathcal{B}(G)$. Then 
    $\pi(A)>0$ iff there exists $0<S_A\leq \|\rhoapp\|_\infty$ and $0 < T_A \leq \|\widehat{\varrho}\|_\infty$ with $\rhoapp, \widehat{\varrho}$ as in \eqref{PMLR:factorization} such that
    \begin{align*}
        \pi_0(A\cap G_{s, \varrho_{\textmd{app}}}\cap G_{t, \widehat{\varrho}}) > 0, \qquad  \forall s\in(0, S_A) \ \forall t\in(0, T_A).   
    \end{align*}
\end{proposition}

\begin{proof}
    Let $x\in G$ and $A\in\mathcal{B}(G)$. We rewrite the density $\varrho$ in the following form:
    \begin{align*}
        \varrho(x)  = \widehat{\varrho}(x)  \rhoapp(x) & = \int_0^{\infty} \int_0^{\infty} \mathbb{1}_{(0,\varrho_{\textmd{app}(x)})}(s) \mathbb{1}_{(0,\widehat{\varrho}(x))}(t)  \,  \mathrm{d}s  \, \mathrm{d}t\\
        & = \int_0^{\infty} \int_0^{\infty} \mathbb{1}_{G_{s, \varrho_{\textmd{app}}}\cap G_{t, \widehat{\varrho}}}(x)  \, \mathrm{d}s  \, \mathrm{d}t
    \end{align*}
    This leads combined with Fubini's theorem to
    \begin{align}
        \pi(A)  = \frac {1}{Z} \, \int_A \varrho(x) \, \pi_0(\mathrm{d}x) & = \frac {1}{Z} \int_0^{\infty} \int_0^{\infty}  \int_A \mathbb{1}_{G_{s, \varrho_{\textmd{app}}}\cap G_{t, \widehat{\varrho}}}(x) \, \pi_0(\mathrm{d}x)  \, \mathrm{d}s  \, \mathrm{d}t \notag \\
        & = \frac {1}{Z} \int_0^{\infty} \int_0^{\infty} \pi_0(A \cap G_{s, \varrho_{\textmd{app}}}\cap G_{t, \widehat{\varrho}})  \, \mathrm{d}s  \, \mathrm{d}t \label{eq:pi}
    \end{align}
    Since $Z>0$ we can conclude that $\pi(A)>0$ if and only if the double integral on the right-hand side is also positive. Define now
    \begin{align*}
       f: (0,\infty)\times(0,\infty) \rightarrow \mathbb R, \quad (t,s) \mapsto \pi_0(A \cap G_{s, \varrho_{\textmd{app}}}\cap G_{t, \widehat{\varrho}})    
    \end{align*}
    and notice, that the mapping $f$ is non-increasing since whenever $(t_1,s_1)<(t_2,s_2)$ we have $G_{s_2, \varrho_{\textmd{app}}}\cap G_{t_2, \widehat{\varrho}}\subseteq G_{s_1, \varrho_{\textmd{app}}}\cap G_{t_1, \widehat{\varrho}}$. This leads to the fact that $\pi(A)>0$ if and only if there exist an $S_A>0$ and an $T_A>0$ such that $\pi_0(A \cap G_{s, \varrho_{\textmd{app}}}\cap G_{t, \widehat{\varrho}})>0$ for all $s\in(0, S_A)$ and $t\in(0, T_A)$.
\end{proof}

\begin{proof}[Proof of Theorem \ref{DASS:Ergodicity}]
It follows from Proposition~\ref{PMLR:prop:hybrid} that \(S_{DA}\) is \(\pi\)-reversible and therefore \(\pi\)-invariant. 
It remains to show \eqref{eq:ergodic:abscon1} and \eqref{eq:ergodic:abscon2} for $S_{DA}$, since then by Proposition \ref{Schaer:Convergence} the statement follows.
To this end, let $x\in G$ and $A\in\mathcal{B}(G)$ be arbitrary in the following. 

If $\pi(A) = 0$, then also $\pi_0(A)=0$ by \eqref{PMLR:distribution} and, thus, $\pi_0(A \cap G_{s, \varrho_{\textmd{app}}}\cap G_{t, \widehat{\varrho}})=0$ for all $(t,s)\in (0,\widehat\varrho(x))\times(0,\rhoapp(x))$.
Hence, we also have $S_{DA}(x,A) = 0$.

Let now $\pi(A)>0$. By Proposition~\ref{positive:measure:iff} there exists $0<S_A\leq \|\rhoapp\|_\infty$ and $0<T_A\leq \|\widehat{\varrho}\|_\infty$, such that $\pi_0(A\cap G_{s, \varrho_{\textmd{app}}}\cap G_{t, \widehat{\varrho}}) > 0$ for all $s\in(0, S_A)$ and $t\in(0, T_A)$. Thus, for $x\in G$ exist $S_{x,A}\in(0, \rhoapp(x)]$ and $T_{x,A}\in(0, \widehat{\varrho}(x)]$ for which by Assumption \ref{assum:positivity_condition} we can conclude that $\pi_{0,t,s}(A)>0$ for all $s\in(0, S_{x,A})$ and $t\in(0, T_{x,A})$. 
Hence, 
\begin{align*}
	S_{DA}(x,A) 
    & = 
    \frac{1}{\varrho(x)}\int_0^{\widehat{\varrho}(x)} \int_0^{\rhoapp(x)} \pi_{0,t,s}(A)  \, \mathrm{d}s  \, \mathrm{d}t\\
    & \geq \frac{1}{\varrho(x)}\int_0^{T_{x,A}} \int_0^{S_{x,A}} \pi_{0,t,s}(A)  \, \mathrm{d}s  \, \mathrm{d}t
    > 0
\end{align*}
as desired. 
\end{proof}

\subsection{Proof of Proposition \ref{PMLR:prop:DAMH:HSS}}
\label{sec:DAMH:HSS}
For $t\in (0, \|\widehat\varrho\|_\infty)$ we show that $H_t$ (defined in Proposition~\ref{PMLR:prop:DAMH:HSS}) is $\pi_{\text{app},t}$-reversible. Let $A,B\in\mathcal{B}(G)$ with $A\cap B=\emptyset$.
Then, using $c:=1/(Z_{\text{app}}\cdot \pi_{\text{app}}(G_{t, \widehat{\varrho}}))$, where $Z_{\text{app}}$ is defined as in (\ref{normalizing:app}), we obtain
\begin{align*}
	& \int_{A} H_t(x,B)  \, \pi_{\text{app},t}(\mathrm{d}x) \\
	&=\int_{A} \left( \frac{1}{\rhoapp(x)} \int_{0}^{\rhoapp(x)} Q(x,B \cap G_{s,\rhoapp}  \cap  G_{t,\widehat{\varrho}}) \mathrm{d}s \right) \frac{\mathbb{1}_{G_{t,\widehat{\varrho}}}(x)}{\pi_{\text{app}}(G_{t, \widehat{\varrho}})}  \, \pi_{\text{app}}(\mathrm{d}x) \\
	&= c  \int_{G} \int_{G} \int_{0}^{\rhoapp(x)} \mathbb{1}_{B \cap G_{s,\rhoapp} \cap G_{t,\widehat{\varrho}}}(y) \mathbb{1}_{A \cap G_{t,\widehat{\varrho}}}(x)  \, \mathrm{d}s  \, Q(x, \mathrm{d}y)  \, \pi_0(\mathrm{d}x) \\ 
	&= c  \int_{G} \int_{G} \int_{0}^{\infty} \mathbb{1}_{(0,\rhoapp(x))}(s) \mathbb{1}_{G_{s,\rhoapp} }(y) \mathbb{1}_{A \cap  G_{t,\widehat{\varrho}}}(x) \mathbb{1}_{B \cap G_{t,\widehat{\varrho}}}(y)  \, \mathrm{d}s  \, Q(x, \mathrm{d}y)  \, \pi_0(\mathrm{d}x) \\ 
	&= c  \int_{G} \int_{G} \int_{0}^{\infty} \mathbb{1}_{G_{s,\rhoapp} }(x) \mathbb{1}_{G_{s,\rhoapp} }(y) \mathbb{1}_{A \cap  G_{t,\widehat{\varrho}}}(x) \mathbb{1}_{B \cap G_{t,\widehat{\varrho}}}(y)  \, \mathrm{d}s  \, Q(x, \mathrm{d}y)  \, \pi_0(\mathrm{d}x).
\end{align*}
Consequently, by exploiting the reversibility of $Q$ w.r.t. $\pi_0$, that is, $Q(x, \mathrm{d}y) \ \pi_0(\mathrm{d}x) = Q(y, \mathrm{d}x) \ \pi_0(\mathrm{d}y)$, and arguing by the same steps backward we obtain
$\int_{A} H_t(x,B) \pi_{\text{app},t}(\mathrm{d}x) = \int_{B} H_t(y,A) \pi_{\text{app},t}(\mathrm{d}y)$.
This can be straightforwardly extended to sets that are not necessarily disjoint.

We turn to the verification of $H(x,A)=M_{DA}(x,A)$ for any $x\in G$ and $A\in\mathcal{B}(G)$, where $M_{DA}$ denotes the transition kernel of the delayed acceptance Metropolis-Hastings algorithm. For any $A\in\mathcal{B}(G)$ that satisfies $x\not\in A$ follows
\begin{align*}
	H(x,A)&= \frac{1}{\widehat{\varrho}(x)} \int_{0}^{\widehat{\varrho}(x)} H_t(x,A)  \, \mathrm{d}t \\
	&= \frac{1}{\widehat{\varrho}(x)} \int_{0}^{\widehat{\varrho}(x)} \frac{1}{\rhoapp(x)} \int_0^{\rhoapp(x)} Q(x,A \cap G_{s,\rhoapp} \cap  G_{t,\widehat{\varrho}})  \, \mathrm{d}s  \, \mathrm{d}t\\
	&= \frac{1}{\widehat{\varrho}(x)} \frac{1}{\rhoapp(x)} \int_{0}^{\widehat{\varrho}(x)} \int_0^{\rhoapp(x)} \int_A \mathbb{1}_{(0,\varrho_{\textmd{app}}(y))}(s) \mathbb{1}_{(0,\widehat{\varrho}(y))}(t)  \, Q(x, \mathrm{d}y)  \, \mathrm{d}s  \, \mathrm{d}t \\  
	&= \int_A \min\{1,\rhoapp(y)/\rhoapp(x) \} \min\{1,\widehat{\varrho}(y)/ \widehat{\varrho}(x) \}  \, Q(x, \mathrm{d}y) = M_{DA}(x,A).
\end{align*}
Applying the former with $A=G\setminus\{x\}$ implies 
\[H(x,\{x\})=1-H(x,G\setminus\{x\}) = 1-M_{DA}(x,G\setminus\{x\})=M_{DA}(x,\{x\}).\]
Consequently, for sets with $x\in A$ we have 
\[
H(x,A) = H(x,\{x\}) + H(x,A\setminus\{x\}) = M_{DA}(x,\{x\}) + M_{DA}(x,A\setminus\{x\}) = M_{DA}(x,A).
\]
\begin{flushright}    
$\square$
\end{flushright}

\subsection{Proof of Theorem~\ref{PMLR:theorem:comparison}}\label{Apendix:Comparison}
We show the statement of Theorem \ref{PMLR:theorem:comparison} by applying a general comparison result of \citet[Lemma~1]{RudolfUllrich2018}. For convenience of the reader, we state the corresponding statement in its abstract form and then verify its assumptions in the setting of Theorem~\ref{PMLR:theorem:comparison}.

Additionally to the measurable space $(G,\mathcal{B}(G))$ let $(I,\mathcal{I},\lambda)$ be a $\sigma$-finite measure space and assume that: 
\begin{itemize}
	\item[{(a)}] 
	There is a function $p\colon G\times I\rightarrow [0,\infty]$ such that $p(x,\cdot)$ is a probability density function w.r.t. $\lambda$ for all $x\in G$ and $p(\cdot,a)$ is integrable w.r.t. $\pi$ for all $a\in I$.
	\item[{(b)}]  
	For every $a\in I$, we have an equivalence relation $\sim_a$ on $G$. We denote the equivalence class of $x\in G$ w.r.t. $\sim_a$ by $[x]_a:=\{y\in G: x \sim_a y\}$.
	\item[{(c)}] 
	For ($\lambda$-almost) every $a\in I$, there exists a transition kernel $P_a$ on $(G, \mathcal{B}(G))$ such that $P_a(x,A)=0$ for each $x\in G$ and $A\in \mathcal{B}(G)$ satisfying $A \subseteq G\setminus [x]_a$. 
\end{itemize}
Moreover, we define for almost all $a\in I$ a probability measure $\pi_{{a}}$ on $G$ by
\begin{align*}
	\pi_a(A):=\frac{\int_{A}p(x,a) \,\pi(\mathrm{d}x)}{\int_{{G}}p(x,a) \,\pi(\mathrm{d}x)}, \quad A\in\mathcal{B}(G).
\end{align*} 

\begin{lemma}[\citet{RudolfUllrich2018}]\label{PMLR:lemma:comparison}
	Assume that for ($\lambda$-almost) all $a\in I$, there are transition kernels $P^{(1)}_a$, $P^{(2)}_a$ on $G$ such that
	\begin{enumerate}
		\item $P^{(1)}_a$ and $P^{(2)}_a$ are self-adjoint operators on $L^2_{\pi_{a}}(G)$;
		\label{en1}
		\item $P^{(1)}_a$ is positive on $L^2_{\pi_{a}}(G)$;
		\item $P^{(1)}_aP^{(2)}_a=P^{(2)}_a$. 
		\label{en3}
	\end{enumerate}
	Then, for the operators $P_1, P_2\colon L^2_{\pi}(G) \rightarrow L^2_{\pi}(G)$ defined by
	\begin{equation*}
		P_if(x)=\int_{I} P^{(i)}_a f(x) p(x,a) \,\lambda(\mathrm{d}a), \qquad i=1,2,
	\end{equation*}
	it follows $P_1\leq P_2$.  
\end{lemma}

\begin{proof}[Proof of Theorem \ref{PMLR:theorem:comparison}]
	We now apply Lemma \ref{PMLR:lemma:comparison} to the setting of Theorem \ref{PMLR:theorem:comparison}.
	To this end, let $I:=(0,\infty)\times(0,\infty)$ and $\lambda:=\lambda_2$ be the two-dimensional Lebesgue measure. 
	We define the function $p\colon G\times I \rightarrow [0,\infty)$ by
	\begin{align*}
		p(x,a):=\frac{1}{\varrho(x)}\mathbb{1}_{(0,\varrho_{\textmd{app}}(x))}(s)\mathbb{1}_{(0,\widehat{\varrho}(x))}(t)
	\end{align*}
	with $a=(s,t)$, which satisfies assumption (a) from above. For every $a=(s,t)\in I$, let $x \sim_a y$ if and only if $x,y\in G_{s, \varrho_{\textmd{app}}}\cap G_{t, \widehat{\varrho}}$, such that $[x]_{a}:=G_{s, \varrho_{\textmd{app}}}\cap G_{t, \widehat{\varrho}}$. This leads to
	\begin{align*}
		\pi_{a}(A)=\frac{\pi_0(A\cap G_{s, \varrho_{\textmd{app}}}\cap G_{t, \widehat{\varrho}})}{\pi_0(G_{s, \varrho_{\textmd{app}}}\cap G_{t, \widehat{\varrho}})}, \quad A\in\mathcal{B}(G).
	\end{align*} 
	For $x\in G$ and $A\in \mathcal{B}(G)$ we set 
	$ S_{DA_{{a}}}(x,A):=\pi_{{a}}(A)$.
	This is clearly a transition kernel that satisfies assumption (c) above. 
	Moreover, the Markov operator $S_{DA}: L^2_{\pi}(G) \rightarrow L^2_{\pi}(G)$ of the delayed acceptance slice sampler can be represented by
	\begin{align*}
		S_{DA}f(x) 
		&=\frac{1}{\varrho(x)} \int_0^{\widehat{\varrho}(x)} \int_0^{\rhoapp(x)} \int_G^{} f(y) \frac{\mathbb{1}_{G_{s, \varrho_{\textmd{app}}}\cap G_{t,\widehat{\varrho}}}(y)}{\pi_0(G_{s, \varrho_{\textmd{app}}}\cap G_{t, \widehat{\varrho}})}  \, \pi_{0}(\mathrm{d}y)  \, \mathrm{d}s  \, \mathrm{d}t \\
		&= \int_I \int_G f(y)  S_{DA_{{a}}}(x,\mathrm{d}y) \frac{1}{\varrho(x)}\mathbb{1}_{(0,\varrho_{\textmd{app}(x)})}(s) \ \mathbb{1}_{(0,\widehat{\varrho}(x))}(t)  \, \mathrm{d}s  \, \mathrm{d}t \\ 
		&= \int_I S_{DA_{{a}}}  f(x)  p(x,{a}) \,\lambda_2(\mathrm{d}{a}).
	\end{align*}
	We also set  
	\begin{align*}
		M_{DA_{{a}}}(x,A)=  Q(x, A  \cap  G_{s,\rhoapp}  \cap  G_{t,\widehat{\varrho}}) +(1- Q(x, G_{s,\rhoapp}  \cap  G_{t,\widehat{\varrho}})) \cdot \delta_x(A)   
	\end{align*}
	and note that assumption (c) from above is satisfied for this transition kernel.
	Using
	\begin{align*}
		\int_0^{\widehat{\varrho}(x)} \int_0^{\rhoapp(x)} \mathbb{1}_{(0,\varrho_{\textmd{app}}(y))}(s) \, \mathbb{1}_{(0,\widehat{\varrho}(y))}(t)  \, \mathrm{d}s  \, \mathrm{d}t 
		& = \min\{\rhoapp(x),\rhoapp(y)\} \, \min\{\widehat{\varrho}(x), \widehat{\varrho}(y)\}       
	\end{align*}
	yields
	\begin{align*}
		M_{DA}f(x)& = \int_G f(y) \widetilde{\alpha}(x,y)  \, Q(x,\mathrm{d}y) +f(x)\left(1-\int_{G} \widetilde{\alpha}(x,y)  \, Q(x,\mathrm{d}y)\right)\\
		&= \frac{1}{\varrho(x)} \int_0^{\widehat{\varrho}(x)} \int_0^{\rhoapp(x)} \int_G f(y)  \mathbb{1}_{(0,\varrho_{\textmd{app}}(y))}(s)\ \mathbb{1}_{(0,\widehat{\varrho}(y))}(t)  \, Q(x, \mathrm{d}y)  \, \mathrm{d}s  \, \mathrm{d}t \\ 
		&\qquad	+  \frac{f(x)}{\varrho(x)} \int_0^{\widehat{\varrho}(x)} \int_0^{\rhoapp(x)} \int_G  \left(1- \mathbb{1}_{(0,\varrho_{\textmd{app}}(y))}(s)\ \mathbb{1}_{(0,\widehat{\varrho}(y))}(t)\right)  \, Q(x, \mathrm{d}y)  \, \mathrm{d}s  \, \mathrm{d}t\\	
		&= \int_I \int_G f(y)  M_{DA_{{a}}}(x,\mathrm{d}y) \ \frac{1}{\varrho(x)}\mathbb{1}_{(0,\varrho_{\textmd{app}(x)})}(s)\ \mathbb{1}_{(0,\widehat{\varrho}(x))}(t)  \, \mathrm{d}s  \, \mathrm{d}t \\ 
		&= \int_I   M_{DA_{{a}}}f(x) p(x,{a}) \,\lambda_2(\mathrm{d}{a}).
	\end{align*}
	Now we verify \ref{en1}.-\ref{en3}. of Lemma~\ref{PMLR:lemma:comparison} with $P_a^{(1)}=M_{DA_{{a}}}$ and $P_a^{(2)}=S_{DA_{{a}}}$, such that the application of it gives the statement of Theorem~\ref{PMLR:theorem:comparison}. We have:
	\begin{enumerate}
		\item
		\emph{Self-adjointness of $S_{DA_{{a}}}$ and $M_{DA_{{a}}}$ w.r.t.\phantom{.}$\pi_{{a}}$}: It is sufficient to show the $\pi_{{a}}$-reversibility of the transition kernels $S_{DA_{{a}}}$ and $M_{DA_{{a}}}$. 
		For $S_{DA_{{a}}}$ it is obvious due to the relation $S_{DA_{{a}}}(x,A)=\pi_{{a}}(A)$. 
		For proving the reversibility of $M_{DA_{{a}}}$ 
		let $A,B \in \mathcal{B}(G)$ with $A\cap B=\emptyset$.
		We set $c:=1/\pi_0(G_{s, \varrho_{\textmd{app}}}\cap G_{t, \widehat{\varrho}})$ and obtain
		\begin{align*}
			&	\int_{A}^{} M_{DA_{{a}}}(x,B)  \, \pi_{{a}}(\mathrm{d}x) 
			= \int_{A} Q(x, B  \cap  G_{s,\rhoapp}  \cap  G_{t,\widehat{\varrho}})  \, \pi_{{a}}(\mathrm{d} x) \\
			& = c \int_{G} \int_{G} \mathbb{1}_B(y)  \ \mathbb{1}_{G_{s,\rhoapp}\cap G_{t,\widehat{\varrho}}}(y)\ \mathbb{1}_A(x) \mathbb{1}_{G_{s,\varrho_{\textmd{app}}}\cap G_{t,\widehat{\varrho}}}(x)  \,  Q(x,\mathrm{d}y) \, \pi_{0}(\mathrm{d}x) \\
			& = c \int_{G} \int_{G} \mathbb{1}_B(y)  \mathbb{1}_{G_{s, \varrho_{\textmd{app}}}\cap G_{t, \widehat{\varrho}}}(y) \ \mathbb{1}_A(x) \mathbb{1}_{G_{s, \varrho_{\textmd{app}}}\cap G_{t,\widehat{\varrho}}}(x)  \, Q(y,\mathrm{d}x)  \, \pi_{0}(\mathrm{d}y) \\
			& = \int_{B}^{} M_{DA_{{a}}}(y,A)  \, \pi_{{a}}(\mathrm{d}y),
		\end{align*}
		where the last equality follows due to the reversibility of $Q$ w.r.t. $\pi_0$. 
		This implies by standard arguments the claimed $\pi_a$-reversibility of $M_{DA_a}$.
		\item
		\emph{Positivity of $M_{DA_{{a}}}$ on $L^2_{\pi_a}(G)$}: For any $f\in L^2_{\pi_a}(G)$ we show $\langle M_{DA_{{a}}}f, f \rangle_{\pi_{{a}}}\geq 0$.
		With $a=(s,t)\in I$ we have 
		\begin{align*}
			  \langle M_{DA_{{a}}}f, f \rangle_{\pi_{{a}}}
			& = \int_{[x]_{a}} \int_{[x]_{a}} f(y)  \, M_{DA_{{a}}}(x, \mathrm{d}y)  f(x)  \,  \pi_{{a}}(\mathrm{d}x) \\
			& = \int_{[x]_{a}}^{} \int_{[x]_{a}}^{} f(y) \mathbb{1}_{G_{s,\rhoapp}  \cap  G_{t,\widehat{\varrho}}}(y)  \, Q(x,\mathrm{d}y) f(x)  \, \pi_{{a}}(\mathrm{d}x) \\
			& \qquad + \int_{[x]_{a}}^{} (1-Q(x, G_{s,\rhoapp}  \cap  G_{t,\widehat{\varrho}})) f(x)^2 \,  \pi_{{a}}(\mathrm{d}x).
		\end{align*} 
		We can write the first term in the sum as
		\begin{align*}
			&\int_{[x]_{a}}^{} \int_{[x]_{a}}^{} f(y) f(x) \mathbb{1}_{G_{s, \varrho_{\textmd{app}}}\cap G_{t,\widehat{\varrho}}}(y)  \, Q(x,\mathrm{d}y)  \, \pi_{{a}}(\mathrm{d}x) \\
			&\quad = c \int_{G}^{} \int_{G}^{} f(y) f(x) \mathbb{1}_{G_{s, \varrho_{\textmd{app}}}\cap G_{t,\widehat{\varrho}}}(y)\ \mathbb{1}_{G_{s, \varrho_{\textmd{app}}}\cap G_{t,\widehat{\varrho}}}(x)  \, Q(x,\mathrm{d}y)  \, \pi_0(\mathrm{d}x) \\ 
			&\quad= c \, \langle Q\mathbb{1}_{G_{s,\varrho_{\textmd{app}}}\cap G_{t,\widehat{\varrho}}}f, \mathbb{1}_{G_{s, \varrho_{\textmd{app}}}\cap G_{t,\widehat{\varrho}}}f \rangle_{\pi_0}. 
		\end{align*}
		Finally, by the positivity of $Q$ on $L^2_{\pi_0}(G)$ we get
		\begin{align*}
			\langle M_{DA_{{a}}}f, f \rangle_{\pi_{{a}}}  
			& = c \, \langle Q\mathbb{1}_{G_{s,\varrho_{\textmd{app}}}\cap G_{t,\widehat{\varrho}}}f, \mathbb{1}_{G_{s,\varrho_{\textmd{app}}}\cap G_{t,\widehat{\varrho}}}f \rangle_{\pi_0} \\
            & \qquad +\int_{[x]_{a}} (1-Q(x, G_{s,\varrho_{\textmd{app}}}\cap G_{t,\widehat{\varrho}})) f(x)^2  \, \pi_{{a}}(\mathrm{d}x) 
			\geq 0. 
		\end{align*}
		\item \emph{$M_{DA_{{a}}} S_{DA_{{a}}}=S_{DA_{{a}}}$}: We have
		\begin{align*}
			M_{DA_{{a}}} S_{DA_{{a}}}f(x)
			&=\int_{[x]_{a}}^{}\int_{[x]_{a}}^{} f(y)  \,  S_{DA_{{a}}}(z, \mathrm{d}y)  \,  M_{DA_{{a}}}(x, \mathrm{d}z) \\  
			&=\int_{[x]_{a}}^{}\int_{[x]_{a}}^{} f(y)  \,  \pi_a(\mathrm{d}y)  \,  M_{DA_{{a}}}(x, \mathrm{d}z) \\ 
			&=\int_{[x]_{a}} 
			M_{DA_{{a}}}(x, [x]_{a})
			f(y)  \, \pi_a(\mathrm{d}y) \\ 
			&=\int_{[x]_{a}}^{} f(y)  \,  S_{DA_{{a}}}(x, \mathrm{d}y)
			= S_{DA_{{a}}}f(x). 
		\end{align*}
	\end{enumerate}
	As a direct conclusion of Lemma \ref{PMLR:lemma:comparison}, we have shown the statement.  
\end{proof}

\subsection{Proof of Proposition \ref{prop:DAHSS:reversible}}\label{app:DAHSS:reversible}
Let $x\in G$ and $A,B\in\mathcal{B}(G)$. Then with $Z$ defined as in (\ref{PMLR:distribution}) we get 
\begin{align*}
	&\int_{A} H_{{DA}}(x,B)  \, \pi(\mathrm{d}x)
	= \int_{A} \left( \frac{1}{\varrho(x)} \int_{0}^{\widehat{\varrho}(x)} \int_{0}^{\rhoapp(x)} H_{t,s}(x, B)  \, \mathrm{d}s  \, \mathrm{d}t \right) Z^{-1} \varrho(x)  \, \pi_0(\mathrm{d}x) \\
	&= Z^{-1}  \int_{A} \int_{0}^{\infty} \int_{0}^{\infty} H_{t,s}(x, B) \mathbb{1}_{G_{s,\rhoapp}}(x) \mathbb{1}_{G_{t,\widehat{\varrho}}}(x) \frac{\pi_0(G_{s,\rhoapp} \cap  G_{t,\widehat{\varrho}})}{\pi_0(G_{s,\rhoapp} \cap G_{t,\widehat{\varrho}})}
	 \,\mathrm{d}s  \, \mathrm{d}t  \, \pi_0(\mathrm{d}x) \\
	&= Z^{-1}  \int_{0}^{\infty} \int_{0}^{\infty} \pi_0(G_{s,\rhoapp} \cap  G_{t,\widehat{\varrho}}) \int_{A} H_{t,s}(x, B)  \, \pi_{0,t,s} (\mathrm{d}x)  \, \mathrm{d}s  \, \mathrm{d}t  \\
	&= Z^{-1}  \int_{0}^{\infty} \int_{0}^{\infty} \pi_0(G_{s,\rhoapp} \cap  G_{t,\widehat{\varrho}}) \int_{B} H_{t,s}(y, A)  \,  \pi_{0,t,s} (\mathrm{d}y)  \, \mathrm{d}s  \, \mathrm{d}t.
\end{align*}
The last equality follows from the fact that $H_{t,s}$ is $\pi_{0,t,s}$-reversible. Arguing by the same steps (with $x$ taking the role of $y$) backwards gives the desired reversibility, i.e. 
 \begin{align*}
        \int_{A} H_{{DA}}(x,B)  \, \pi(\mathrm{d}x)=\int_{B} H_{{DA}}(y,A)  \, \pi(\mathrm{d}y), \qquad A,B\in\mathcal{B}(G).
    \end{align*}
For the second part of the proof, showing the $\pi$-invariance of $H_{{DA}}$, we can conclude analogously with some steps before that  
\begin{align*}
    \pi H_{DA}(A)&= \int_{G} H_{{DA}}(x,A) \, \pi(\mathrm{d}x) \\
                 &= Z^{-1}  \int_{0}^{\infty} \int_{0}^{\infty} \pi_0(G_{s,\rhoapp} \cap  G_{t,\widehat{\varrho}}) \int_{G} H_{t,s}(x, A)  \,  \pi_{0,t,s} (\mathrm{d}x)  \, \mathrm{d}s  \, \mathrm{d}t \\
                 &=  Z^{-1}  \int_{0}^{\infty} \int_{0}^{\infty} \pi_0(G_{s,\rhoapp} \cap  G_{t,\widehat{\varrho}}) \, \pi_{0,t,s} (A)  \, \mathrm{d}s  \, \mathrm{d}t \\
                 &=  Z^{-1} \int_{0}^{\infty} \int_{0}^{\infty} \pi_0(A \cap G_{s,\rhoapp} \cap  G_{t,\widehat{\varrho}}) \, \mathrm{d}s  \, \mathrm{d}t = \pi(A).    
\end{align*}
Here, the third equality follows from the fact that $H_{t,s}$ is $\pi_{0,t,s}$-invariant and the last equality due to (\ref{eq:pi}) from the proof of Proposition \ref{positive:measure:iff}.  
\begin{flushright}    
$\square$
\end{flushright}

\subsection{Proof of Theorem \ref{theo:DAESS}}
\label{app:DAESS}

The on-slice transition kernels of DA-ESS \change{(Algorithm \ref{alg:DAESS})} for $x\in G$, $A\in\mathcal{B}(G)$, $s \in (0, \rhoapp(x))$ and $t \in (0, \widehat{\varrho}(x))$ take the form 
\[
   E_{t,s}(x,A):=\int_{\mathbb R^d} E_{x,v,t,s}(A) \, \pi_0(\mathrm{d} v)
\]
where $E_{x,v,t,s}$ represents the transition mechanism for moving to a set $A\in\mathcal{B}(G)$ using the delayed acceptance shrinkage routine in Algorithm \ref{alg:DA-Shrinkage} given  $x, v, t$ and $s$. 
\change{Since the delayed acceptance routine accepts a proposal if and only if it belongs to $G_{s, \varrho_{\textmd{app}}}\cap G_{t, \widehat{\varrho}}$ and otherwise applies exactly the same shrinkage update as the plain ESS shrinkage routine, the resulting on-slice transition coincides with the ESS transition on the intersection $G_{s, \varrho_{\textmd{app}}}\cap G_{t, \widehat{\varrho}}$. Hence, applying} the proof in \citet[Theorem~3.1]{HasenpflugEtAl2023} to the delayed acceptance scheme also yields the \change{$\pi_{0,t,s}$}-reversibility of $E_{t,s}$ in that setting \change{and therefore, by Proposition~\ref{prop:DAHSS:reversible}, the desired result.} 

In the following, let $x\in G$ and $A\in\mathcal{B}(G)$ be arbitrary. To prove \change{ergodicity}, we verify that Assumption \ref{statement1}. \change{and \ref{statement2}} of Proposition \ref{prop:DAHSS:ergodic} holds. 

\begin{enumerate}

    \item  
    Let $\pi(A)>0$, then an immediate consequence of Proposition~\ref{positive:measure:iff} is that there exist $S_{x,A}\in(0, \rhoapp(x)]$ and $T_{x,A}\in(0, \widehat{\varrho}(x)]$, such that $\pi_0(A\cap G_{s, \varrho_{\textmd{app}}}\cap G_{t, \widehat{\varrho}}) > 0$ for all $s\in(0, S_{x,A})$ and $t\in(0, T_{x,A})$. To show \change{that} $E_{t,s}(x,A)>0$ for all $s\in(0, S_{x,A})$ and $t\in(0, T_{x,A})$, note that the probability of moving to a set $A\in\mathcal{B}(G)$ after completing all steps in lines 2-5 of Algorithm \ref{alg:DAESS} is greater than or equal to the 
    probability \change{that the first angle proposal is accepted and the resulting state belongs to $A$.}
    Hence, this leads, together with \change{Tonelli’s theorem}, to   
    \[
       E_{t,s}(x,A)\geq \frac{1}{2\pi}\int_{-\pi}^{\pi}\int_{\mathbb R^d} \mathbb{1}_{A\cap G_{s,\varrho_{\textmd{app}}}\cap G_{t,\widehat{\varrho}}}(p(x,v,\theta)) \, \pi_0(\mathrm{d} v) \, \mathrm{d}\theta. 
    \]
    In the following, we show that the right-hand side of the inequality is positive. We define for $\theta_0\in(-\pi,\pi)\setminus\{0\}$ the linear mapping
    \begin{align*}
       T_{x,\theta_0}: \mathbb R^d \rightarrow \mathbb R^d, \quad v \mapsto p(x, v, \theta_0)=\cos(\theta_0) \ x + \sin(\theta_0)\ v.    
    \end{align*}
    The mapping $T_{x,\theta_0}$ is bijective and its preimage is given by 
    \begin{align}\label{inverse:mapping}
       T_{x,\theta_0}^{-1}(A)=\{v\in\mathbb R^d\colon T_{x,\theta_0}(v)\in A\}.
    \end{align}
    Consequently we get for all $\theta_0\in(-\pi,\pi)\setminus\{0\}$, $s\in(0, S_{x,A})$ and $t\in(0, T_{x,A})$   
    \begin{align*}
        \int_{\mathbb R^d} \mathbb{1}_{A\cap G_{s,\varrho_{\textmd{app}}}\cap G_{t,\widehat{\varrho}}}(p(x,v,\theta_0)) \, \pi_0(\mathrm{d} v)&=\pi_0(\{v\in\mathbb R^d\colon p(x, v, \theta_0)\in A\cap G_{s,\varrho_{\textmd{app}}}\cap G_{t,\widehat{\varrho}}\}) \\
        &=\pi_0(\{v\in\mathbb R^d\colon T_{x,\theta_0}(v)\in A\cap G_{s,\varrho_{\textmd{app}}}\cap G_{t,\widehat{\varrho}}\}) \\
        &= \pi_0(T_{x,\theta_0}^{-1}( A\cap G_{s,\varrho_{\textmd{app}}}\cap G_{t,\widehat{\varrho}})) >0,
    \end{align*}
    where the positivity follows from the change-of-variables theorem and the well-known fact, that in finite dimensions, a non-degenerate Gaussian measure is equivalent to the Lebesgue measure. Therefore, we can conclude that for all $s\in(0, S_{x,A})$ and all $t\in(0, T_{x,A})$
    \begin{align*}
        E_{t,s}(x,A)&\geq \frac{1}{2\pi}\int_{(-\pi,\pi)\setminus\{0\}} \int_{\mathbb R^d} \mathbb{1}_{A\cap G_{s,\varrho_{\textmd{app}}}\cap G_{t,\widehat{\varrho}}}(p(x,v,\theta)) \, \pi_0(\mathrm{d} v) \, \mathrm{d}\theta > 0 
    \end{align*}

    \item 
    \change{If $\pi(A)=0$, then also $\pi_0(A)=0$ by (\ref{PMLR:distribution}) and $\lambda_d(A)=0$ since the Gaussian prior measure and the Lebesgue measure are equivalent. The delayed acceptance shrinkage routine terminates when it generates an angle in the set $\Theta_{G_{s,\varrho_{\textmd{app}}}\cap G_{t,\widehat{\varrho}}}(x,v)$ as defined in \eqref{eq:theta_B_def}. To describe those accepted points which additionally belong to $A$, we consider $\Theta_{A\cap G_{s,\varrho_{\textmd{app}}}\cap G_{t,\widehat{\varrho}}}(x,v)$. We first show that $\lambda_1(\Theta_{A\cap G_{s,\varrho_{\textmd{app}}}\cap G_{t,\widehat{\varrho}}}(x,v))=0$ for $\pi_0$-almost every $v\in\mathbb R^d$. Since $A\cap G_{s,\varrho_{\textmd{app}}}\cap G_{t,\widehat{\varrho}}\subseteq A$ we have $\lambda_d(A\cap G_{s,\varrho_{\textmd{app}}}\cap G_{t,\widehat{\varrho}})=0$ and therefore as above
    \[
      \int_{\mathbb R^d} \mathbb{1}_{A\cap G_{s,\varrho_{\textmd{app}}}\cap G_{t,\widehat{\varrho}}}(p(x,v,\theta_0)) \, \pi_0(\mathrm{d} v) = \pi_0(T_{x,\theta_0}^{-1}( A\cap G_{s,\varrho_{\textmd{app}}}\cap G_{t,\widehat{\varrho}})) = 0,
    \]
    with $T_{x,\theta_0}^{-1}$ given as in \eqref{inverse:mapping}, for every $\theta_0\in(-\pi,\pi)\setminus\{0\}$. Notice, that $\{-\pi,0,\pi\}$ has one-dimensional Lebesgue measure zero. This leads together with Tonelli to
    \begin{align*}
    0 
    &=\int_{\mathbb R^d} \int_{-\pi}^{\pi} \mathbb{1}_{A\cap G_{s,\varrho_{\textmd{app}}}\cap G_{t,\widehat{\varrho}}}(p(x,v,\theta)) \, \mathrm{d}\theta  \, \pi_0(\mathrm{d} v)\\
    &= \int_{\mathbb R^d} \lambda_1(\Theta_{A\cap G_{s,\varrho_{\textmd{app}}}\cap G_{t,\widehat{\varrho}}}(x,v)) \, \pi_0(\mathrm{d} v),
    \end{align*}
    which implies $\lambda_1(\Theta_{A\cap G_{s,\varrho_{\textmd{app}}}\cap G_{t,\widehat{\varrho}}}(x,v))=0$ for $\pi_0$-almost every $v\in \mathbb R^d$. Arbitrarily fix from now on such a $v\in \mathbb R^d$. As in \citet{HasenpflugEtAl2023}, consider the unstopped sequence corresponding to the delayed acceptance shrinkage routine in Algorithm~\ref{alg:DA-Shrinkage}, and let $\Theta_1, \Theta_2, \dots$, denote its successive angle proposals. 
    The first proposal angle is generated before any shrinkage takes place and satisfies $\Theta_1\sim \mathrm{U}(-\pi,\pi)$. Hence,
    \begin{align*}
        \mathbb P\left(\Theta_1\in \Theta_{A\cap G_{s,\varrho_{\textmd{app}}}\cap G_{t,\widehat{\varrho}}}(x,v)\right)= \frac{\lambda_1\left(\Theta_{A\cap G_{s,\varrho_{\textmd{app}}}\cap G_{t,\widehat{\varrho}}}(x,v)\right)}{2\pi}=0.
    \end{align*}
    For $k\geq 2$, let $\Theta^{\min}_k$ and $\Theta^{\max}_k$ denote the current shrinkage interval bounds immediately before the $k$-th angle proposal is generated. By construction of Algorithm~\ref{alg:DA-Shrinkage}, conditional on these bounds, the $k$-th proposal is uniformly distributed on $I_k=(\Theta^{\min}_k, \Theta^{\max}_k)$, cf. \citet[Equation~(4)]{HasenpflugEtAl2023}. Therefore, 
    \begin{align*}
        \mathbb{P}(\Theta_k\in \Theta_{A\cap G_{s,\varrho_{\textmd{app}}}\cap G_{t,\widehat{\varrho}}}(x,v)\mid I_k)= \frac{\lambda_1\left(\Theta_{A\cap G_{s,\varrho_{\textmd{app}}}\cap G_{t,\widehat{\varrho}}}(x,v)\cap I_k \right)}{\lambda_1\left(I_k\right)}=0.
    \end{align*}
    By the tower property we get  
    \begin{align*}
       \mathbb P(\Theta_k\in \Theta_{A\cap G_{s,\varrho_{\textmd{app}}}\cap G_{t,\widehat{\varrho}}}(x,v))=\mathbb E\left[\mathbb{P}(\Theta_k\in \Theta_{A\cap G_{s,\varrho_{\textmd{app}}}\cap G_{t,\widehat{\varrho}}}(x,v)\mid I_k)\right]=0.   
    \end{align*}
    Now, define the stopping time by $\tau_{t,s}:=\inf\{k\geq 1: \Theta_k \in \Theta_{G_{s,\varrho_{\textmd{app}}}\cap G_{t,\widehat{\varrho}}}(x,v)\}$, which is the first index for which $p(x,v, \Theta_k)\in G_{s,\varrho_{\textmd{app}}}\cap G_{t,\widehat{\varrho}}$. Since $\rhoapp$ and $\widehat{\varrho}$ are lower semi-continuous, $G_{s,\varrho_{\textmd{app}}}\cap G_{t,\widehat{\varrho}}$ is open and since $\theta\mapsto p(x,v,\theta)$ is continuous, we get that $\Theta_{G_{s,\varrho_{\textmd{app}}}\cap G_{t,\widehat{\varrho}}}(x,v)$ is open on the circle.
    Moreover, the set is non-empty since $0\in\Theta_{G_{s,\rhoapp}\cap G_{t,\widehat{\varrho}}}(x,v)$.
    Thus, by \citet[Lemma~2.7]{HasenpflugEtAl2023} we have $\tau_{t,s}<\infty$ almost surely. Since $\{\tau_{t,s}<\infty\}=\cup^\infty_{k=1}\{\tau_{t,s}=k\}$, where the events are pairwise disjoint, we obtain 
    \begin{align*}
        E_{x,v,t,s}(A)&=\mathbb{P}(\Theta_{\tau_{t,s}}\in \Theta_{A\cap G_{s,\varrho_{\textmd{app}}}\cap G_{t,\widehat{\varrho}}}(x,v), \tau_{t,s}<\infty) \\
        &=\sum_{k=1}^\infty \mathbb{P}(\Theta_{\tau_{t,s}}\in \Theta_{A\cap G_{s,\varrho_{\textmd{app}}}\cap G_{t,\widehat{\varrho}}}(x,v), \tau_{t,s}=k) \\
        &=\sum_{k=1}^\infty \mathbb{P}(\Theta_{k}\in \Theta_{A\cap G_{s,\varrho_{\textmd{app}}}\cap G_{t,\widehat{\varrho}}}(x,v), \tau_{t,s}=k) \\
        &\leq \sum_{k=1}^\infty \mathbb P(\Theta_k\in \Theta_{A\cap G_{s,\varrho_{\textmd{app}}}\cap G_{t,\widehat{\varrho}}}(x,v)) = 0.
    \end{align*}
    Hence, 
    \[
    E_{t,s}(x,A):=\int_{\mathbb R^d} E_{x,v,t,s}(A) \, \pi_0(\mathrm{d} v)=0
    \]
    for all $t\in(0,\widehat{\varrho}(x))$ and all $s\in(0, \rhoapp(x))$.
    }
    
\end{enumerate}
    
\begin{flushright}    
$\square$
\end{flushright}

\subsection{Proof of Theorem \ref{theo:DAHRUSS}}
\label{app:DAHRUSS}

The on-slice transition kernels of DA-HRUSS for $x\in G$, $A\in\mathcal{B}(G)$, $s \in (0, \rhoapp(x))$ and $t \in (0, \widehat{\varrho}(x))$, \change{underlying Algorithm \ref{alg:DAHRSS} under Assumption \ref{assum:bimodal} and \eqref{eq:line_section_def}}, take the form 
\begin{align*}
    R_{t,s}(x,A)&=\int_{\mathbb{S}_{d-1}} \change{\Biggl[}  \xi_{t,s}(x,v) \, \frac{\lambda_1(L_{A \cap G_{s,\varrho_{\textmd{app}}}\cap G_{t,\widehat{\varrho}}}(x,v))}  {\lambda_1(L_{G_{s,\varrho_{\textmd{app}}}\cap G_{t,\widehat{\varrho}}}(x,v))} \\
    &\qquad+ (1-\xi_{t,s}(x,v)) \sum_{i=1}^2 \mathbb{1}_{G_{s,\varrho_{\textmd{app}}}\cap G^{(i)}_{t,\widehat{\varrho}}}(x) \,  \frac{\lambda_1(L_{A \cap G_{s,\varrho_{\textmd{app}}}\cap G^{(i)}_{t,\widehat{\varrho}}}(x,v))}  {\lambda_1(L_{G_{s,\varrho_{\textmd{app}}}\cap G^{(i)}_{t,\widehat{\varrho}}}(x,v))} \change{\Biggr]}  \, \frac{\sigma_d(\mathrm{d}v)}{\kappa_d} 
\end{align*}
\change{where whenever $\widehat{\varrho}$ is bimodal, $t\in[t_1,t_2)$, and both line sections
$L_{G_{s,\varrho_{\mathrm{app}}}\cap G^{(1)}_{t,\widehat{\varrho}}}(x,v)$ and $L_{G_{s,\varrho_{\mathrm{app}}}\cap G^{(2)}_{t,\widehat{\varrho}}}(x,v)$
are nonempty}
\begin{align*}
     \xi_{t,s}(x,v):=\frac{\lambda_1(L_{G_{s,\varrho_{\textmd{app}}}\cap G_{t,\widehat{\varrho}}}(x,v))}{\lambda_1(L_{G_{s,\varrho_{\textmd{app}}}\cap G_{t,\widehat{\varrho}}}(x,v))+\delta_{t,s}(x,v)} 
\end{align*}
\change{with}
\begin{align*}
      \delta_{t,s}(x,v)
  =\inf\bigl\{|z_1 - z_2|\;:\;
    z_1\in L_{G_{s,\varrho_{\mathrm{app}}}\cap G^{(1)}_{t,\widehat{\varrho}}}(x,v),\;
    z_2\in L_{G_{s,\varrho_{\mathrm{app}}}\cap G^{(2)}_{t,\widehat{\varrho}}}(x,v)
  \bigr\},
\end{align*}
\change{
and in all remaining cases, i.e., if $\widehat{\varrho}$ is unimodal, if $t\notin[t_1,t_2)$, or if at least one of the two line sections is empty, we set}
\[
    \xi_{t,s}(x,v):=1.
\]

\noindent \change{The $\pi_{0,t,s}$-reversibility of $R_{t,s}$ is a consequence of \cite[Lemma~4.10]{Rudolf2012} and of  \cite[Lemma~9]{LatuszynskiRudolf2024}, both viewed in the context of the delayed acceptance framework. Consequently by Proposition~\ref{prop:DAHSS:reversible}, the transition kernel of DA-HRUSS is reversible w.r.t. $\pi$.}  
In the following, let $x\in G$ and $A\in\mathcal{B}(G)$ be arbitrary. We \change{verify} for ergodicity Assumption \ref{statement1}. and \ref{statement2}. of Proposition \ref{prop:DAHSS:ergodic}. 
\begin{enumerate}

     \item If $\pi(A)>0$, an immediate consequence of Proposition~\ref{positive:measure:iff} is that there exist some $S_{x,A}\in(0, \rhoapp(x)]$ and $T_{x,A}\in(0, \widehat{\varrho}(x)]$, such that $\pi_0(A\cap G_{s, \varrho_{\textmd{app}}}\cap G_{t, \widehat{\varrho}}) > 0$ for all $s\in(0, S_{x,A})$ and $t\in(0, T_{x,A})$. Following the approach in the proof of \cite[Lemma~11]{LatuszynskiRudolf2024} including using again \cite[Theorem~15.13]{Schilling2005}, we receive for (almost) all $s\in(0, S_{x,A})$ and (almost) all $t\in(0, T_{x,A})$
     \begin{align*}
      R_{t,s}(x,A)&\geq \int_{\mathbb{S}_{d-1}} \xi_{t,s}(x,v) \, \frac{\lambda_1(L_{A \cap G_{s,\varrho_{\textmd{app}}}\cap G_{t,\widehat{\varrho}}}(x,v))}  {\lambda_1(L_{G_{s,\varrho_{\textmd{app}}}\cap G_{t,\widehat{\varrho}}}(x,v))}  \, \frac{\sigma_d(\mathrm{d}v)}{\kappa_d} \\ 
      & = \frac{2}{\kappa_d} \int_{A\cap G_{s,\varrho_{\textmd{app}}}\cap G_{t,\widehat{\varrho}}} \frac{\xi_{t,s}(x,\frac{x-y}{\|x-y\|})}{\|x-y\|^{d-1} \, \lambda_1(L_{G_{s,\varrho_{\textmd{app}}}\cap G_{t,\widehat{\varrho}}}(x,\frac{x-y}{\|x-y\|}))}  \, \pi_0(\mathrm{d}y) \\
      & \change{= \frac{2}{\kappa_d} \int_{A\cap G_{s,\varrho_{\textmd{app}}}\cap G_{t,\widehat{\varrho}}} \frac{\pi_0(\mathrm{d}y)}{\|x-y\|^{d-1} \, (\lambda_1(L_{G_{s,\varrho_{\textmd{app}}}\cap G_{t,\widehat{\varrho}}}(x,\frac{x-y}{\|x-y\|}))+\delta_{t,s}(x, \frac{x-y}{\|x-y\|}))}} \\
      &\geq\frac{\change{2 \cdot }\ \pi_0(A\cap G_{s,\varrho_{\textmd{app}}}\cap G_{t,\widehat{\varrho}})}{\kappa_d\cdot \underset{w_1,w_2\in G_{s,\varrho_{\textmd{app}}}}{\sup}\| w_1-w_2 \|^d}>0,
 \end{align*}
where the last inequality is due to 
\[
   \lambda_1\left(L_{G_{s,\varrho_{\textmd{app}}}\cap G_{t,\widehat{\varrho}}}\left(x,\change{\frac{x-y}{\|x-y\|}}\right)\right)+\delta_{t,s}\left(x,\change{\frac{x-y}{\|x-y\|}}\right)\leq \underset{w_1,w_2\in G_{s,\varrho_{\textmd{app}}}}{\sup}\| w_1-w_2 \|. 
\]
\change{Furthermore, since $G_{s,\rhoapp}$ is bounded for all $s\in(0,\|\rhoapp\|_\infty)$ we have
\[
   \underset{w_1,w_2\in G_{s,\varrho_{\textmd{app}}}}{\sup}\| w_1-w_2 \|^d<\infty.
\]
}

    \item If $\pi(A)=0$, then also $\pi_0(A)=0$ by (\ref{PMLR:distribution}). Applying the formula in \citet[Theorem~15.13]{Schilling2005} yields
     \begin{align*}
       \int_{\mathbb{S}_{d-1}} \lambda_1(L_A(x,v)) \, \sigma_d(\mathrm{d}v)  
       & = \int_{\mathbb{S}_{d-1}} \int_{-\infty}^\infty \mathbb{1}_A(x+pv) \, \mathrm{d}p \, \sigma_d(\mathrm{d}v) \\
       & = 2 \int_A \frac{1}{\| x-y \|^{d-1}}\pi_0(\mathrm{d}y)=0.   
     \end{align*}
     Hence, we obtain $\lambda_1(L_A(x,v))=0$ for $\sigma_d$-almost every $v\in\mathbb{S}_{d-1}$, which immediately implies $R_{t,s}(x,A)=0$ for all $t\in(0,\widehat{\varrho}(x))$ and all $s\in(0, \rhoapp(x))$ \change{since $L_{A \cap G_{s,\varrho_{\textmd{app}}}\cap G_{t,\widehat{\varrho}}}(x,v)\subseteq L_A(x,v)$ and $L_{A \cap G_{s,\varrho_{\textmd{app}}}\cap G^{(i)}_{t,\widehat{\varrho}}}(x,v))\subseteq L_A(x,v)$ for $i=1,2$.}
     
\end{enumerate}

\begin{flushright}    
$\square$
\end{flushright}

\subsection{Proof of Theorem \ref{theo:DAGPSS}}
\label{app:DAGPSS}

The on-slice transition kernels of DA-GPSS for $x= r_0 \, \nu_0\in G$, $A\in\mathcal{B}(G)$, $s \in (0, \rhoapp(x))$ and $t \in (0, \widehat{\varrho}(x))$, \change{underlying Algorithm \ref{alg:DAGPSS}}, take the form $B_{t,s}(r_0 \, \nu_0,A)=B^{(D)}_{t,s}B^{(R)}_{t,s}(r_0 \, \nu_0,A)$ and are the compositions of the transition kernels of the direction update $B^{(D)}_{t,s}$ and the transition kernels of the radius update $B^{(R)}_{t,s}$, \change{where the latter is given due to Assumption \ref{assum:radius:update} by}


\[
   B^{(R)}_{t,s}(r_0\, v, A)=\frac{\lambda_1(L_{A\cap G_{s,\varrho_{\textmd{app}}}\cap G_{t,\widehat{\varrho}}}(v))}{\lambda_1(L_{G_{s,\varrho_{\textmd{app}}}\cap G_{t,\widehat{\varrho}}}(v))}  
\]
\change{with $L_{\cdot}(v)$ as defined in \eqref{segment:gpss}}.

\change{The arguments of \citet[Lemma~4.6, Theorem~4.13]{Schaer2025} can now be extended to the delayed acceptance version to establish the $\pi_{0,t,s}$-reversibility of the kernels $B^{(R)}_{t,s}$ and $B^{(D)}_{t,s}$. The $\pi_{0,t,s}$-invariance of $B_{t,s}$ follows then directly from}
\[
  \change{\pi_{0,t,s}B_{t,s}=\pi_{0,t,s}B^{(D)}_{t,s}B^{(R)}_{t,s}= \pi_{0,t,s}B^{(R)}_{t,s}=\pi_{0,t,s}.}
\]
\change{Hence, by Proposition~\ref{prop:DAHSS:reversible} the transition kernel of DA-GPSS is invariant w.r.t. $\pi$.} 

\change{Next, we verify the assumptions of Proposition \ref{prop:DAHSS:ergodic} to show ergodicity. To this end, let again $x= r_0 \, \nu_0\in G$, $s \in (0, \rhoapp(x))$ and $t \in (0, \widehat{\varrho}(x))$ be arbitrary. 
We follow similar steps as in \cite{SchaerEtAl2023} and obtain by Equation (4.20) in the proof of \cite[Theorem~4.14]{Schaer2025}}  
    \[
       B^{(D)}_{t,s}(r_0\, v_0, \mathrm{d}z) \geq \varepsilon \, \int_{\mathbb{S}_{d-1}} \mathbb{1}_{G_{s,\varrho_{\textmd{app}}}\cap G_{t,\widehat{\varrho}}}(r_0\, v) \, \delta_{r_0 \, v}(\mathrm{d}z) \, \sigma_d(\mathrm{d}v)
    \]
    for a $\varepsilon>0$, see the proof of \cite[Theorem~4.14]{Schaer2025} for the explicit definition of $\varepsilon$. Following the proof of \cite[Theorem~3.2]{SchaerEtAl2023} we analogously have
    \begin{align*}
        B_{t,s}(r_0\, v_0, A)&= \int_{ G_{s, \varrho_{\textmd{app}}}\cap G_{t, \widehat{\varrho}}} B^{(R)}_{t,s}(z, A) \, B^{(D)}_\change{{t,s}}(r_0\, v_0, \mathrm{d}z) \\
                           &\geq \varepsilon \, \int_{\mathbb{S}_{d-1}} \int_{ G_{s, \varrho_{\textmd{app}}}\cap G_{t, \widehat{\varrho}}}  B^{(R)}_{t,s}(z, A) \, \mathbb{1}_{G_{s,\varrho_{\textmd{app}}}\cap G_{t,\widehat{\varrho}}}(r_0\, v) \, \delta_{r_0 \, v}(\mathrm{d}z) \, \sigma_d(\mathrm{d}v) \\
                           &= \varepsilon \, \int_{\mathbb{S}_{d-1}} B^{(R)}_{t,s}(r_0 \, v, A) \, \mathbb{1}_{G_{s,\varrho_{\textmd{app}}}\cap G_{t,\widehat{\varrho}}}(r_0\, v) \, \sigma_d(\mathrm{d}v) \\ 
                           &= \varepsilon \, \int_{\mathbb{S}_{d-1}} \int_{0}^\infty  \frac{\mathbb{1}_{A \cap G_{s,\varrho_{\textmd{app}}}\cap G_{t,\widehat{\varrho}}}(r\,v)}{\lambda_1(L_{G_{s,\varrho_{\textmd{app}}}\cap G_{t,\widehat{\varrho}}}(v))} \, \mathbb{1}_{G_{s,\varrho_{\textmd{app}}}\cap G_{t,\widehat{\varrho}}}(r_0\,v) \, \mathrm{d}r \, \sigma_d(\mathrm{d}v) \\
                           &= \varepsilon \, \int_{A \cap G_{s,\varrho_{\textmd{app}}}\cap G_{t,\widehat{\varrho}}} \frac{\mathbb{1}_{G_{s,\varrho_{\textmd{app}}}\cap G_{t,\widehat{\varrho}}}(r_0\frac{y}{\|y\|})}{\|y\|^{d-1} \, \lambda_1(L_{G_{s,\varrho_{\textmd{app}}}\cap G_{t,\widehat{\varrho}}}(\frac{y}{\|y\|}))} \, \mathrm{d}y,  
    \end{align*}
    where the last equality again results from \cite[Theorem~15.13]{Schilling2005}. 
    
    \change{Now, consider $A\in\mathcal B(G)$ with $\pi(A)>0$.} Then, due to Proposition~\ref{positive:measure:iff} there exist $S_{x,A}\in(0, \rhoapp(x)]$ and $T_{x,A}\in(0, \widehat{\varrho}(x)]$, such that $\pi_0(A\cap G_{s, \varrho_{\textmd{app}}}\cap G_{t, \widehat{\varrho}}) > 0$ for all $(s,t)\in(0, S_{x,A})\times(0, T_{x,A})$. 
    Let us define
    \begin{align*}
        F_\change{{s,t}}
        :=
        \left\{y \in A \cap G_{s,\varrho_{\textmd{app}}}\cap G_{t,\widehat{\varrho}} \ \colon \ r_0\frac{y}{\|y\|} \in G_{s,\varrho_{\textmd{app}}}\cap G_{t,\widehat{\varrho}} 
        \right\} 
    \end{align*}
    \change{for which we have, since $\pi(A)>0$, that}
    \begin{align*}
        \int^{T_{x,A}}_0 \int^{S_{x,A}}_0 \lambda_d(F_\change{{s,t}}) \, \mathrm{d}s \, \mathrm{d}t = &\int_A \left(\int^{S_{x,A}}_0 \mathbb{1}_{G_{s,\varrho_{\textmd{app}}}}(y) \, \mathbb{1}_{G_{s,\varrho_{\textmd{app}}}}\left(r_0\frac{y}{\|y\|}\right) \, \mathrm{d}s\right) \\ 
        & \qquad \left(\int^{T_{x,A}}_0 \mathbb{1}_{G_{t, \widehat{\varrho}}}(y) \, \mathbb{1}_{G_{t, \widehat{\varrho}}}\left(r_0\frac{y}{\|y\|}\right) \, \mathrm{d}t\right) \, \mathrm{d}y >0,
    \end{align*}
    where the positivity for the second integral (third integral analogously) follows by 
    \begin{align*}
      \int^{S_{x,A}}_0 \mathbb{1}_{G_{s,\varrho_{\textmd{app}}}}(y) \, \mathbb{1}_{G_{s,\varrho_{\textmd{app}}}}\left(r_0\frac{y}{\|y\|}\right) \, \mathrm{d}s = \min\left(S_{x,A}, \rhoapp(y), \rhoapp\left(r_0\frac{y}{\|y\|}\right)\right) >0.     
    \end{align*}
    \change{Hence, there exists $(\tilde{s},\tilde{t})\in(0,S_{x,A})\times(0,T_{x,A})$ such that
    $\lambda_d(F_{\tilde{s},\tilde{t}})>0$. 
    Since the superlevel sets of $\rhoapp$ and $\rhohat$ are nested, we have $F_{\tilde{s},\tilde{t}}\subseteq F_{s,t}$ whenever $s\leq \tilde{s}$ and $t\leq \tilde{t}$. Thus, by redefining $S_{x,A} := \tilde s$ and $T_{x,A}:= \tilde t$, we obtain $\lambda_d(F_{s,t})>0$ for all
    $(s,t)\in(0,S_{x,A})\times(0,T_{x,A})$.} Furthermore by using \cite[Lemma~\change{B}.4]{SchaerEtAl2023} in our setting, we can conclude 
    \[
    \lambda_1\left(L_{G_{s,\varrho_{\textmd{app}}}\cap G_{t,\widehat{\varrho}}}\left(\frac{y}{\|y\|}\right)\right)<\infty
    \] 
    for $\lambda_1\otimes\lambda_1\otimes \lambda_d$-almost all $(s,t, y)\in(0, S_{x,A})\times (0, T_{x,A}) \times A \cap G_{s,\varrho_{\textmd{app}}}\cap G_{t,\widehat{\varrho}}$.
    Therefore, we receive for (almost) all $s\in(0, S_{x,A})$ and (almost) all $t\in(0, T_{x,A})$  
    \[
    B_{t,s}(r_0\, v_0, A)\geq \varepsilon \, \int_{A \cap G_{s,\varrho_{\textmd{app}}}\cap G_{t,\widehat{\varrho}}} \frac{\mathbb{1}_{G_{s,\varrho_{\textmd{app}}}\cap G_{t,\widehat{\varrho}}}(r_0\frac{y}{\|y\|})}{\|y\|^{d-1} \, \lambda_1(L_{G_{s,\varrho_{\textmd{app}}}\cap G_{t,\widehat{\varrho}}}(\frac{y}{\|y\|}))} \, \mathrm{d}y >0 
    \]
    \change{which verifies the first assumption of Proposition \ref{prop:DAHSS:ergodic}.}

\change{
If $\pi(A)=0$, then also $\pi_0(A)=0$ by \eqref{PMLR:distribution}. By \cite[Theorem~15.13]{Schilling2005} we have
\begin{align*}
    0 = \pi_0(A) = \int_{\mathbb{S}_{d-1}} \lambda_1(L_A(v)) \, \sigma_d(\mathrm{d}v)
\end{align*}
and therefore $\lambda_1(L_{A\cap G_{s,\varrho_{\textmd{app}}}\cap G_{t,\widehat{\varrho}}}(v))=0$ for $\sigma_d$-almost all $v$. This implies $B^{(R)}_{t,s}(r_0\, v, A)=0$ for $\sigma_d$-almost all $v$. 
Analogously to the direction update of GPSS in \cite[Theorem~4.11]{Schaer2025}, the direction kernel of DA-GPSS $B^{(D)}_{t,s}(r_0\, v_0,\cdot)$ is absolutely continuous with respect to $\sigma_d$ by the same arguments as for DA-ESS in the proof of Theorem \ref{theo:DAESS}. 
Hence, we obtain 
\begin{align*}
    B_{t,s}(r_0\, v_0, A)= \int_{ G_{s, \varrho_{\textmd{app}}}\cap G_{t, \widehat{\varrho}}} B^{(R)}_{t,s}(z, A) \, B^{(D)}_{t,s}(r_0\, v_0, \mathrm{d}z)=0
\end{align*}
for all $t\in(0,\widehat{\varrho}(x))$ and all $s\in(0, \rhoapp(x))$ which verifies the second assumption of Proposition \ref{prop:DAHSS:ergodic} and, thus, completes the proof. \hfill \qed
}


\end{appendices}

\bibliography{sn-bibliography}


\end{document}